%% file: ref.tex
\documentclass[a4paper,12pt]{article}
\newif\ifREVversion
\newif\ifDIFFversion
\ifdefined\externalcmd
  \DIFFversiontrue
\else
  \REVversiontrue
\fi
\usepackage[english]{babel}
\usepackage{amsmath, amsfonts, amssymb, amsthm}
\usepackage{multirow}
\usepackage{graphicx}
\usepackage{fancyhdr}
\usepackage[a-1b]{pdfx}
\usepackage{hyperref}
\usepackage{float}
\usepackage{url}
\usepackage{pdfpages}
\usepackage[utf8]{inputenc}
\usepackage[top=1in, bottom=1.25in, left=1in, right=1 in]{geometry}
\usepackage{amssymb, amsmath}
\usepackage{bbm} % to define \One
\usepackage{xcolor}
\usepackage{graphicx}
\usepackage{adjustbox}
\input{commands.tex}
\allowdisplaybreaks
\providecommand{\keywords}[1]
{
  \small	
  \textbf{\textit{Keywords---}} #1
}
\title{Before and after default: information and optimal portfolio via anticipating calculus}
\author{Jos\'{e} A. Salmer\'{o}n
\thanks{Department of Statistics, University Carlos III of Madrid,  Av. de la Universidad, 30, 28911, Legan\'{e}s, Spain.
  \email{joseantonio.salmeron@uc3m.es}.}
\and Giulia Di Nunno
\thanks{Department of Mathematics, University of Oslo, P.O. Box 1053 Blindern, N-0316 Oslo, Norway and Norwegian School of Economics (NHH), Helleveien 30, N-5045 Bergen, Norway.
    \email{giulian@math.uio.no}.}
   \and Bernardo D'Auria 
    \thanks{Department of Mathematics “Tullio Levi Civita”, University of Padova, Via Trieste, 63, 35131 Padova PD, Italy.
   \email{bernardo.dauria@unipd.it}.}
}

\numberwithin{equation}{section}
\usepackage{color}
\usepackage{amsmath,amsfonts,mathtools}
\usepackage{comment}
\usepackage{multirow} % used in tables
\usepackage{enumitem} % used in enumerate environments
\usepackage{standalone} % used to include independent standalone tex files
\usepackage{etoolbox} % used for recurring commands in tikz pictures
\usepackage{pgfplots} % used for plotting graphs
\usepackage{xspace} % to append a space to customized commands phase
\newtheorem{Theorem}{Theorem}[section]
\newtheorem{Assumption}[Theorem]{Assumption}
\newtheorem{Example}[Theorem]{Example}
\newtheorem{Proposition}[Theorem]{Proposition}
\newtheorem{Definition}[Theorem]{Definition}
\newtheorem{Corollary}[Theorem]{Corollary}
\newtheorem{Lemma}[Theorem]{Lemma}
\newtheorem{Remark}[Theorem]{Remark}
\numberwithin{Theorem}{section}

\providecommand{\citep}[1]{(\cite{#1})}
\newcommand{\email}[1]{\href{mailto:#1}{#1}}

\pgfplotsset{compat=1.16}

\ifDIFFversion
\input{commands-latexdiff.tex}
\fi

\begin{document}
%\includepdf[pages={1,2}]{cover-page.pdf}
\maketitle
\begin{abstract}
\ifDIFFversion
\input{abstract-diff.tex}
\else\ifREVversion
\input{abstract-rev.tex}
\else
\input{abstract.tex}
\fi\fi
\end{abstract}


\maketitle

\ifDIFFversion
\input{content-diff.tex}
\else\ifREVversion
\input{content-rev.tex}
\else
\input{content.tex}
\fi\fi

\section*{Funding}
The research leading to these results is within the project STORM (Stochastic for Time-Space Risk Models) funded by the Research Council of Norway (RCN), Project number 274410. 
The first and the third research was partially supported by the Spanish \emph{Ministerio de Economía y Competitividad} grant
%s MTM2017-85618-P (via FEDER funds) and 
PID2020-116694GB-I00.
The first author acknowledges financial support by an FPU Grant (FPU18/01101) of Spanish \emph{Ministerio de Universidades} and the \emph{``Programa Propio de Investigación para la Movilidad de Investigadores"} programme for funding a 3-month research stay at the University of Oslo.

%% BIBLIOGRAPHY
%\printbibliography
\bibliography{ref.bib} 

\appendix
\addcontentsline{toc}{chapter}{Appendix}
\input{short-appendix-BM.tex}
\input{short-appendix-Pois.tex}
\end{document}

%% file: commands.tex
\newcommand{\cA}{\mathcal{A}}
\newcommand{\cB}{\mathcal{B}}

\newcommand{\cE}{\mathcal{E}}
\newcommand{\cF}{\mathcal{F}}
\newcommand{\cG}{\mathcal{G}}
\newcommand{\cH}{\mathcal{H}}
\newcommand{\cJ}{\mathcal{J}}
\newcommand{\cN}{\mathcal{N}}
\newcommand{\cS}{\mathcal{S}}

\newcommand{\bD}{\mathbb{D}}
\newcommand{\bE}{\mathbb{E}}
\newcommand{\bF}{\mathbb{F}}
\newcommand{\bG}{\mathbb{G}}
\newcommand{\bS}{\mathbb{S}}

\newcommand{\bOne}{\mathbbm{1}}

\newcommand{\bN}{\mathbb{N}}
\newcommand{\bR}{\mathbb{R}}

\newcommand{\bV}{\mathbb{V}}
\newcommand{\bW}{\mathbb{W}}
\newcommand{\bZ}{\mathbb{Z}}
\newcommand{\PP}{\boldsymbol{P}}
\newcommand{\QQ}{\boldsymbol{Q}}
\newcommand{\EE}{\boldsymbol{E}}

\ifpdf
  \DeclareGraphicsExtensions{.eps,.pdf,.png,.jpg}
\else
  \DeclareGraphicsExtensions{.eps}
\fi

\numberwithin{equation}{section}

\DeclareMathOperator*{\argmax}{arg\,max}

\DeclareMathOperator*{\esssup}{ess\,sup}

\providecommand{\norm}[1]{\lVert#1\rVert}
\providecommand{\abs}[1]{\lvert#1\rvert}
\providecommand{\citep}[1]{(\cite{#1})}

\newcommand{\prT}[1]{(#1_t, \, 0 \leq{t} \leq{T})}

\newcommand{\Ito}{It\^{o}\xspace}

%% file: abstract-rev.tex
Default risk calculus plays a crucial role in portfolio optimization when the risky asset is under threat of bankruptcy.
However, traditional stochastic control techniques are not applicable in this scenario, and additional assumptions are required to obtain the optimal solution in a before-and-after default context. 
We propose an alternative approach using forward integration, which allows to avoid one of the restrictive assumptions, the \emph{Jacod density hypothesis}.
We demonstrate that, in the case of logarithmic utility, the weaker \emph{intensity hypothesis} is the appropriate condition for optimality.
Furthermore, we establish the semimartingale decomposition of the risky asset in the filtration that is progressively enlarged to accommodate the default process, under the assumption of the  existence of the optimal portfolio.
This work aims to provide valueable insights for developing effective risk management strategies when facing default risk.

%%%%%%%%%%%%%%%%%%%%%%%%%%%
%%% The previous abstract :
%%%%%%%%%%%%%%%%%%%%%%%%%%%

    %We work within the framework of the classical optimal portfolio problem using the forward integration.
    %with some modifications.
    %We consider an incomplete market with two assets in which the agent could invest in: a bond and a risky defaultable asset. For modelling the evolution of the risky asset, we include noises both from a Brownian motion and a compensated Poisson random measure.
    %Moreover, we add a default process associated to a random time $\tau$.
    %In this paper we have three main goals:
    %the first one is trying to extend the paper~\cite{diNunno2014} until time $T$ also under the set $\{\tau < T\}$.
    %The second one is to expand the approach given by~\cite{OksendalBiagini2005,DiNunnoOksendalProske2006} to a defaultable portfolio problem.
    %The third one is generalizing some previous works in default risk analysis with weaker assumptions, in particular, we refer the papers~\cite{ElKarouiJeanblancJiao2009,JiaoKharroubiPham2013}. The main assumption that we want to avoid is known in the literature as the Jacod's or density hypothesis.
    %Moreover, to achieve the optimization solution, we  use the approach before and after that appears in \cite{JiaoPham2011}.
    

%% file: content-rev.tex
\section{Introduction}
In this paper, we address the problem of constructing the optimal portfolio consisting of a riskless bond and a defaultable, risky asset. 
The dynamics of the risky asset are modeled by a Brownian motion $W=\prT{W}$, a compensated Poisson random measure $\tilde N =\prT{\tilde N}$ and a default process 
$H = (\bOne_{\{\tau\leq t\}},\,0\leq t \leq T)$, 
where~$\tau$ is a random time.
We establish necessary and sufficient conditions for the existence of a local maximum portfolio in our market and provide a semimartingale decomposition of the processes~$W$ and~$\tilde N$ in the progressive enlarged filtration~$\bG := \bF\vee\sigma(\tau\wedge t)$ 
needed to account for the default time. 
Here $\bF$ denotes the minimum  complete right-continuous filtration generated by~$W$ and~$\tilde N$, which will be introduced in next section.

Previous works have approached the optimal portfolio problem under various assumptions.
In~\cite{merton1969}, standard stochastic control techniques were used under CRRA utility, assuming that the dynamics of the risky asset were modeled by a geometric Brownian motion.
A general solution was presented in \cite{kramkovSchachemayer1999}, for both a complete and incomplete market, with the risky asset assumed to be a semimartingale.
In a defaultable framework, an explicit solution was derived in \cite{Amendinger2003} for different utilities under initial enlargement of filtrations, assuming the so called \emph{Jacod density hypothesis}.
In~\cite{BELLAMY2001}, the optimal portfolio problem was solved under the assumption of discontinuous dynamics, modeled by a Lévy process. 
In~\cite{diNunno2014} a default process was introduced to model bankruptcy, and the trading was assumed to stop after the default. 
The knowledge of the default was not included in the agent's natural filtration, and the random time associated to the default process was assumed to satisfy the {density hypothesis}.
In~\cite{OksendalBiagini2005}, a different approach was taken. By assuming the existence of a local maximum solution, a semimartingale decomposition of the Brownian motion driving the risky asset was provided in the enlarged filtration.
%They managed to drop the {density hypothesis}, which is not trivial to verify.
This work was then extended in \cite{DiNunnoOksendalProske2006} by introducing jumps in the dynamics of the risky asset modeled with a compensated Poisson component.

Over the past twenty-five years, default risk analysis has been a subject of extensive study.
In~\cite{Merton74}, the case where the default time is a stopping time in the natural filtration of the price process is examined, while in~\cite{CooperMartin96}, a comparison is made with the case where the default time is a stopping time only in the larger filtration~$\bG$.
In~\cite{DuffieLando99}, the semimartingale decomposition of the process~$H$ in the  filtration~$\bG$ is explicitly computed, assuming the default time to be a specific hitting time.
%Some additional results about default risk are given in~\cite{ElliotJeanblancYor2000} and~\cite{JeanblancGapeev20}. 
%In the latter, the default is studied with respect to Swaps derivatives.\\
Currently, default risk analysis in the context of optimal control problems is an active area of research. A recent relevant work is, for instance, \cite{DiTella2020PRP}, in which the optimal portfolio problem is analyzed under the immersion hypothesis. This hypothesis assumes that $W$ and~$\tilde N$ are~$\bG$-martingales with respect to the progressive enlargement $\bG\supset\bF$. 
Subsequently, in \cite{diTella2021}, this assumption is relaxed by employing the weaker density hypothesis.
In our work, we aim to further relax the assumptions and investigate what happens when this hypothesis is not satisfied.

This paper aims to provide necessary and sufficient conditions for the existence of an optimal portfolio in a defaultable framework, without assuming the {density hypothesis}.
By incorporating the universal framework of the forward integration, that allows to include different information flows at the same time, we extend existing literature by providing a tool that allows to analyze situations where the default time distribution is singular with respect to the Lebesgue measure. For example, this framework enables us to study cases such as the default time being given by the maximum of the Brownian motion before the time horizon. 
Specifically, we show that under the logarithmic utility assumption, the existence of a local maximum for the before-default strategy is strictly related to the \emph{intensity hypothesis}, as detailed in Theorem~\ref{theo.opt.log.bef} and Examples~\ref{ex.argmax} and~\ref{ex.aksamit.half-final}.
Our proposed solution involves using the anticipating calculus in a non-standard way, where  the domain of the Malliavin derivative $D_t$ is extended to the bigger  Hida-Malliavin distribution space, $(\cS)^*$, allowing for $D_\cdot F\not\in L^2(dt\times\PP)$.
Subsection~\ref{subsec:white.noise} provides details on the anticipating calculus.
In the second part of the paper, we prove that, under a general utility function and
assuming the existence of an optimal portfolio, the Brownian motion and the compensated version of the Poisson random measure driving the risky asset are indeed semimartingales in the enlarged filtration $\bG$. 
%We show our computations for the Hunt processes, obtaining more explicit expressions.

To better frame our results in the context of the existing literature, we consider our framework similar to the one in \cite{diTella2021}, however we do not assume the density hypothesis.
We expand the trading time beyond the occurrence of the bankruptcy, so extending the work in~\cite{diNunno2014}. Thus, we are able to prove  the existence of a measure $\QQ$ under which the processes~$W$ and~$\tilde N$ are semimartingales in the enlarged filtration~$\bG$.
In comparison with~\cite{DiNunnoOksendalProske2006}, our model is more general as it includes the default process and considers a general utility function.
%Moreover, we extend the result in~\cite{OksendalBiagini2005}, which only consider purely continuous models.
%With respect to~\cite{OksendalBiagini2005}, our model is also more general as it includes a Poisson compensated random measure. 
%In particular, we are extending the chapter 16 in \cite{DiNunnoOksendalProske2009} with some of the ideas of chapter 9 in \cite{OksendalSulemBook2019}.
Moreover, unlike previous literature such as~\cite{OksendalBiagini2005,ElKarouiJeanblancJiao2009, JiaoKharroubiPham2013,JiaoPham2011}, %Kharroubi2012}, %CoculescuJeanblanc2012, 
which assumes the density hypothesis associated to the default in the context of stochastic control problems, this paper is the first to introduce a progressive enlargement of filtrations without this assumption.
We extend the relationship between the Malliavin trace, the Skorohod integral and the forward integration to the space~$(\cS)^*$, by allowing a progressive enlargement with respect to a larger class of random times, and so potentially enhancing the applicability of these concepts to a broader range of scenarios.
These results naturally carry mathematical value independently of the present context of application.

The paper is structured as follows:
In Section~\ref{sec:model}, we provide a detailed definition of the set-up and the main assumptions that are used throughout the paper. %In addition, it includes the basic concepts of the forward integral in subsection \ref{subsec:forward.integral}.
In Section~\ref{sec:opt.problem}, we present the optimization problem and the set of admissible strategies.
The problem is divided into two parts, before and after the occurrence of a default, and we rewrite it as a different stochastic control functional. 
In Subsection~\ref{subsec:white.noise}, we introduce key concepts of anticipating calculus in the context of the Hida-Malliavin distribution space.
Additionally, we explicitly solve the problem for the logarithmic utility in Proposition~\ref{prop.opt.portfolio.after} and Theorem~\ref{theo.opt.log.bef}.
In Section~\ref{sec:suff.aproach}, we provide sufficient conditions for the existence of the optimal strategy and the semimartingale decomposition for both the Brownian motion and the compensated Poisson random measure in the enlarged filtration and for a general utility function.
Finally, we conclude the paper with some important remarks.
Appendices~\ref{App.Brownian} and~\ref{App.Poisson} contain further details on the anticipating calculus subsection, covering both the Brownian motion and the compensated Poisson random measure.
\section{A market model with default risk}\label{sec:model}
We work in a complete filtered probability space~$(\Omega, \cG_T,\PP,\bG)$, where the filtration $\bG$ is assumed to be right-continuous. 
We assume that the agent is going to invest in a market composed by two assets in a finite horizon time $T>0$. 
The first one is a risk-less bond and the second one is a risky-defaultable stock. 
The dynamics of both are given by the following SDEs,
\begin{subequations}\label{def.assets}
\begin{align}
    \frac{dD_t}{ D_t} &=\rho_tdt\ ,\quad D_0 = 1\label{def.assets1}\\
    \frac{d^-S_t}{S_{t-}} &= \mu_tdt + \sigma_td^-W_t + \int_{\bR_0}\theta_t(z)\tilde N(d^-t,dz) + \kappa_tdH_t \ ,\quad S_0 = s_0 > 0 \label{def.assets2}
\end{align}
\end{subequations}
where $W=\prT{W}$ is a Brownian motion 
with $\bF^{W}$ its natural filtration.
The random field
$\tilde N(dt,dz) = N(dt,dz)-\nu(dz)dt$ is
the compensated version of a Poisson random measure
$N(dt,dz)$ with $\EE[N(dt,dz)]=\nu(dz)dt$ and $\bF^{N}$ its natural filtration. 
By $d^-W_t$ and $\tilde N(d^-t,dz)$ we denote the forward integration, we refer to Subsection~\ref{subsec:white.noise} for the details.
We assume that~$W$ and~$\tilde N$ are independent under the probability measure $\PP$. 
Moreover, the Borel measure~$\nu(dz)$ on~$\bR_0:=\mathbb{R} \backslash\{0\}$ is $\sigma$-finite and satisfies
$$ \int_{\bR_0} \left(1\wedge z^2\right)\nu(dz) <+\infty\ , $$
where $a\wedge b := \min\{a,b\}$. 
The filtration $\bF = \{\cF_t:0\leq t \leq T\}$ is generated by the Brownian motion and the compensated Poisson measure, and it is augmented by the zero $\PP$-measure sets, $\cN$:
$$\cF_t := \sigma (W_s ,\ N((s, t],B):\ 0\leq s \leq t, B\in\cB(\bR_0))\vee\mathcal{N}\ .$$
By $\cB(\bR_0)$ we denote the Borel $\sigma$-algebra on $\bR_0$.% and\hbox{$\cF^W_t\vee\cF^N_t = \cF_t$} $\forall\,t\in[0,T]$ and $\PP-$a.s.

The process $H=\prT{H}$ represents the default of the risky asset and it is defined by 
$H_t = \bOne_{\{\tau\leq t\}}$
where $\tau$ is assumed to be an \hbox{$\cF_T$-measurable} random time.
%We consider the progressive enlargement of $\bF$ by the knowledge given by if $\tau$ has happened or not. 
In this framework the market is incomplete as the risky asset has discontinuous paths provided by the jumps of the compensated Poisson random measure as well as the jump of the default process~$H$.
The information flow $\bG = \{\cG_t:0\leq t \leq T\}$ 
is the minimum right-continuous and complete filtration such that $\tau$ is a~$\bG$-stopping time,
i.e., $$\cG_t:=\cF_t\vee\sigma(\tau\wedge t)\ .$$
Note that $\bG\supset\bF$ and it is a progressive  enlargement of filtration only if $\tau$ is not an $\bF$-stopping time. 
We define the (right-continuous and complete) filtration of an agent who only observes the price process $S$ as $\bS=\{\cS_t:0\leq t \leq T\}$,
$$ \cS_t := \sigma (S_s:\ 0\leq s \leq t)\vee\mathcal{N}\ .$$
Note that $\cS_t\subset \cG_t$ but, in general, it is not equal. Moreover, $\cS_t\subset \cF_t$ only if $\tau$ is an $\bF$-stopping time, in which case, as we mentioned above, $\cF_t=\cG_t$. 
The different information flows require different assumptions on the nature of the driving noises, i.e., 
their being a semimartingale or not with respect to the various filtrations and consequently different assumptions of predictability on the integrands for an It\^o type calculus.
This justifies our use of the forward integration,
which provides a universal framework where different filtrations can be considered simultaneously and that reduces to the classical It\^o
integration when measurability and semimartingale conditions are matched.
This approach was first introduced in \cite{diNunno2014}.
Else, if one wishes to keep the framework of It\^o integration one has 
to prove the desired property that the default process $H$ is a semimartingale in $\bG$ and the compensator satisfies some continuity conditions, and naturally also $W$ and $\tilde N$ to be semimartingales.
This is the critical point for applying It\^o type calculus.
In this case, the Doob-Meyer decomposition claims that there exists a unique $\bG$-predictable increasing process $A$ such that 
$J:=(H_t-A_t,\ 0\leq t\leq T)$ is a~$\bG$-martingale, 
see for example Section~3.3 in \cite{ElliotJeanblancYor2000}. 
To short the notation, we define the $\bF$-\emph{Azéma} supermartingale as
\begin{equation}
    Z_t:=\PP(\tau > t \vert \cF_t)\ .
\end{equation}
As $\tau$ is $\cF_T$-measurable, note that $Z_T = \bOne_{\{\tau > T\}} = 1 - H_T$.
The following lemma gives an explicit expression for the compensator $A$, to see the details Proposition 2.15 of \cite{Aksamit2017} can be consulted.
\begin{Lemma}\label{H.decomposition}
The process $J = H - A$, with
\begin{equation}%\label{H.decomposition}
    J_t = H_t - \int_0^{t\wedge\tau}\frac{dA^{\tau}_s}{Z_{s-}}\ ,\quad 0\leq t \leq T\ ,
\end{equation}
is a $(\PP,\bG)$-martingale,
where $A^{\tau}$ is the $\bF$-dual predictable projection of $H$.
\end{Lemma}
\noindent
Using \cite{Neal95}, we know that the $\bF$-dual predictable projection of $H$ is a bounded variation process.
As it is discussed in \cite{ElliotJeanblancYor2000} (originally proved in \cite{Dellacherie72}), 
if the compensator of $H$ is continuous, then $\tau$ is not an $\bF$-stopping time and it is a totally inaccessible $\bG$-stopping time.
In order to exploit some nice properties of the default time, in the literature there exist two main approaches: 
the first one is called the intensity approach and it is based on assuming that the process $A$ is absolutely continuous with respect to the Lebesgue measure, i.e., 
there exists a $\bG$-predictable process $\lambda=\prT{\lambda}$ such that
\begin{equation}\label{approach.intensity}
    \textit{(intensity hypothesis)}\quad  H_t-\int_0^t\lambda_sds\ ,\quad 0\leq t\leq T\ , 
\end{equation}
is a $\bG$-martingale.
Note that, by Lemma~\ref{H.decomposition}, the compensator always exists but under hypothesis~\eqref{approach.intensity} the singular and the jump part of the compensator are eliminated.
The second one is called the density approach and it is based on assuming that the \emph{Jacod density hypothesis} holds true, i.e.,
there exists a process $q(\eta) = (q_t(\eta),\,0 \leq t \leq T)$ such that,
\begin{equation}\label{approach.density}
    \textit{(density hypothesis)}\quad \PP(\tau\in d\eta|\cF_t) = q_t(\eta)\zeta(d\eta)\ ,%\quad 0\leq t\leq T\ .
\end{equation}
being $\zeta$ the law of $\tau$.
%\BER{[Please define $\zeta(d\eta)$]}
In \cite{ElKarouiJeanblancJiao2009}, it is proved that the density approach is stronger that the intensity one.
In particular, assuming the {density hypothesis}~\eqref{approach.density} holds true, 
the process~$\lambda$ in~\eqref{approach.intensity} can be fully recovered. 
However, assuming~\eqref{approach.intensity}, the process~$q_t(\eta)$ can be recovered only in case~$t\leq \eta$. 
The {density hypothesis} is often assumed in default risk setting, we refer to the papers~\cite{diTella2021,HillairetJiao11,HillairetJiao15,JiaoPham2011} %blanchetHillairetJiao17,
as main examples. 
The intensity hypothesis is also used in the literature, for example in~\cite{LimQuenez2015}.
In some papers, the authors assume the so-called immersion property or ($\mathcal{H}$)-hypothesis, 
in which the processes~$W$ and~$\tilde N$ are~$\bG$-martingales. 
We refer to, e.g,~\cite{KharroubiLim13} in the context of mean-variance hedging or Section~3.2 of \cite{ElKarouiJeanblancJiao2009} to see its relation with the density hypothesis.

In this work, a priori we assume neither the density nor the intensity approach, nor the \hbox{{($\mathcal{H}$)-hypothesis}}, 
as we work in the framework of forward integration. 
In such a way, we can consider cases such as the time in which the Brownian motion reaches its maximum, see Example~\ref{ex.argmax} below for more details.
%\BER{[Does this r.v. satisfy the Jacod hypotesis?]}
In Section~\ref{sec:suff.aproach} we are looking for sufficient optimality conditions for a general utility function, we assume the existence of a local maximum strategy for our problem and we derive, for example, the $\bG$-semimartingality of the processes $W$ and $\tilde N$.
%In order to solve the problem for a general utility function, 
%in Section~\ref{sec:suff.aproach} we assume the existence of a local maximum strategy for our problem and, as a consequence of the optimality, we prove that the intensity approach must hold.

Going back to the market model \eqref{def.assets}, we assume that the market coefficients are càglàd~{$\bG$-adapted} processes satisfying
%The drift process $\mu$ and the default coefficient $\kappa$ are càglàd $\bG$-adapted stochastic processes, the volatility process $\sigma$ is a càglàd stochastic process and the random variable $\sigma_t$ is $\cF^B_T$ measurable, i.e., we assume that is an anticipating random variable in the filtration of the Brownian motion.  The process $\theta$ is a càglàd stochastic process such that $\theta_t(z)$ is measurable in $\cF_T^N$ $\nu$-a.e. 
 the following integrability condition
\begin{equation}\label{market.coef.integr}
     \EE\left[\int_0^T\abs{\rho_t}+\abs{\mu_t}+\abs{\sigma_t}^2+\int_{\bR_0}\abs{\theta_t(z)}^2\nu(dz)\,dt\right] <+\infty\ .
\end{equation}
The default time $\tau$ satisfies the following condition
\begin{equation}\label{assumption.jumps}
    \PP(\exists\,t\in[0,T]\text{ such that }\{\tau = t\}\text{ holds true and }N(\Delta t,B)>0)) = 0
\end{equation}
for any $B\subset\bR_0$ compact and $N(\Delta t,B):=N((0,t],B)-N((0,t),B)$.
This means that the default time does not occur $\PP$-a.s. at the same time of any jump of $N$.
In particular, assumption~\eqref{assumption.jumps} is used in Proposition~\ref{prop.formula.ito.forward} below to apply the It\^o formula under forward integration.

By Lemma 4.4 of \cite{Jeulin1980}, any $\bG$-adapted process $Y=\prT{Y}$ can be expressed as 
$Y_t= Y^{\bF}_t\bOne_{\{\tau > t\}} + Y_t(\tau)\bOne_{\{\tau\leq t\}}$, 
where $Y^{\bF}$ is an $\bF$-adapted process and $Y_t(\tau)$ is $\cF_t\otimes\sigma (\tau)$-measurable. Using this fact, we can express the following processes in a decomposition before and after the occurrence of the default time:
\begin{subequations}\label{def.G.decomp.market.coef}
\begin{align}
    S_t &= S^{\bF}_t\bOne_{\{\tau > t\}} + S_t(\tau)\bOne_{\{\tau\leq t\}}\label{G.process.S}\\
    \rho_t &= \rho^{\bF}_t\bOne_{\{\tau > t\}} + \rho_t(\tau)\bOne_{\{\tau\leq t\}}\\
    \mu_t &= \mu^{\bF}_t\bOne_{\{\tau > t\}} + \mu_t(\tau)\bOne_{\{\tau\leq t\}}\\
    \sigma_t &= \sigma^{\bF}_t\bOne_{\{\tau > t\}} + \sigma_t(\tau)\bOne_{\{\tau\leq t\}}\\
   \theta_t(z) &= \theta^{\bF}_t(z)\bOne_{\{\tau > t\}} + \theta_t(z,\tau)\bOne_{\{\tau\leq t\}}\ .
\end{align}
\end{subequations}
The integral with respect to the default process $H_t = \bOne_{\{\tau\leq t\}}$ is defined as a Riemann-Stieltjes integral, i.e., 
\begin{equation}\label{def.integral.default}
   \int_0^{t}\phi_s dH_s := \phi_\tau H_t = \phi_\tau\bOne_{\{\tau\leq t\}}\ . %< \infty\quad \PP\text{-a.s.}
\end{equation}
Moreover, we assume that $\kappa_t\neq 0$ $\PP$-$a.s.$, 
to guarantee the presence of the default in the model.
The following proposition is an adaptation of
the version of Theorem~2.5 of~\cite{diNunno2014} to the framework of defaultable assets presented in \eqref{def.assets2}.
\begin{Proposition}\label{prop.formula.ito.forward}
Let $Y$ be a process that satisfies the following equation,
\begin{equation}\label{forward.process}
    \frac{d^{-}Y_t}{Y_{t-}} = \mu_tdt +\sigma_td^{-}W_t +\int_{\bR_0}\theta_t(z)\tilde N(d^-t,dz) +\kappa_t dH_t\ ,
\end{equation}
where %$\mu$ and %$\PP$-a.s.,
$\sigma$ is forward integrable with respect to $W$, 
$\theta$ and $\abs{\theta}$ are forward integrable with respect to~$\tilde N$, \eqref{market.coef.integr} and \eqref{assumption.jumps} hold true
and $\kappa < \infty$, $\PP$-$a.s$.
Assume $f\in C^2(\bR)$ then,
\begin{align}\label{eq.forward.ito}
    df(Y_s) =& \left( f'(Y_{s-})Y_{s-}\mu_s + f''(Y_{s-})Y^2_{s-}\sigma^2_s\right)ds\notag\\
    &+ \int_{\bR_0}f(Y_{s-}+Y_{s-}\theta_s(z))-f(Y_{s-}) -f'(Y_{s-})Y_{s-}\theta_s(z) \,\nu(dz)ds\notag\\
    &+  f'(Y_{s-})Y_{s-}\sigma_sd^-W_s + \int_{\bR_0}f(Y_{s-}+Y_{s-}\theta_s(z))-f(Y_{s-})\,\tilde N(d^-s,dz)\notag\\
    &+\left(f(Y_{s-}+Y_{s-}\kappa_s)-f(Y_{s-})\right)\,dH_s\ .
\end{align}
\end{Proposition}
Applying Proposition \ref{prop.formula.ito.forward} to the dynamics of the risky asset,
given by \eqref{def.assets2}, %is a forward 
we get an explicit expression as follows, 
\begin{align*}
    \ln \frac{S_t}{s_0} =& \int_0^t \left(\mu_s -\frac{1}{2}\sigma^2_s+ \int_{\bR_0}\ln(1+\theta_s(z)) -\theta_s(z) \nu(dz)\right)ds\\
    &+ \int_0^t \sigma_sd^-W_s + \int_0^t\int_{\bR_0}\ln(1+\theta_s(z)) \tilde N(d^-s,dz) +\int_0^t \ln(1+\kappa_s)dH_s\ .
\end{align*}
To have the process $S$ well-defined, we assume that
\begin{align}
    -1&<\theta_t(z)\ ,\quad dt\times\nu(dz)\times \PP-\text{a.s.}\label{hyp.market.1}\\
    -1&< \kappa_t <  \infty\ ,\quad dt\times \PP-\text{a.s.}\label{hyp.market.2}
\end{align}
In the next section we study the optimal expected utility of the terminal wealth on the time horizon $[0,T]$.
By combining the approaches developed in~\cite{diNunno2014} 
and~\cite{JiaoPham2011} we split the optimization problem~\eqref{optimization.problem} into two new problems: 
one before default in the stochastic interval~$[[0,\tau\wedge T[[$, and the other one after default in the stochastic interval~$[[\tau\wedge T,T]]$, 
where the latter interval can be empty. 
To do this, we need to rewrite the dynamics of the risky asset in terms of before and after default.
The following proof differs from the one in~\cite{JiaoPham2011} for using the forward integration.
\begin{Proposition}\label{S.SDE.before.after}
The component $S^\bF$ of S in \eqref{G.process.S} satisfies the following SDE,
\begin{subequations}
\begin{align}
    \frac{dS_t^{\bF}}{S_{t-}^{\bF}} &= \mu^{\bF}_tdt + \sigma^{\bF}_tdW_t + \int_{\bR_0}\theta^{\bF}_t(z)\tilde N(dt,dz),\ t\in[[0,\tau\wedge T[[\\
    S_0^{\bF} &= s_0 >0\ ,
\end{align}
\end{subequations}
and the component $S(\tau)$ satisfies
\begin{subequations}
\begin{align}
    \frac{dS_t(\tau)}{S_{t-}(\tau)} &= \mu_t(\tau)dt + \sigma_t(\tau)d^-W_t  + \int_{\bR_0}\theta_t(z,\tau) \tilde N(d^-t,dz),\quad  t\in[[\tau\wedge T,T]]\\
    S_\tau(\tau) &= S_{\tau-}^{\bF}(1+\kappa_\tau)\ .
\end{align}
\end{subequations}
%and on the set $\{\tau=\eta\}$ we have 
%\begin{align*}
%    \frac{dS_t(\eta)}{S_{t-}(\eta)} &= \mu_t(\eta)dt + \sigma_t(\eta)dW_t  + \int_{\bR_0}\theta_t(z,\eta) \tilde N(dt,dz)\\
%    S_\eta(\eta) &= S_{\eta-}^{\bF}(1+\kappa_\eta)\ ,\quad  t\in[\eta,T]\ .
%\end{align*}
\end{Proposition}
\begin{proof}
As in \eqref{G.process.S}, we express the price process as
$S_t = S_t^{\bF}\bOne_{\{\tau > t\}} + S_t(\tau)\bOne_{\{\tau\leq t\}}$. 
%i.e., we can split the solution under $\{\tau\leq t\}$ and $\{\tau> t\}$. 
Under~$\{\tau> t\}$ it is clear that
\hbox{$S^{\bF}_t\bOne_{\{\tau > t\}} = S_t\bOne_{\{\tau > t\}}$}, then by applying Proposition~\ref{prop.formula.ito.forward} we have the following
\begin{align*}
    S^{\bF}_t\bOne_{\{\tau > t\}} =& S_0 \exp \Big\{ \int_0^t \mu_s -\frac{1}{2}\sigma^2_s+ \int_{\bR_0}\ln(1+\theta_s(z)) -\theta_s(z) \,\nu(dz)ds\\
    &+ \int_0^t \sigma_sd^-W_s + \int_0^t\int_{\bR_0}\ln(1+\theta_s(z)) \tilde N(d^-s,dz)\Big\}\bOne_{\{\tau > t\}}\ ,\\
    =& S_0 \exp \Big\{ \int_0^t \mu^{\bF}_s -\frac{1}{2}\left(\sigma^{\bF}_s\right)^2 + \int_{\bR_0}\ln(1+\theta_s^{\bF}(z)) -\theta^{\bF}_s(z) \,\nu(dz)ds\\
    &+ \int_0^t \sigma^{\bF}_sdW_s + \int_0^t\int_{\bR_0}\ln(1+\theta^{\bF}_s(z)) \tilde N(ds,dz)\Big\}\bOne_{\{\tau > t\}}\ ,
\end{align*}
where, according to the decomposition of the market coefficients \eqref{def.G.decomp.market.coef}, 
we have used the following pathwise relation,
$$ \exp\Big\{\int_0^t\mu^\bF_s \bOne_{\{\tau > s\}} + \mu_s(\tau) \bOne_{\{\tau \leq s\}} ds\Big\} \bOne_{\{\tau > t\}} =  \exp\Big\{\int_0^t\mu^\bF_s ds\Big\} \bOne_{\{\tau > t\}}\ .  $$ 
Now, on the set $\{\tau\leq t \leq T\}$, we have
\begin{align*}
    S_t(\tau)\bOne_{\{\tau\leq t\}} =& S_t\bOne_{\{\tau\leq t\}} =  S_0\exp \Big\{ \int_0^t \mu_s -\frac{1}{2}\sigma^2_sds\\
    &+ \int_0^t\int_{\bR_0}\ln(1+\theta_s(z)) -\theta_s(z) \,\nu(dz)ds+ \int_0^t \sigma_sd^-W_s\\ 
    &+ \int_0^t\int_{\bR_0} \ln(1+\theta_s(z)) \tilde N(d^-s,dz) + \ln(1+\kappa_\tau)\bOne_{\{\tau\leq t\}}\Big\}\bOne_{\{\tau\leq t\}}\ ,
\end{align*}
where we used the following pathwise equality 
$$ (1+\kappa_\tau)^{\bOne_{\{\tau\leq t\}}}\bOne_{\{\tau\leq t\}}=(1+\kappa_\tau)\bOne_{\{\tau\leq t\}}\ .$$
As for the forward integrals, we use the following equality
\begin{align*}
    \int_0^t\sigma_sd^-W_s &= \int_0^t\bOne_{\{\tau > s\}}\sigma^{\bF}_sd^-W_s + \int_0^t\bOne_{\{\tau \leq s\}}\sigma_s(\tau)d^-W_s%//
    %&= \int_0^\eta \sigma^{\bF}_sdW_s\Big|_{\eta=\tau} + \int_\eta^t\sigma_s(\eta)dW_s\Big|_{\eta=\tau}\ ,\quad 0\leq\eta\leq t\ ,
\end{align*}
%where we have used Lemma 7.20 in \cite{OksendalDraouil15} in order to substitute $\tau$ by $\eta\in[0,t]$. 
to get  
\begin{align*}
    S_t(\tau)\bOne_{\{\tau\leq t\}} =& S_{\tau-}^{\bF} \left(1+\kappa_\tau\right) \exp \Big\{ \int_{\tau}^t \mu_s(\tau) -\frac{1}{2}\sigma^2_s(\tau)ds\\
    &+ \int_{\tau}^t\int_{\bR_0}\ln(1+\theta_s(z,\tau)) -\theta_s(z,\tau) \,\nu(dz)ds\\
    &+ \int_{\tau}^t \sigma_s(\tau)d^-W_s + \int_{\tau}^t\int_{\bR_0}\ln(1+\theta_s(z,\tau)) \tilde N(d^-s,dz)\Big\}\bOne_{\{\tau\leq t\}}\ .
\end{align*}
%\begin{align*}
%    S_t(\tau)\bOne_{\{\tau\leq t\}} &= S_{\eta-}^{\bF} \left(1+\kappa_\eta\right) \exp \Big\{ \int_{\eta}^t \mu_s(\eta) -\frac{1}{2}\sigma^2_s(\eta)ds\\
%    &+ \int_{\eta}^t\int_{\bR_0}\ln(1+\theta_s(z,\eta)) -\theta_s(z,\eta) \,\nu(dz)ds\\
%    &+ \int_{\eta}^t \sigma_s(\eta)dW_s\Big|_{\tau = \eta} + \int_{\eta}^t\int_{\bR_0}\ln(1+\theta_s(z,\eta)) \tilde N(ds,dz)\Big|_{\tau = \eta}\Big\}\ .
%\end{align*}
The result then holds by expressing the equation above in differential form.
\end{proof}
%\subsection{Additional Notation}\label{subsec:add.notation}
Let $\PP$ and $\QQ$ denote two probability measures.
$\EE$ indicates the expectation operator under measure $\PP$ and $\EE_{\QQ}$ the one under measure $\QQ$.
We denote by $dt$ the Lebesgue measure and by $L^2(B,dt)$ 
the set of deterministic functions which satisfy
$\norm{f}^2_2 = \int_B f^2(t)dt  <+\infty,$ $B\subseteq\bR$.
Sometimes we omit $B$ when it is clear from the context and we write $L^2(dt)$.
We use $L^2(\nu\times dt)$ in the corresponding product space 
$\bR_0\times[0,T]$ with its associated norm.
We define the space $L^2(\PP,\cF_T)$,
or simply $L^2(\PP)$, as
the class of equivalence of random variables $F$ such that $\EE[F^2] <+ \infty$ with
$\cF_T$-measurable representative.
In addition, we let $L^2(dt\times\PP,\bF)$, or simply
$L^2(dt\times\PP)$, to be the space of equivalent class of processes $X$ satisfying $\int_0^T \EE\left[X_s^2\right] ds  <+ \infty$ with $\bF$-adapted representative.

\section{Optimization Problem}\label{sec:opt.problem}
Using the previous set-up, it is assumed that an agent can control her portfolio by a \textit{self-financing} process $\pi=\prT{\pi}$,
with the aim to maximize her expected utility gain at a finite terminal time $T>0$.
If we denote by $X^\pi=\prT{X^\pi}$ the wealth of the portfolio of the investor under $\pi$,
its dynamics are given by the following SDE, for $0 \leq t \leq T$,
\begin{equation}\label{def.X}
    \frac{d^-X_t^{\pi}}{X_{t-}^{\pi}} = (1-\pi_t) \frac{dD_t}{D_t} + \pi_t \frac{d^-S_t}{S_{t-}} \ ,\quad X^{\pi}_0=x_0 > 0\ .
\end{equation}
By using the evolution of both assets given in \eqref{def.assets} we get
\begin{align*}
    \frac{d^-X^{\pi}_t}{X^{\pi}_{t-}} =& (1-\pi_t) \rho_tdt + \pi_t \left( \mu_tdt + \sigma_t d^-W_t+\int_{\bR_0}\theta_t(z)\tilde N(d^-t,dz) + \kappa_tdH_t \right)
\end{align*}
where the SDE is well-defined on the probability space $(\Omega, \cG_T,\PP,\bG)$ within the progressive enlarged filtration.
Before giving a proper definition of the set of processes $\pi$ that we consider, we look for the natural conditions they should satisfy.
If we apply the It\^o formula for forward integration given in Proposition~\ref{prop.formula.ito.forward} to $ \ln X^{\pi}_t$, we get,
\begin{align}\label{solution.X}
\ln \frac{X^{\pi}_t}{x_0} =& \int_0^{t}\left( \rho_s + \pi_s (\mu_s-\rho_s) - \frac{1}{2}\pi^2_s\sigma^2_s +\int_{\bR_0} \ln(1+\pi_s\theta_s(z))-\pi_s\theta_s(z) \nu(dz)\right)ds\notag \\
&+\int_0^{t}\pi_s\sigma_sd^-W_s+ \int_0^{t} \int_{\bR_0} \ln(1+\pi_s\theta_s(z)) \tilde N(d^-s,dz) + \int_0^t\ln (1+\pi_s\kappa_s)dH_s\ ,
\end{align}
provided that these integrals are well-defined. To assure  this, we assume the following integrability conditions.
\begin{equation}\label{hyp1}
        \EE\left[ \int_0^T\abs{\rho_s}+\abs{\pi_s}\abs{\mu_s-\rho_s}+\pi^2_s\sigma^2_s + \int_{\bR_0}\pi^2_s\theta^2_s(z)\nu(dz)\  ds \right] < +\infty
\end{equation}
%\begin{align}
%        &\EE\left[ \int_0^T\abs{\rho_s}+\abs{\pi_s}\abs{\mu_s-\rho_s}+\pi^2_s\sigma^2_s\, ds \right] < +\infty\label{hyp1}\\
%        &\EE\left[ \int_0^T\int_{\bR_0}\pi^2_s\theta^2_s(z)\, \nu(dz) ds \right] < +\infty\ .\label{hyp2}
%\end{align}
To guarantee that \hbox{$X^{\pi}$} is well-defined, 
we assume that for all $\pi$ there exists $\epsilon^\pi>0$ such that,
\begin{equation}\label{hyp3}
\min\{ 1+\pi_t\theta_t(z), 1+\pi_t\kappa_t\} > \epsilon^{\pi}\ ,
\quad dt\times\nu(dz)\times \PP-\text{a.s.}
\end{equation}
In order to carry out the optimization of the portfolio, we consider a \textit{utility function} \hbox{$U:\bR^+\rightarrow \bR$} denoting the utility of the investor. We extend the domain of the function to $\bR$ by assuming that $U(x):=-\infty$ if $x<0$.
In addition we assume that it is continuously differentiable, strictly increasing and strictly concave in its domain and satisfying the Inada conditions,
\begin{subequations}
\begin{align}
    U'(0)       &:= \lim_{x\to 0^+}U'(x) = + \infty \label{Inada.condition1}\\
    U'(+\infty) &:= \lim_{x\to\infty}U'(x) = 0\ .\label{Inada.condition2}
\end{align}
\end{subequations}
Now, we can define properly the optimization problem as the supremum of the expected utility gain of the agent's wealth at the finite horizon time $T$.
\begin{equation}\label{optimization.problem}
    J(x_0,\pi):= \EE\left[U( X^{\pi}_T)|X^\pi_0 = x_0\right]\ ,\quad \bV_T^{\bE}:= \sup_{\pi\in\cA(\bE)} J(x_0,\pi)\ ,\quad \bE\in \{\bF,\bG\}\ .
\end{equation}
%where $\bV_T^{\bE}$ is defined as the optimal value of the portfolio at time $T$ given the information flow $\bE$,
Finally we give the definition of the set $\cA(\bE)$ of admissible strategies for the agent playing with information flow $\bE\in \{\bF,\bG\}$.
\begin{Definition}
In a financial market whose coefficients verify \eqref{market.coef.integr}, \eqref{assumption.jumps}, \eqref{hyp.market.1} and \eqref{hyp.market.2},
we define the set $\cA(\bE)$ of admissible strategies as the set which contains all the portfolios $\pi$ satisfying  \eqref{hyp1} and \eqref{hyp3} such that,
\begin{itemize}
    \item $\pi$ is càglàd and adapted w.r.t. the filtration $\bE$.
    \item $\pi\sigma$ is càglàd and forward integrable w.r.t. $W$.
    \item $\pi\theta$, $\ln(1+\pi\theta)$ and $\frac{\pi\theta}{1+\pi\theta}$ are càglàd and forward integrable w.r.t. $\tilde N$.
\end{itemize}
\end{Definition}
Since in some cases we can characterize the optimal solution for the problem \eqref{optimization.problem} only in a local sense, 
we introduce the definition of local maximum.
\begin{Definition}
We say that $\pi\in\cA(\bE)$ is a local maximum for the problem~\eqref{optimization.problem} if  
\begin{equation} 
\EE[U( X^{\pi+y\beta}_T)]\leq \EE[U(X^{\pi}_T)]\ , 
\end{equation}
for all bounded $\beta\in\cA(\bE)$ and $\abs{y}<\delta$, for some $\delta > 0$.
\end{Definition}
With the same procedure applied in Section \ref{sec:model}, 
we rewrite the dynamics of the process $X^{\pi}$ before and after default, we omit the proof because it is analogous to Proposition~\ref{S.SDE.before.after}.
We recall the decomposition
$X^{\pi}_t =
X^{\pi,\bF}_t\bOne_{\{\tau > t\}} 
+ X_t^{\pi}(\tau)\bOne_{\{\tau\leq t\}}$.
\begin{Proposition}\label{X.SDE.before.after}
The component $X^{\pi,\bF}$ of $X$ satisfies the following SDE,
\begin{subequations}
\begin{align}
    \frac{dX^{\pi,\bF}_t}{X^{\pi,\bF}_{t-}} &= (1-\pi^{\bF}_t) \rho^{\bF}_tdt +\pi^{\bF}_t \frac{dS_t^{\bF}}{S_{t-}^{\bF}},\quad t\in[[0,\tau\wedge T[[ \\
    X^{\pi,\bF}_0 &= x_0\ ,
\end{align}
\end{subequations}
and the component $X^\pi(\tau)$
\begin{subequations}\label{def.SDE.X.after}
\begin{align}
    \frac{dX^{\pi}_t(\tau)}{X^{\pi}_{t-}(\tau)} &= (1-\pi_t(\tau))\rho_t(\tau)dt +\pi_t(\tau) \frac{d^-S_t(\tau)}{S_{t-}(\tau)} ,\quad t\in[[\tau\wedge T,T]] \label{X.SDE.after.tau}\\
    X^{\pi}_\tau(\tau) &= X^{\pi,\bF}_{\tau-}(1+\pi^{\bF}_\tau\kappa_\tau)\ . \label{X.just.after.tau}
\end{align}
\end{subequations}
%for any $\eta\in[0,T]$,
%\begin{align*}
%    \frac{dX^{\pi}_t(\eta)}{X^{\pi}_{t-}(\eta)}\Big|_{\tau = \eta} &= (1-\pi_t(\eta))\rho_t(\eta)dt +\pi_t(\eta) \frac{d^-S_t(\eta)}{S_{t-}(\eta)}\Big|_{\tau = \eta} ,\quad t\in[\eta,T] \\
%    X^{\pi}_\eta(\eta) &= X^{\pi,\bF}_{\eta-}(1+\pi^{\bF}_\eta\kappa_\eta)\ .
%\end{align*}
\end{Proposition}
Since
$1 = \bOne_{\{\tau > T\}} + \bOne_{\{\tau\leq T\}}$,
we can \textit{split} the main problem in two new problems,
one before and the other after the default time with respect to the horizon time~$T$. 
In order to short the notation, we define
\begin{equation}\label{tau.measure}
    d\PP^{\tau}(\eta):=\PP(\tau\in d\eta\,|\,\tau \leq T)\ ,\quad 0\leq\eta\leq T\ .
\end{equation}
Following \cite{JiaoPham2011} we get the next lemma, 
where $\EE^{x_0}[\cdot] := \EE[\cdot |X_0 = x_0]$.
\begin{Lemma}
The expected utility gain of the agent’s wealth can be written as follows
\begin{equation}\label{J.bef.aft}
    J(x_0,\pi) = \EE^{x_0}\left[ U\left(X^{\pi,\bF}_T\right) (1-H_T) + \int_{0}^{T} \EE\left[ U\left(X^{\pi}_T(\eta)\right)|X^{\pi}_\eta(\eta) \right] d\PP^{\tau}(\eta)\right]  \ .
\end{equation}
\end{Lemma}
\begin{proof}
\begin{align}\label{other.optimal.functional}
J(x_0,\pi)&=\EE^{x_0}[U(X^{\pi}_T)] =\EE^{x_0}\left[ U\left(X^{\pi}_T\right) \bOne_{\{\tau > T\}} +  U\left(X^{\pi}_T\right) \bOne_{\{\tau\leq T\}}\right] \notag\\
%&=\EE^{x_0}\left[U\left(X^{\pi,\bF}_T\right) (1-H_T)  + \EE\left[U\left(X^{\pi}_T\right) \bOne_{\{\tau\leq T\}}|\cG_\tau\right] \right] \notag\\
%&=\EE^{x_0}\left[ U\left(X^{\pi,\bF}_T\right) (1-H_T) + \EE\left[U\left(X^{\pi}_T(\tau)\right) |\cG_\tau\right]\bOne_{\{\tau\leq T\}}\right] \notag\\
&=\EE^{x_0}\left[ U\left(X^{\pi,\bF}_T\right) (1-H_T) + \EE\left[U\left(X^{\pi}_T(\tau)\right) |\tau,X^{\pi}_\tau(\tau)\right]\bOne_{\{\tau\leq T\}}\right]\\
&=\EE^{x_0}\left[ U\left(X^{\pi,\bF}_T\right) (1-H_T) + \int_{0}^{T} \EE\left[ U\left(X^{\pi}_T(\eta)\right)|\tau=\eta,X^{\pi}_\eta(\eta) \right] d\PP^{\tau}(\eta)\right]  \notag \ .
\end{align}
%where %we have used the $\sigma$-algebra
%$$\cG_\tau := \{F\in\cG_T:F\cap\{\tau\leq t\}\in\cG_t\ \forall\, 0\leq t\leq T\}\ ,$$
In the second equality we used the SDE \eqref{def.SDE.X.after}.
\end{proof}
As every $\bG$-admissible $\pi$ can be written as
$\pi = (\pi^{\bF}_t\bOne_{\{\tau > t\}} + \pi_t(\tau)\bOne_{\{\tau\leq t\}},\,0\leq t \leq T)$, 
we define the set 
$$\cA(\tau) = \{ \pi(\tau) :\exists \pi^{\bF}\in\cA(\bF)\text{ such that }  \pi^{\bF}_\cdot\bOne_{\{\tau > \cdot\}} + \pi_\cdot(\tau)\bOne_{\{\tau\leq \cdot\}}\in\cA(\bG)  \}\ .$$
We define the after default optimization problem as
\begin{equation}\label{optimization.problem.after}
    \bV(\eta,x) =\esssup _{\pi \in \cA(\tau)} \EE^{\eta,x}\left[U\left(X^{\pi}_T(\tau)\right) %| \tau = \eta, X^{\pi}_\tau(\tau) = x
    \right]
    \ ,\quad  (\eta, x) \in[0, T] \times(0, \infty) \ ,
\end{equation}
where
$\EE^{\eta,x}[\cdot]:=\EE[\cdot|\tau = \eta, X^{\pi}_\tau(\tau) = x]$.
Next proposition extends the results in~\cite{JiaoPham2011} by dropping the {density hypothesis}. 
This implies we cannot assume absolute continuity of the distribution of $\tau$ with respect to the Lebesgue measure, then we use instead the measure introduced in~\eqref{tau.measure}. 
\begin{Proposition}\label{after-and-before}
If $\bV(\eta,x)<\infty,\,\PP$-a.s. for all $(\eta, x) \in[0, T] \times(0, \infty)$, 
then
\begin{align}\label{eq.HJB.after-before}
   \bV_T^{\bG} = \sup _{\pi^{\bF}\in\cA(\bF)} &\EE^{x_0}\left[U\left(X^{\pi,\bF}_T\right) (1-H_T)+\int_{0}^{T} \bV\left(\eta, X^{\pi,\bF}_\eta\left(1+\pi^{\bF}_\eta \kappa_\eta\right)\right) d\PP^{\tau}(\eta)\right]
\end{align}
\end{Proposition}
\begin{proof}
Let $\pi\in\cA(\bG)$, we can compute the functional,
\begin{align*}
J(\pi,x_0)&=\EE^{x_0}\left[U\left(X^{\pi,\bF}_T\right) (1-H_T) + \int_{0}^{T} \EE[U\left(X^{\pi}_T(\eta)\right)|X^\pi_\eta(\eta)]d\PP^{\tau}(\eta)  \right]\\
&\leq\EE^{x_0}\left[U\left(X^{\pi,\bF}_T\right) (1-H_T)  + \int_{0}^{T}  \bV(\eta, X^\pi_\eta(\eta))d\PP^{\tau}(\eta) \right]\leq \widehat \bV
\end{align*}
with
\begin{equation*}
    \widehat\bV:= \sup_{\pi^{\bF}\in\cA(\bF)}\EE^{x_0}\left[U\left(X^{\pi,\bF}_T\right) (1-H_T) + \int_{0}^{T}  \bV(\eta, X^{\pi,\bF}_\eta (1+\pi^{\bF}_\eta\kappa_\eta)) d\PP^{\tau}(\eta) \right]
\end{equation*}
and where in the last inequality we used \eqref{X.just.after.tau}.
We have proved that $\bV_T^{\bG}\leq\widehat \bV$,
now we aim to prove the opposite inequality.
We fix an arbitrary $\pi^{\bF}\in\cA(\bF)$. 
By the definition of the optimum $\bV$, 
for any $\omega \in \Omega$, $\eta \in[0, T]$ and $\epsilon > 0$, 
there exists $\pi^{\epsilon, \omega}(\eta) \in \cA(\eta)$ which is an $\epsilon$-optimal strategy, i.e., % for $\bV(\eta, \cdot )$, i.e.,% at  $\left(\omega, X(\omega, \eta,\eta)\right) .$ 
%Using a measurable result of -Jiao and Pham cite \textit{``Wagner, D.H.: Survey of measurable selection theorems: an update''}. Sect. 3-,one can find $\pi^{\epsilon}(\cdot,\eta) \in \cA(\eta)$ such that $\pi^{ \epsilon}(\omega,t, \eta)=\pi^{\epsilon, \omega}(\omega,t, \eta)$ holds$\PP \otimes dt\otimes d\eta$-a.s., and then
$$ \bV(\eta, X^{\pi}_\eta(\eta))-\epsilon \leq
\EE\left[U\left(X^{\pi^{\epsilon}}_T(\eta)\right)|\tau=\eta, X^{\pi^\epsilon}_\eta(\eta) \right] $$
%$ \PP \otimes d \theta$ -almost everywhere.
By constructing the strategy 
$\pi^{\epsilon}_t = \pi^{\bF}_t\bOne_{\{\tau>t\}} + \pi^{ \epsilon}_t(\tau)\bOne_{\{\tau\leq t\}}$
and using the previous inequalities, we get
\begin{align*}
\bV_T^{\bG} & \geq  \EE\left[  U\left(X^{\pi,\bF}_T\right) (1-H_T)  + \int_{0}^{T} \EE\left[ U\left(X^{\pi^{\epsilon}}_T(\eta)\right)|\tau=\eta,X^{\pi^\epsilon}_\eta(\eta)\right]d\PP^{\tau}(\eta) \right]\\
& \geq \EE\left[ U\left(X^{\pi,\bF}_T\right)  (1-H_T) + \int_{0}^{T} \bV\left(\eta, X^{\pi}_\eta(\eta)\right)d\PP^{\tau}(\eta)\right]-\epsilon
\end{align*}
and we get the result, since $\epsilon$ is arbitrary.
\end{proof}
\subsection{Anticipating calculus in the white noise framework}\label{subsec:white.noise}
For the sake of completeness, we introduce some notions of anticipating calculus in the context of the white noise analysis.
We mainly follow the lines of Chapters~5 and~13 in~\cite{DiNunnoOksendalProske2009}. 
The anticipating calculus is needed for the computations of Subsections \ref{subsec:after.default.log} and~\ref{subsec:before.default.log}.
\subsubsection{Brownian motion}\label{subsub.Brownian}
First, we introduce the \emph{Hermite polynomials} as follows
$$ h_n(s) := (-1)^n \exp\left(\frac{s^2}{2}\right)\frac{d^n}{ds^n}\exp\left(-\frac{s^2}{2}\right)\ ,\quad (s,n)\in [0,T]\times\bZ^+\ , $$
and the \emph{Hermite functions} defined as
$$ e_{k+1}(s) := \frac{\pi^{-\frac{1}{4}}}{\sqrt{k!}} \exp\left(-\frac{s^2}{2}\right) h_{k}(\sqrt{2}s)\ ,\quad (s,k)\in [0,T]\times\bZ^+\ .$$
It can be proved that the family $\{e_k\}_{k\geq 1}$ is an orthonormal basis of $L^2(\bR,dt)$. 
Let's consider the following stochastic It\^o integrals
$$ \theta_k := \int_0^T e_k(s)dW_s\ . $$
Denoting by $\cJ$ the set of finite non-negative multi-index $\alpha = (\alpha_1,...,\alpha_m)$ for $m\in\bN$,
we define 
\begin{equation}\label{def.chaos.H}
    H_{\alpha} := \prod_{k=1}^m h_{\alpha_k}(\theta_k)\ ,\quad  \alpha = (\alpha_1,...,\alpha_m)\in\cJ\ .
\end{equation}
The family $\{H_{\alpha}\}_{\alpha\in\cJ}$ is an orthogonal basis for the space $L^2(\PP)$ in the sense of the following theorem.
We refer to Theorems 2.2.3 and 2.2.4 of \cite{oksendal2010book} 
for a detailed proof.
\begin{Theorem}
Let $X$ be an $\cF^W_T$-measurable random variable in $L^2(\PP)$, then there exists a unique sequence $\{a_{\alpha}\}\subset\bR$ such that
$$ X = \sum_{\alpha\in\cJ}a_{\alpha}H_{\alpha}\ , $$
and the norm can be computed as 
$\norm{X}^2_{L^2(\PP)} = \sum\alpha!a_{\alpha}^2\ ,$ 
with $\alpha! = \alpha_1!\alpha_2!...\alpha_m!$. % for all $\alpha = (\alpha_1,...,\alpha_m),$ $m\geq 1$.
\end{Theorem}
%We introduce the notation  $$ \left(2\bN\right)^{\alpha} = \prod_{j=1}^{m}(2j)^{\alpha_j}\ ,\quad \alpha = (\alpha_1,...,\alpha_m)\in\cJ\ . $$
In the next definition we introduce the Hida test function and the Hida-Malliavin distribution spaces, that will play an important role in our calculations.
\begin{Definition}
\begin{itemize}
    \item Let $f = \sum_{\alpha\in\cJ}a_{\alpha}H_{\alpha}\in L^2(\PP)$ be a random variable, 
    we say that $f$ belongs to the Hida test function Hilbert space $(\cS)_{k}$, for $k\in\bR$, if 
    $$ \norm{f}^2_{k} := \sum_{\alpha\in\cJ} \alpha !\, a_{\alpha}^2 \prod_{j=1}^{m}(2j)^{k\alpha_j} < \infty\ . $$
    We define the Hida test function space
    $ (\cS) = \bigcap_{k\in\bR}(\cS)_{k} $
    equipped with the projective topology.
    \item Let $F = \sum_{\alpha\in\cJ}b_{\alpha}H_{\alpha}$ be a formal sum, we say that $F$ belongs to the Hida-Malliavin distribution Hilbert space $(\cS)_{-q}$, for $q\in\bR$, if 
    $$ \norm{F}^2_{-q} := \sum_{\alpha\in\cJ} \frac{\alpha !\,     b_{\alpha}^2}{\prod_{j=1}^{m}(2j)^{q\alpha_j}} < \infty\ . $$
    We define the Hida-Malliavin distribution space
    $ (\cS)^* = \bigcup_{q\in\bR}(\cS)_{-q} $
    equipped with the inductive topology, i.e, convergence is studied with $\norm{\cdot}_q$ for some $q\in\bR$.
\end{itemize}
\end{Definition}
The space $(\cS)^*$ is the dual of $(\cS)$ and we consider the action of $F$ on $f$ as follows,
$$ \langle F, f \rangle = \sum_{\alpha\in\cJ}a_\alpha b_\alpha \alpha!\ .$$
Note that the inclusions $(\cS)\subset L^2(\PP)\subset(\cS)^*$ holds true. 
We define the generalized expectation of $F=\sum b_{\alpha} H_{\alpha}\in (\cS)^*$ as $\EE[F] := b_0$,
and the generalized conditional expectation of $F$ with respect to the $\sigma$-algebras of the natural filtration $\bF^W$ as 
$$ \EE\left[F|\cF^W_t\right] := \sum_{\alpha\in\cJ} b_{\alpha}\EE\left[H_{\alpha}|\cF^W_t\right] $$
when it exists under the convergence in $(\cS)^*$.
In addition, we define 
$$ \EE[F|\cG_t] := \sum_{\alpha\in\cJ} b_{\alpha}\EE[H_{\alpha}|\cG_t] \ ,\quad \cF_t\subset \cG_t\ ,$$
with the convergence in $(\cS)^*$ and the object $Z:=\EE[H_{\alpha}|\cG_t]$ satisfies
$\EE[\bOne_A Z] = \EE[\bOne_A H_{\alpha}],$ $\forall A\in\cG_t.$
Next, we give the definition of Malliavin derivative in the space $(\cS)^*$.
\begin{Definition}\label{def.malliavin.S*}
Let $F=\sum_{\alpha\in\cJ}a_{\alpha}H_{\alpha}\in(\cS)^*$ be a Hida-Malliavin distribution, 
then we define the Malliavin derivative as
$$ D_t F = \sum_{\alpha\in\cJ}\sum_{k=1}^{\infty} a_{\alpha}\alpha_k e_k(t) H_{\alpha - \varepsilon(k)}\ , $$
whenever this sum converges in $(\cS)^*$ and where $\varepsilon(k) := (0,0,...,1)$ is the canonical vector with~$1$ in the $k$-th position.
We define the set $Dom(D_t)$ as its domain.
\end{Definition} 
As it is pointed out in Chapter 6 of \cite{DiNunnoOksendalProske2009}, 
Definition \ref{def.malliavin.S*} is a natural generalization to~$(\cS)^*$ of the $L^2(\PP)$-Malliavin derivative as they coincide in the domain of the second one, usually denoted by $\bD_{1,2}$,
and we have $L^2(\PP) \subset Dom(D_t) \subset (\cS)^*$.
Let $Y=\prT{Y}$ such that $Y_t\in(\cS)^*$ for any $t$, 
if the following limit exists
\begin{equation}\label{def.mall.trace.BM}
    D_{t^+} Y_t := \lim_{s\to t^+}D_s Y_t
\end{equation}
in $(\cS)^*$, we say that $D_{t^+}Y_t$ is the \emph{Malliavin trace} of $Y$.
\begin{Definition}
The singular (or pointwise) white noise 
$\bW=\{\bW_t:\,0\leq t\leq T\}$ is defined as
$\bW_t = \sum_{k=1}^{\infty} e_k(t)H_{\varepsilon(k)}$.
\end{Definition}
It can be proved that $\bW_t$ is well-defined as an object in $(\cS)^*$ for any $t$,
and it satisfies
$\frac{d}{dt}W_t = \bW_t$,
see Section 5.2 of \cite{DiNunnoOksendalProske2009} to consult the details.
Now that we have defined the notion of derivative we introduce different types of integration. 
The first one generalizes the Lebesgue integral for~$L^2([0,T],dt)$.
\begin{Definition}\label{def.integral.S}
If $Z:[0,T]\longrightarrow (\cS)^*$ has the property
$\langle Z_t,\psi\rangle\in L^1([0,T],dt)$ for any $\psi\in (\cS)$,
then we define the integral $\int_0^T Z_tdt$ as the unique element of $(\cS)^*$ such that
$$ \Bigl\langle \int_0^T Z_t dt,\psi\Bigr\rangle = \int_0^T \langle Z_t ,\psi\rangle dt\ ,\quad \psi\in (\cS) \ ,$$
and $Z$ is called $dt$-integrable in $(\cS)^*$.
\end{Definition}
In the next definition we introduce the notion of forward integral in the context of $(\cS)^*$.
\begin{Definition}\label{def.integral.forward.Brownian}
 A càglàd stochastic process $Y = \prT{Y}$
 is forward integrable in $(\cS)^*$ with respect to the Brownian motion $W=\prT{W}$ if the limit 
$$ \int_0^T Y_t d^-W_t = \lim_{\epsilon\to 0}\int_0^T Y_t\frac{W_{t+\epsilon}-W_t}{\epsilon} dt $$
exists in $(S)^*$.
\end{Definition}
When the limit in Definition~\ref{def.integral.forward.Brownian} is taken in probability we say that the process is classically forward integrable.
This object can be seen as a generalization of the \Ito integral and 
it provides a universal framework that considers different filtrations at the same time, in the sense of Lemma~\ref{forward.brownian.semimart} in Appendix~\ref{App.Brownian}.
We introduce an operation named Wick product within the space of distributions~$(\cS)^*$.
\begin{Definition}
Let $F = \sum_{\alpha\in\cJ}a_\alpha H_\alpha\ ,$  $G = \sum_{\beta\in\cJ}b_\beta H_\beta\in (\cS)^*$, then we define the Wick product as
$$ F\diamond G := \sum_{\alpha,\beta\in\cJ}a_\alpha b_\beta H_{\alpha+\beta}\ . $$
\end{Definition}
With the previous set-up we state the Skorohod integral within $(\cS)^*$. 
It is motivated by Theorem 5.20 of \cite{DiNunnoOksendalProske2009}, 
as a natural generalization of the definition of Skorohod integral for $L^2(dt\times\PP)$-processes.
\begin{Definition}\label{def.Sk.BM}
Let $Y = \prT{Y}$ be such that $Y_t\in (\cS)^*$, 
then we define the Skorohod integral as
$$\int_0^T Y_t \diamond \bW_t\, dt\ , $$
provided the integrand is $dt$-integrable in $(\cS)^*$.
\end{Definition}
The following lemma plays a crucial role in the main result of the section, Theorem \ref{theo.forw.skor.mall.BM}. 
It extends the domain of applicability of the Malliavin derivative by allowing that 
\hbox{$D_t F \in (\cS)^*$} for any $t$ 
and not only \hbox{$D_\cdot F\in L^2(dt\times\PP)$}.
\begin{Lemma}\label{lem.wick.ord.mall}
Let $F\in L^2(\PP)$ be a $\cF_T^W$-measurable such that $D_\cdot F$ is $dt$-integrable in $(\cS)^*$, then
\begin{equation}
    F \diamond (W_t-W_s) = F \cdot (W_t-W_s) - \int_s^t D_u F du\ ,\quad 0 \leq s < t \leq T\ .
\end{equation}
\end{Lemma}
\begin{proof}We refer to Appendix~\ref{App.Brownian} for the \hyperlink{proof.lem.wick.ord.mall}{details of the proof.}
\end{proof}
We are ready now for the main result of this section, which
provides an explicit relationship in $(\cS)^*$ between the forward integral, the Skorohod integral and the Malliavin trace.
We report that, in the context of fractional white noise analysis, 
a similar statement was attempted in Theorem 3.7 of \cite{BiaginiOksendalForward}, under the strong assumption of $D_{\cdot^+}Y_\cdot\in L^2(dt\times\PP)$. 
Here, we managed to avoid such assumption.
\begin{Theorem}\label{theo.forw.skor.mall.BM}
Let $Y\in L^2(dt\times\PP,\bG)$, 
if $D_{\cdot^+}Y_\cdot$ is $dt$-integrable in $(\cS)^*$, 
then 
\begin{equation}\label{eq.forw.skor.mall}
    \int_0^T Y_t d^{-}W_t = \int_0^T Y_t \diamond \bW_t dt + \int_0^T D_{t^+} Y_t dt
\end{equation}
holds true in $(\cS)^*$.
\end{Theorem}
\begin{proof}We refer to Appendix \ref{App.Brownian} for the \hyperlink{proof.theo.forw.skor.mall.BM}{details of the proof.}
\end{proof}

\subsubsection{Poisson random measure}\label{subsub.Poisson}
We repeat the main ideas of Subsection~\ref{subsub.Brownian} but in the case of the compensated Poisson random measure. 
In this case, we refer to Appendix~\ref{App.Poisson} for the details and we state only the main result and the most important definitions.
The first one is the Malliavin derivative, 
it involves the polynomial basis $\{p_j\}$ defined in \eqref{polinom.p}
and the $L^2(\PP)$-orthogonal basis $\{K_{\alpha}\}_{\alpha\in\cJ}$ defined in \eqref{eq.chaos.pois}, 
analogously to the Brownian basis $\{H_\alpha\}_{\alpha\in\cJ}$.
\begin{Definition}\label{def.Mallavin.pois}
Let $F=\sum_{\alpha\in\cJ}a_{\alpha}K_{\alpha}\in(\cS)^*$, 
then we define the Malliavin derivative as
$$ D_{t,z} F := \sum_{\alpha\in\cJ}\sum_{i,j=1}^{\infty} a_{\alpha}\alpha_{z(i,j)} e_i(t)p_j(z)\, K_{\alpha - \varepsilon(i,j)} \ , $$
whenever this sum converges in $(\cS)^*$ and the set $Dom(D_{t,z})$ as its domain.
\end{Definition}
If the limit $D_{t^+,z} \theta_t(z) := \lim_{s\to t^+}D_{s,z} \theta_t(z)$ exists
%\begin{equation}%\label{def.mall.trace.pois}
%    D_{t^+,z} F := \lim_{s\to t^+}D_{s,z} F
%\end{equation}
in $(\cS)^*$, we say that $D_{t^+,z} \theta_t(z)$ is the {Malliavin trace} of $\theta_t(z)$.
\begin{Definition}\label{def.integral.forward.Poisson}
A càglàd stochastic process $\theta = (\theta_t(z),\ 0\leq t\leq T,\,z\in \bR_0)$ 
is forward integrable in~$(\cS)^*$ with respect to the compensated Poisson random measure $\tilde N$ if the limit
$$\int_{0}^{T} \int_{\mathbb{R}_{0}} \theta_s(z) \tilde N\left(d^{-} s, d z\right):=\lim _{m \rightarrow \infty} \int_{0}^{T} \int_{\mathbb{R}_{0}} \theta_s(z) \bOne_{\{z\in U_m\}} \tilde N(d s, d z)$$
 exists in $(\cS)^*$, where $\{U_{m}:\,m\in\bN\}$ is an increasing sequence of compact sets $U_{m} \subseteq \mathbb{R}_{0}$ with 
 $\nu\left(U_m\right)<\infty$ such that $\lim_{m \rightarrow \infty} U_m=\bR_0$.
\end{Definition}
We finally introduce the Skorohod integral in $(\cS)^*$.
It involves the white noise $\bN$, stated in Definition \ref{def.white.noise.pois}, 
and the version of Wick product for the Poisson case, see Definition \ref{def.wick.pois}.
\begin{Definition}\label{def.sko.pois}
Let $Y = \prT{Y}$ be such that $Y_t\in (\cS)^*$, then we define the Skorohod integral as
\begin{equation}
    \int_0^T \int_{\bR_0}  Y_t \diamond \tilde\bN(t,z)\,\nu(dz)dt\ ,
\end{equation}
provided the integrand is $(\nu\times dt)$-integrable in $(\cS)^*$.
\end{Definition}
In the following, 
we state the main result of the section for the Poisson case.
It gives an explicit relationship among the forward integral, 
the Skorohod integral and the Malliavin trace in $(\cS)^*$.
We present the proof as a novel contribution, to the best of our knowledge, never attempted before.
We need to apply a different argument compared with the Brownian case because we can not handle with the Hermite polynomials.
In Sections 6.2 and 6.3 of~\cite{PrivaultBook} the Charlier polynomials are introduced playing a similar role for the Poisson process, 
however, only in case of $\nu(dz)=\lambda\delta_{z_0}(z)$ we can recover similar properties.
We refer to Lemma~15.5 of~\cite{DiNunnoOksendalProske2009} as the statement that we are extending in the next theorem.
\begin{Theorem}\label{theo.forw.sko.mall.pois}
Let $\theta$ be a forward integrable such that $\theta_t(z)\in L^2(\PP)$ and $D_{t^+,z}\theta_t(z)$, $\left(D_{t^+,z} \theta_t(z)\diamond \tilde \bN(t,z)\right)$ are $(\nu\times dt)$-integrable in $(\cS)^*$ for every $(t,z)\in[0,T]\times\bR_0$,
then 
\begin{align}\label{pois.forw.sko.mall}
    \int_0^T\int_{\bR_0} \theta_t(z) \tilde N(d^{-}t,dz) =& \int_0^T \int_{\bR_0} \left(\theta_t(z) +  D_{t^+,z}\theta_t(z)\right) \diamond \tilde\bN(t,z)\,\nu(dz)dt\notag\\ 
    &+ \int_0^T \int_{\bR_0}  D_{t^+,z} \theta_t(z)\,\nu(dz) dt
\end{align}
holds true in $(\cS)^*$.
\end{Theorem}
\begin{proof}
%Consider an increasing sequence of compact sets 
%$\{U_m\}_{m=1}^\infty$ such that $\cup_{m} U_m= \bR_0$.
%Let $U_m$ be a compact set of $\bR_0$ such that $0\not\in U_m$.
Let $m\in\bN$,
$0 = t_0 < t_1 < \ldots < t_{l-1}<t_l = T$, $z_0 < z_1 < \ldots <  z_{k}$
be a partition of $[0,T]\times ([1/m,m] \cup [-m,-1/m])$, a compact included in $\bR_0$.
We consider the following step càglàd process 
$$ \theta^{l,k}_t(z):= \sum_{i,j=0}^{l,k} \theta_{t_i}(z_j)\bOne_{\{(t,z)\in[t_{i},t_{i+1})\times[z_j,z_{j+1})\}}\ .$$
By linearity, an application of  Lemma~\ref{lem.wick.ord.mall.pois} implies that
%and we are going to show that \eqref{pois.forw.sko.mall} holds true,
\begin{align*}
    \int_0^T\int_{\bR_0} \theta^{l,k}_t(z) \tilde N(d^{-}t,dz) %\lim_{m\to\infty} \int_0^T\int_{\bR_0} \theta^{l,k}_t(z)\bOne_{\{z\in U_m\}} \tilde N(dt,dz)
    %\\
    %=& \lim_{m\to\infty}\sum_{i,j=0}^{l,k} \int^{t_{i+1}}_{t_i}\int_{[z_{2j},z_{2j+1}]\cap U_m} \theta^{l,k}_{t_{i}}(z_{2j}) \tilde N(dt,dz)\\
    %=& \lim_{m\to\infty}\sum_{i,j=0}^{l,k} %\int_{t_{i-1}}^{t_i}\int_{z_{2j-1}}^{z_{2j}}
    %\theta^{l,k}_{t_{i}}(z_{2j}) \left(\tilde N(t_{i+1},[z_{2j},z_{2j+1}]\cap U_m)-\tilde N(t_{i},[z_{2j},z_{2j+1}]\cap U_m)\right)\\
    %=& \lim_{m\to\infty}\sum_{i,j=0}^{l,k} \int^{t_{i+1}}_{t_i} \int_{[z_{2j},z_{2j+1}]\cap U_m} \left(\theta^{l,k}_{t_{i}}(z_{2j})+D_{u,z} \theta^{l,k}_{t_{i}}(z_{2j})\right)\diamond \tilde\bN(u,z) +  D_{u,z}\theta^{l,k}_{t_{i}}(z_{2j})\,\nu(dz)du \\
    %=& \lim_{m\to\infty} \int_{0}^{T} \int_{U_m} \left(\theta^{l,k}_{t}(z)+D_{u,z} \theta^{l,k}_{t}(z)\right)\diamond \tilde\bN(u,z)\,+  D_{u,z}\theta^{l,k}_{t}(z)\,\,\nu(dz) du \\
    =& \int_{0}^{T} \int_{\bR_0} \left(\theta^{l,k}_{t}(z)+D_{t,z} \theta^{l,k}_{t}(z)\right)\diamond \tilde\bN(t,z)\,+  D_{t,z}\theta^{l,k}_{t}(z)\,\,\nu(dz) dt \ .
\end{align*}
We complete the proof by letting $m\to\infty$  and using an approximation argument by adapting Lemma 6.7 in \cite{DiNunnoOksendalProske2009} to the Poisson case.
\end{proof}
\subsection{After Default: Logarithmic case}\label{subsec:after.default.log}
In this subsection we restrict our computations to the logarithmic utility,
i.e. $U(x) = \ln x$ that allows to obtain more explicit results.
In particular, we analyze the after default problem, which
is stated as 
\begin{equation}\label{optimization.problem.after.log}
    \bV(\eta,x) =\esssup _{\pi \in \mathcal{A}(\tau)} \EE^{\eta,x}\left[\ln X^{\pi}_T(\tau)\right]
    \ ,\quad (\eta,x)\in [0,T]\times \bR^+
\end{equation}
where the random variable $\ln X^{\pi}_T(\tau)$ has explicit representation derived from
\begin{align}\label{X.log.after}
\ln \frac{X^{\pi}_T(\tau)}{X^{\pi}_\tau(\tau)} =& \int_{\tau}^{T} \rho_s(\tau) + \pi_s(\tau) (\mu_s(\tau)-\rho_s(\tau)) - \frac{1}{2}\pi^2_s(\tau)\sigma^2_s(\tau)\, ds\notag \\
&+ \int_\tau^{T} \int_{\bR_0} \ln(1+\pi_s(\tau)\theta_s(z,\tau))-\pi_s(\tau)\theta_s(z,\tau) \,\nu(dz) ds\notag\\ 
&+\int_\tau^{T}\pi_s(\tau)\sigma_s(\tau)d^-W_s+ \int_\tau^{T} \int_{\bR_0} \ln(1+\pi_s(\tau)\theta_s(z,\tau)) \tilde N(d^-s,dz)\ .
\end{align}
The stochastic integrals with random intervals are treated as follows,
$$ \int_\tau^{T}\pi_s(\tau)\sigma_s(\tau)d^-W_s = \int_0^{T}\bOne_{\{\tau \leq s\}}\pi_s(\tau)\sigma_s(\tau)d^-W_s\ . $$
In the next proposition we solve the optimization problem \eqref{optimization.problem.after.log} using the anticipating calculus described in Subsection~\ref{subsec:white.noise}.
\begin{Proposition}\label{prop.opt.portfolio.after}
%Suppose t
There exists a local maximum for the problem~\eqref{optimization.problem.after.log} with initial condition $(\eta, x) \in[0, T] \times(0, \infty)$
if and only if
it satisfies, %the following condition 
for $s\in[\eta, T]$, 
\begin{align}\label{gen.opt.condition.log.after}
    \pi_s(\eta)&\sigma^2_s(\eta) = \mu_s(\eta)-\rho_s(\eta) + \EE\left[D_{s^+}\sigma_s(\eta)\,|\,\cG_s\right]+ \sigma_s(\eta)\diamond\bW_s \notag \\
    &+ \int_{\bR_0} \frac{-\pi_s(\eta)\theta^2_s(z,\eta)}{1+\pi_s(\eta)\theta_s(z,\eta)}  +\EE\left[D_{s^+,z}\frac{\theta_s(z,\eta)}{1+\pi_s(\eta)\theta_s(z,\eta)}\,|\,\cG_s\right] \nu(dz)\notag\\
    &+\int_{\bR_0}\left(\frac{\theta_s(z,\eta)}{1+\pi_s(\eta)\theta_s(z,\eta)}+\EE\left[D_{s^+,z}\frac{\theta_s(z,\eta)}{1+\pi_s(\eta)\theta_s(z,\eta)}\,|\,\cG_s\right]\right) \diamond\tilde\bN(s,z) \nu(dz) \ .
\end{align}
\end{Proposition}
\begin{proof}
We use a perturbation argument as in Theorem 16.20 of \cite{DiNunnoOksendalProske2009}.
We define 
$$G_{\eta,x}(\pi):=\EE^{\eta,x} \left[\ln \frac{X^{\pi}_T(\tau)}{X^{\pi}_{\tau}(\tau)} \right]$$ 
and aim to optimize the function $\pi\longrightarrow G_{\eta,x}(\pi)$ with $(\eta,x)$ fixed.
Let $\beta\in\cA(\tau)$ bounded and consider $\delta > 0$ such that $\pi+\epsilon\beta\in\cA(\tau)$ for $\epsilon\in (-\delta,\delta)$
and we consider the function $I(\epsilon):=G_{\eta,x}(\pi+\epsilon\beta)$.
%We short the notation by $\EE^{\eta,x}[\cdot] = \EE[\cdot|X_\tau^\pi(\tau)=x,\ \tau = \eta]$.
A local maximum strategy must satisfy the following first order condition,
\begin{align*}
    0 =& I'(0) = \lim_{\epsilon\to 0}\frac{I(\epsilon)-I(0)}{\epsilon} =  \lim_{\epsilon\to 0}\frac{G_{\eta,x}(\pi+\epsilon\beta)-G_{\eta,x}(\pi)}{\epsilon} \\
    =& \EE^{\eta,x}\left[\int_{\tau}^T\beta_s(\mu_s(\tau)-\rho_s(\tau))-\beta_s\pi_s(\tau)\sigma^2_s(\tau)  - \int_{\bR_0}\frac{\beta_s\pi_s(\tau)\theta^2_s(z,\tau)}{1+\pi_s(\tau)\theta_s(z,\tau)} \nu(dz)\,ds\right]\\
    &+\EE^{\eta,x}\left[ \int_{\tau}^T\beta_s\sigma_s(\tau)d^-W_s + \int_\tau^T\int_{\bR_0}\frac{\beta_s\theta_s(z,\tau)}{1+\pi_s(\tau)\theta_s(z,\tau)}\tilde N(d^-s,dz)\right]\ .
\end{align*}
We take $\beta_s = \xi_t\bOne_{\{t < s \leq t+h\}}$ for $\eta\leq t\leq T-h$ and $\xi_t$ any bounded $\cG_t$-measurable random variable.
Then the condition is written as,
\begin{align*}
    0 =& \EE^{\eta,x}\left[\xi_t\int_{t}^{t+h}\mu_s(\tau)-\rho_s(\tau)-\pi_s(\tau)\sigma^2_s(\tau) - \int_{\bR_0} \frac{\pi_s(\tau)\theta^2_s(z,\tau)}{1+\pi_s(\tau)\theta_s(z,\tau)}  \nu(dz)\,ds \right]\\
    &+\EE^{\eta,x}\left[\xi_t\left( \int_{t}^{t+h}\sigma_s(\tau)d^-W_s + \int_t^{t+h}\int_{\bR_0}\frac{\theta_s(z,\tau)}{1+\pi_s(\tau)\theta_s(z,\tau)}\tilde N(d^-s,dz) \right)\right]
\end{align*}
which, as $\xi_t$ is general, is equivalent to
\begin{align}\label{gen.opt.condition.log.after.conversely}
    0 =& \EE^{\eta,x}\left[\int_{t}^{t+h}\mu_s(\tau)-\rho_s(\tau)-\pi_s(\tau)\sigma^2_s(\tau) -\int_{\bR_0} \frac{\pi_s(\tau)\theta^2_s(z,\tau)}{1+\pi_s(\tau)\theta_s(z,\tau)} \nu(dz)\,ds |\cG_t \right]\notag\\
    &+\EE^{\eta,x}\left[ \int_{t}^{t+h}\sigma_s(\tau)d^-W_s + \int_t^{t+h}\int_{\bR_0}\frac{\theta_s(z,\tau)}{1+\pi_s(\tau)\theta_s(z,\tau)}\tilde N(d^-s,dz) |\cG_t \right]\ .
\end{align}
By applying Theorems \ref{theo.forw.skor.mall.BM} and \ref{theo.forw.sko.mall.pois} we rewrite the stochastic integrals as $dt$-integral and $(\nu\times dt)$-integral in $(\cS)^*$ respectively, 
and by rearranging all the terms, we get the condition~\eqref{gen.opt.condition.log.after}.

Conversely, it is enough to see that by assuming \eqref{gen.opt.condition.log.after} we recover \eqref{gen.opt.condition.log.after.conversely} and the perturbation argument holds true because of the concavity, i.e.,
$I''(\epsilon) < 0$ for $\epsilon \in (-\delta, \delta)$.
\end{proof}
\begin{Remark}
    The domain $\bD_{1,2}$ of the Malliavin derivative in general does not contain the indicator random variables, see~\cite{SUZUKI2020108614},
    which will play a prominent role in Theorem~\ref{theo.opt.log.bef} and Examples~\ref{ex.argmax} and~\ref{ex.aksamit.half-final} below.
    This is why we need to extend the framework of our paper and consider the Malliavin derivative, the Skorohod integral and the forward integral in the space of Hida-Malliavin distributions,
    as we have done in Subsection~\ref{subsec:white.noise}.
    In this context, the white noises for the Brownian motion and the Poisson random measure arise naturally, as we have shown in Proposition~\ref{prop.opt.portfolio.after} above.
    Note that Lemma 5.1 of \cite{Bermin02} guarantees that the conditional expectation of the Malliavin derivative of a $L^2(\PP)$-random variable is well defined as an element of $L^2(dt\times\PP)$.
    In addition, using the computations appearing in the proof of 
    Lemma 2.5.7 of \cite{oksendal2010book} 
    we can assure the well definition of the Wick product of an $L^2(dt\times\PP)$ process with respect to the white noises.
\end{Remark}
\subsection{Before Default: Logarithmic Case}\label{subsec:before.default.log}
We remind that the random variable $X^{\pi,\bF}_{\tau-}$ 
has explicit solution as follows,
\begin{align}\label{def.X.bef.stop}
\ln \frac{X^{\pi,\bF}_{\tau-}}{x_0} =& \int_0^{T}\bOne_{\{\tau > s\}}\left( \rho^{\bF}_s + \pi^{\bF}_s (\mu^{\bF}_s-\rho^{\bF}_s) - \frac{1}{2}\left(\pi^{\bF}_s\sigma^{\bF}_s\right)^2\right)\, ds\notag \\
&+ \int_0^{T} \int_{\bR_0} \bOne_{\{\tau > s\}}\left(\ln(1+\pi^{\bF}_s\theta^{\bF}_s(z))-\pi^{\bF}_s\theta^{\bF}_s(z) \right)\,\nu(dz) ds\notag\\ 
&+\int_0^{T}\bOne_{\{\tau > s\}}\pi^{\bF}_s\sigma^{\bF}_sd^-W_s + \int_0^{T} \int_{\bR_0} \bOne_{\{\tau > s\}} \ln(1+\pi^{\bF}_s\theta^{\bF}_s(z)) \tilde N(d^-s,dz)\ .
\end{align}
Considering the logarithmic utility, Proposition \ref{after-and-before} can be simplified as follows.
\begin{Corollary}\label{eq.HJB.after-before.log}
The optimization problem \eqref{optimization.problem}, 
under $U(x)=\ln x$ and the information flow~$\bG$,
can be written as follows
\begin{align}\label{eq.HJB.after-before.revisited}
   \bV_T^{\bG} = \EE^{x_0}[h_\tau H_T] + \sup _{\pi^{\bF}\in\cA(\bF)} &\EE^{x_0}\left[\ln X^{\pi,\bF}_{\tau-} +\ln\left(1+\pi^{\bF}_\tau \kappa_\tau\right) H_T\right]
\end{align}
for some process $h$ with $h_T = 0$ $\PP$-a.s.
\end{Corollary}
\begin{proof}
Let $\pi(\tau)$ the local maximum strategy defined in \eqref{gen.opt.condition.log.after} and we consider the optimal value $\bV_T^\bG$ combined with \eqref{other.optimal.functional}, then,
\begin{align*}
    \bV_T^\bG =&\sup _{\pi^{\bF}\in\cA(\bF)} \EE^{x_0}\left[\ln \left(X^{\pi,\bF}_T\right)(1-H_T) + \EE\left[\ln X^{\pi}_T(\tau) |\tau,X^{\pi}_\tau(\tau)\right]H_T\right] \\
    =&\sup _{\pi^{\bF}\in\cA(\bF)} \EE^{x_0}\left[\ln \left(X^{\pi,\bF}_T\right)(1-H_T) + \ln\left(X^{\pi}_\tau(\tau)\right)H_T + \EE\left[\ln\frac{X^{\pi}_T(\tau)}{X^{\pi}_\tau(\tau)} \Big|\tau,X^{\pi}_\tau(\tau)\right]H_T\right] \\
    =&\EE^{x_0}\left[\EE\left[\widetilde{h}_\tau(T) |\tau, X^{\pi}_\tau(\tau) \right] H_T\right]\\ 
    &+ \sup _{\pi^{\bF}\in\cA(\bF)} \EE^{x_0}\left[\ln \left(X^{\pi,\bF}_T\right)(1-H_T) + \ln\left(
   X^{\pi,\bF}_{\tau-}\left(1+\pi^{\bF}_\tau \kappa_\tau\right)\right)H_T\right] 
\end{align*}
where by $\widetilde{h}_\tau(T)$ we denote the right-hand side of \eqref{X.log.after} with the local maximum strategy $\pi(\tau)$.
Note that by the integral expression 
$\widetilde{h}_T(T) = 0$
and we define
$h_\tau:=\EE\left[\widetilde{h}_\tau(T) |\tau,X^{\pi}_\tau(\tau)\right]$.
Finally, we point out that the first two terms can be expressed as follows,
   \begin{equation*}
        \ln \left(X^{\pi,\bF}_T\right)(1-H_T)  + \ln\left(X^{\pi,\bF}_{\tau-}\right)H_T =
        \ln X^{\pi,\bF}_{\tau-}\ ,%  \ln x_0 + \int_0^{T}\bOne_{\{\tau > s\}}\left( \rho^{\bF}_s + \pi^{\bF}_s (\mu^{\bF}_s-\rho^{\bF}_s) - \frac{1}{2}\left(\pi^{\bF}_s\sigma^{\bF}_s\right)^2\right)\, ds\notag \\
%&+ \int_0^{T} \int_{\bR_0} \bOne_{\{\tau > s\}}\left(\ln(1+\pi^{\bF}_s\theta^{\bF}_s(z))-\pi^{\bF}_s\theta^{\bF}_s(z) \right)\,\nu(dz) ds\notag\\ 
%&+\int_0^{T}\bOne_{\{\tau > s\}}\pi^{\bF}_s\sigma^{\bF}_sd^-W_s + \int_0^{T} \int_{\bR_0} \bOne_{\{\tau > s\}} \ln(1+\pi^{\bF}_s\theta^{\bF}_s(z)) \tilde N(d^-s,dz)\ .
   \end{equation*}
   where the right-hand side is defined in \eqref{def.X.bef.stop}
   and the result follows.
\end{proof}
%We are going to rewrite the terms appearing in \eqref{eq.HJB.after-before.revisited} in order to apply a perturbation argument again.
%Thanks to Lemma \ref{H.decomposition}, the following expectation is known,
%\begin{equation}\label{eq.expect.int.H}
%\EE^{x_0}\left[ m_\tau H_T\right] = \EE^{x_0}\left[\int_0^T m_s dH_s\right] = \EE^{x_0}\left[ \int_0^T   \frac{\bOne_{\{\tau > s\}}\,m_s}{Z_{s-}}dA^{\tau}_s \right]\ ,
%\end{equation} for every $\bG$-adapted and finite process $m=\prT{m}$, so we will use the most useful expression in each step.
In the next theorem we solve the before default optimization problem under the logarithmic utility.
We characterize the existence of a local maximum by the equation \eqref{gen.opt.condition.log.bef} and the intensity hypothesis~\eqref{approach.intensity} on the default process.
\begin{Theorem}\label{theo.opt.log.bef}
There exists a local maximum portfolio %\pi^\bF$ 
for the problem~\eqref{eq.HJB.after-before.revisited}
if and only if
%Then $\pi^\bF$ 
it satisfies the condition
\begin{align}\label{gen.opt.condition.log.bef}
    \pi^\bF_s\left(\sigma^\bF_s\right)^2 =& \mu^\bF_s-\rho^\bF_s -  \int_{\bR_0}\frac{\pi^\bF_s(\theta^\bF_s(z))^2}{1+\pi^\bF_s\theta^\bF_s(z)}\nu(dz) +\frac{\EE\left[D_{s^+}\bOne_{\{\tau > s\}}\,|\cF_s\,\right]}{Z_{s-}}\sigma^\bF_s\notag\\
    &+\frac{\EE\left[D_{s^+,z}\bOne_{\{\tau > s\}}\,|\,\cF_s\right]}{Z_{s-}}\int_{\bR_0} \frac{\theta^\bF_s(z)} {1+\pi^\bF_s\theta^\bF_s(z)}\nu(dz) +  \frac{\kappa_s}{1+\pi^{\bF}_s\kappa_s}\frac{1}{Z_{s-}}\lambda_s
\end{align}
and the intensity hypothesis \eqref{approach.intensity} holds true on $\tau$ with intensity process $\lambda$.
\end{Theorem}
\begin{proof}
We proceed as in the proof of Proposition \ref{prop.opt.portfolio.after}, 
applying a perturbation argument in order to achieve a condition for local optimality.
By applying Corollary~\ref{eq.HJB.after-before.log}, 
we redefine the functional
$$G_{x_0}(\pi) := \EE^{x_0}\left[\ln
   X^{\pi,\bF}_\tau+ \ln
  \left( 1+\pi^{\bF}_\tau \kappa_\tau\right)H_T\right]\ .$$
Let $\beta\in\cA(\bF)$ bounded and consider $\delta > 0$ such that $\pi^\bF+\epsilon\beta\in\cA(\bF)$ for $\epsilon\in (-\delta,\delta)$. 
We consider the function $I(\epsilon):=G_{x_0}(\pi^\bF+\epsilon\beta)$.
%In order to short the notation in the integrands, we define the following $\cF_T$-measurable random variable, $\varphi_s:= (1-H_T) + H_T\bOne_{\{\tau > s\}}$.
A local maximum strategy must satisfy the condition $0 = I'(0)$, which is equivalent to,
\begin{align*}
    0 =& I'(0) = \lim_{\epsilon\to 0}\frac{I(\epsilon)-I(0)}{\epsilon} =  \lim_{\epsilon\to 0}\frac{G_{x_0}(\pi^\bF+\epsilon\beta)-G_{x_0}(\pi^\bF)}{\epsilon} \\
    =& \EE^{x_0}\left[\int_0^T\bOne_{\{\tau > s\}}\beta_s\left(\mu^\bF_s-\rho^\bF_s-\pi^\bF_s(\sigma^\bF_s)^2 - \int_{\bR_0}\frac{\pi^\bF_s(\theta^\bF_s(z))^2}{1+\pi^\bF_s\theta^\bF_s(z)}\nu(dz)\right) ds\right]\\
    +&\EE^{x_0}\left[ \int_{0}^T\bOne_{\{\tau > s\}}\beta_s\sigma^\bF_sd^-W_s + \int_0^T \int_{\bR_0}\frac{\bOne_{\{\tau > s\}}\beta_s\theta^\bF_s(z)}{1+\pi^\bF_s\theta^\bF_s(z)}\tilde N(d^-s,dz) + \frac{\beta_\tau\kappa_\tau}{1+\pi^\bF_\tau\kappa_\tau}H_T \right] \ . 
\end{align*}
By applying Theorems \ref{theo.forw.skor.mall.BM} and \ref{theo.forw.sko.mall.pois}, we compute the expectation of the forward integrals as follows
\begin{align*}
     \EE^{x_0}\left[ \int_t^{t+h}\bOne_{\{\tau > s\}}\beta_s\sigma^\bF_sd^-W_s\right] &=\EE^{x_0}\left[ \int_t^{t+h}(D_{s^+}\bOne_{\{\tau > s\}})\beta_s\sigma_s^\bF ds \right]\\
     \EE^{x_0}\left[\int_t^{t+h} \int_{\bR_0}\frac{\bOne_{\{\tau > s\}}\beta_s\theta^\bF_s(z)}{1+\pi^\bF_s\theta^\bF_s(z)}\tilde N(d^-s,dz)\right] &=  \EE^{x_0}\left[\int_t^{t+h} \int_{\bR_0}D_{s^+,z}\bOne_{\{\tau > s\}}\frac{\beta_s\theta^\bF_s(z)\nu(dz)ds}{1+\pi^\bF_s\theta^\bF_s(z)}\right]
\end{align*}
because $\beta,\sigma^{\bF},\theta^\bF,\pi^\bF$ are $\bF$-adapted processes and therefore their Malliavin derivative is null.
Moreover, thanks to the $\bF$-dual predictable projection we can rewrite
%Lemma \ref{H.decomposition} we can rewrite
\begin{equation}\label{eq.expect.int.H.bis}
    \EE^{x_0}\left[\frac{\beta_\tau\kappa_\tau}{1+\pi^\bF_\tau\kappa_\tau}H_T \right] = \EE^{x_0}\left[ \int_0^T   \frac{\beta_s\kappa_s}{1+\pi^\bF_s\kappa_s} %\frac{\bOne_{\{\tau > s\}}}{Z_{s-}}
    dA^{\tau}_s \right]\ .
\end{equation}
We take $\beta_s = \xi_t\bOne_{\{t < s \leq t+h\}}$ for $\xi_t$ any bounded $\cF_t$-measurable random variable.
Then the condition is written as,
\begin{align}\label{eq.cond.exp.bef}
    0 =& \EE^{x_0}\left[\int_t^{t+h}\bOne_{\{\tau > s\}} \left(\mu^\bF_s-\rho^\bF_s-\pi^\bF_s(\sigma^\bF_s)^2 - \int_{\bR_0}\frac{\pi^\bF_s(\theta^\bF_s(z))^2}{1+\pi^\bF_s\theta^\bF_s(z)} \nu(dz)\right)ds \right.\notag\\
    &+\left.  \int_t^{t+h}(D_{s^+}\bOne_{\{\tau > s\}})\sigma_s^\bF ds+\int_t^{t+h} \int_{\bR_0}(D_{s^+,z}\bOne_{\{\tau > s\}})\frac{\theta^\bF_s(z)}{1+\pi^\bF_s\theta^\bF_s(z)}\,\nu(dz)ds \right.\notag\\
    &+\left.\int_t^{t+h}  \frac{\kappa_s}{1+\pi^\bF_s\kappa_s}
    %\frac{\bOne_{\{\tau > s\}}}{Z_{s-}}
    dA^{\tau}_s 
    %\right. +\left.\int_t^{t+h}  \frac{\kappa_s}{1+\pi^\bF_s\kappa_s}
    %\frac{\bOne_{\{\tau > s\}}}{Z_{s-}}
%    dA^{\tau,\perp}_s
    \vert\cF_t\right]\ .
%    &+\EE^{x_0}\left[\int_0^T   \bOne_{\{s\leq \tau\}}\frac{\kappa_s}{1+\pi^{\bF}_s\kappa_s}\frac{1}{Z_{s-}}dA^{\tau}_s\vert\cF_t \right]
\end{align}
The $\bF$-dual predictable projection admits the following representation, see Lemma 2.1 of \cite{FontanaThorsten}, 
\begin{equation}
    A^{\tau}_t = \int_0^t \lambda_sds + A^{\tau,\perp}_t 
    + \sum_{0\leq s\leq t}\Delta A_s^\tau\ ,
\end{equation}
where $\lambda$ is an $\bG$-predictable integrable process,
$A^{\tau,\perp}=\prT{A^{\tau,\perp}}$ is an increasing and continuous process such that $dA^{\tau,\perp}_s \perp ds$. 
%\textcolor{red}{Note that Lemma~1.48 of~\cite{Aksamit2017} guarantees the continuity of $A^\tau$.}
%where $A^{\tau,ac}$ denotes the absolutely continuous part of $A^\tau$, by applying the same reasoning as in Theorem 3.8 of \cite{OksendalBiagini2005}.
Then, %for any $h$ such that $t+h\leq T$, we get
%We reason as in Theorem 3.8 of \cite{OksendalBiagini2005} and we take the absolutely continuous part with respect to the Lebesgue measure of $A^\tau$.
%Note that $\bOne_{\{\tau > s\}} = \bOne_{\{\tau > s\}}$ and we can simplify the terms as follows,
\begin{align}\label{eq.aux.opt.tau.bef}
    0 =& \EE^{x_0}\left[\int_t^{t+h}\bOne_{\{\tau > s\}}\left(\mu^\bF_s-\rho^\bF_s-\pi^\bF_s(\sigma^\bF_s)^2 - \int_{\bR_0}\frac{\pi^\bF_s(\theta^\bF_s(z))^2}{1+\pi^\bF_s\theta^\bF_s(z)} \nu(dz)  \right)\right.\notag\\
    &+\left.  \int_{\bR_0}(D_{s^+,z}\bOne_{\{\tau > s\}})\frac{\theta^\bF_s(z)}{1+\pi^\bF_s\theta^\bF_s(z)}\,\nu(dz)\, + (D_{s^+}\bOne_{\{\tau > s\}})\sigma_s^\bF +  \frac{\kappa_s\, \lambda_s}{1+\pi^\bF_s\kappa_s}\,ds
     \right.\notag\\
    &+\left. \int_t^{t+h}\frac{\kappa_s}{1+\pi^\bF_s\kappa_s}
    dA^{\tau,\perp}_s   
    +\sum_{t\leq s \leq t+h} \frac{\kappa_s}{1+\pi^\bF_s\kappa_s} \Delta A_s^\tau  
    \vert\cF_t\right]\ .
\end{align}
We define the following $\bF$-adapted process $Y=\prT{Y}$,
\begin{align}
    Y_t :=& \int_0^t \EE[\bOne_{\{\tau > s\}}\vert\cF_t] \left(\mu^\bF_s-\rho^\bF_s-\pi^\bF_s(\sigma^\bF_s)^2 - \int_{\bR_0}\frac{\pi^\bF_s(\theta^\bF_s(z))^2}{1+\pi^\bF_s\theta^\bF_s(z)} \nu(dz)  \right)\notag\\
    &+ \int_{\bR_0}\EE[D_{s^+,z}\bOne_{\{\tau > s\}}\vert\cF_t]\frac{\theta^\bF_s(z)}{1+\pi^\bF_s\theta^\bF_s(z)}\,\nu(dz)\, + \EE[D_{s^+}\bOne_{\{\tau > s\}}\vert\cF_t]\sigma_s^\bF
    \notag\\
    &+  \frac{\kappa_s\, \lambda_s}{1+\pi^\bF_s\kappa_s}\,ds + \int_0^{t}\frac{\kappa_s}{1+\pi^\bF_s\kappa_s}
    dA^{\tau,\perp}_s
    +\sum_{0\leq s \leq t} \frac{\kappa_s}{1+\pi^\bF_s\kappa_s} \Delta A_s^\tau  
   \ ,
\end{align}
which is a finite variation process and, by~\eqref{eq.aux.opt.tau.bef}, 
it is an $\bF$-martingale. 
So we conclude that it is a null process, for the absolutely continuous, the singular and the purely discontinuous parts.
In particular, as $\kappa_s\neq0$, we get that 
$A^{\tau,\perp}_s=\Delta A^\tau_s = 0$
and \eqref{approach.intensity} holds true.

Conversely, we assume \eqref{gen.opt.condition.log.bef} and \eqref{approach.intensity}.
Then we recover \eqref{eq.cond.exp.bef} and the perturbation arguments holds true by the concavity property of the functional~$I$, i.e., $I''(0)<0$.
\end{proof}
In the following example we illustrate some implications of Theorem~\ref{theo.opt.log.bef}. 
In particular, we show a random time for which \eqref{approach.intensity} is not satisfied as the compensator of the process $H$ in the enlarged filtration $\bG$ results singular to the Lebesgue measure. It follows that the optimal portfolio among the admissible strategies  $\cA(\bG)$ does not exist.
\begin{Example}\label{ex.argmax}
We consider the following example with
\begin{equation}\label{random.time.arg.max}
    \tau = \argmax_{s\in[0,T]} W_s\ , 
\end{equation}
that is the time in which the Brownian motion reaches it maximum before the time horizon.
This case was previously studied for example in \cite{Imkeller02}, but, for the best of our knowledge, it was never been considered in the optimal portfolio problem with progressive enlargement and defaultable framework.
We denote by 
$$M_{s,t}:=\sup_{u\in[s,t]} W_u\ ,$$
and define $M_s:=M_{0,s}$.
It is clear that 
$\{\tau > s\} = \{ M_s < M_{s,T} \}$ 
and so we can compute the Malliavin trace of the second term.
Note that, as the random time depends only on the Brownian motion,
we have $D_{s^+,z}\bOne_{\{\tau > s\} } = 0$ for all $(s,z)$.
\begin{Lemma}
The Malliavin trace of $\bOne_{\{\tau > s\}}$ is 
    \begin{equation}\label{ex.malliavin.derivative}
        D_{s^+}\bOne_{\{\tau>s\}} = \delta_{M_s}(M_{s,T}) \bOne_{\{M_{s,T}>W_s\}}\ .
    \end{equation}
Moreover, its conditional expectation is computed as 
\begin{equation}\label{ex.malliavin.derivative.cond.exp}
     \EE\left[ D_{s^+}\bOne_{\{\tau>s\}} |\cF_s \right] = \frac{2 \exp\left(-\frac{(M_s-W_s)^2}{2(T-s)}\right)}{ \sqrt{2\pi (T-s)}}\ . 
\end{equation}
\end{Lemma}
\begin{proof}
    By applying Corollary~5.3 of \cite{Bermin02}, as it is argued in~Example 5.3 of the same reference, we have
    $$ D_{s^+}\bOne_{\{\tau>s\}} = \delta_{M_s}(M_{s,T})  D_{s^+}M_{s,T}\ . $$
    So we need to compute the trace of $M_{s,T}$.
    In Proposition 2.1.10 of~\cite{nualart2006malliavin} it is proved that the Malliavin derivative of $M_{s,T}$ exists and it is computed as 
    $D_{s^+}M_{s,T} = \bOne_{\{M_{s,T}>W_s\}},$
    and we obtain~\eqref{ex.malliavin.derivative}. 
    For the conditional expectation, we get
    \begin{align*}  
    \EE\left[ D_{s^+}\bOne_{\{\tau>s\}} |\cF_s \right] &=  \EE\left[  \delta_{M_s}(M_{s,T}) \bOne_{\{M_{s,T}>W_s\}} |\cF_s \right] \\
%    &= \EE\left[ \EE\left[ \delta_{M_s}(M_{s,T}) \bOne_{\{M_{s,T}>W_s\}}|\cG_s \right] |\cF_s \right] \\
%    &= \frac{1}{\PP(\tau > s|\cF_s)}\EE\left[ \bOne_{\{\tau>s\}} \delta_{M_s}(M_{s,T}) \bOne_{\{M_{s,T}>W_s\}} |\cF_s \right] \\
    &=  %\frac{1}{\PP(\tau > s|\cF_s)} 
    \int_{W_s}^{\infty} \delta_{M_{s}}(m) f_s(m)dm =f_s(M_{s})\ ,%\frac{f(M_{s})}{\PP(\tau > s|\cF_s)} \ ,
\end{align*}
being $f_s$ the density of the random variable $M_{s,T}$ given $\cF_s$, which is equivalent to consider the variable $M_{T-s}$ in the domain $(W_s,+\infty)$, and the result holds true.
\end{proof}
In this example, we proceed by assuming that there exists an admissible local maximum for the optimization problem with logarithmic utility and, as we are going to prove that the intensity hypothesis \eqref{approach.intensity} does not hold, we will get a contradiction.
We carry out the perturbation argument to derive the condition for the local maximum strategy
by mimicking the proof of Theorem~\ref{theo.opt.log.bef}.
With respect to the term of the default process in the optimal strategy,
we take the left-hand side of \eqref{eq.expect.int.H.bis} and we rewrite it by using Lemma 3.2 of \cite{LiRutkowski2014}. We can apply this result because $\tau$ is a Honest time, see~\cite{Imkeller02}.
% 7.4.1.2 of \cite{Jeanblanc2009mathematical},
\begin{align*}
    \EE\left[ \frac{\beta_\tau\kappa_\tau}{1+\pi_\tau^\bF\kappa_\tau} H_T \right] &= \EE\left[ \EE\left[ \frac{\beta_\tau\kappa_\tau}{1+\pi_\tau^\bF\kappa_\tau} H_T |\cG_t \right]\right] \\
    &=\EE\left[\frac{\beta_\tau\kappa_\tau}{1+\pi_\tau^\bF\kappa_\tau}\bOne_{\{\tau\leq t\}} - \bOne_{\{\tau > t\}} \frac{1}{Z_t}\EE\left[\int_t^T \frac{\beta_s\kappa_s}{1+\pi_s^\bF\kappa_s} dZ_s\vert\cF_t\right]\right]\ ,
\end{align*}
and, as in the proof of Theorem \ref{theo.opt.log.bef}, 
we consider $\beta_s = \xi_t\bOne_{\{t<s\leq t+h\}}$ for~$\xi_t$ any bounded \hbox{$\cF_t$-measurable} random variable.
Note that, in the first term we have 
$\bOne_{\{\tau\leq t\}}\bOne_{\{ t<\tau\leq t+h \}} = 0$.
Then, by taking $\cF_t$-conditional expectation we get,
\begin{align*}
     \EE\left[ \frac{\beta_\tau\kappa_\tau}{1+\pi_\tau^\bF\kappa_\tau} H_T \right] &= \EE\left[\EE\left[- \xi_t \frac{ \bOne_{\{\tau > t\}}}{Z_t}\EE\left[\int_t^{t+h} \frac{\kappa_s}{1+\pi_s^\bF\kappa_s} dZ_s\vert\cF_t\right]\vert\cF_t\right]\right]\\
     &= -\EE\left[ \xi_t\frac{\EE\left[\bOne_{\{\tau > t\}} \vert\cF_t\right]}{Z_t}\EE\left[\int_t^{t+h} \frac{\kappa_s}{1+\pi_s^\bF\kappa_s} dZ_s\vert\cF_t\right]\right]\\
    &= -\EE\left[\xi_t\int_t^{t+h} \frac{\kappa_s}{1+\pi_s^\bF\kappa_s} dZ_s\right]\ ,
\end{align*}
where we have used that $\EE\left[\bOne_{\{\tau > t\}} \vert\cF_t\right] = \PP(\tau > t\vert \cF_t) = Z_t$.
Then, looking at the right-hand side of \eqref{eq.expect.int.H.bis} we have proved that
$$  \EE\left[ \xi_t\int_t^{t+h}   \frac{\kappa_s}{1+\pi^\bF_s\kappa_s}
%\frac{\bOne_{\{\tau > s\}}}{Z_{s-}}
dA^{\tau}_s \right] = -\EE\left[\xi_t\int_t^{t+h} \frac{\kappa_s}{1+\pi_s^\bF\kappa_s} dZ_s\right] $$
for any $\xi_t\in\cF_t$ bounded, so we rewrite the default-term appearing in \eqref{eq.cond.exp.bef} as
\begin{align}\label{ex.eq.compensator.Z}
    \EE\left[\int_t^{t+h}  \frac{\kappa_s}{1+\pi^\bF_s\kappa_s}
    %\frac{\bOne_{\{\tau > s\}}}{Z_{s-}}
    dA^{\tau}_s\vert\cF_t\right] &= - \EE\left[\int_t^{t+h} \frac{\kappa_s} {1+\pi_{s}^{\bF}\kappa_s} dZ_s \vert\cF_t\right]
    \\ 
\label{ex.eq.compensator.Z.2}
&=  \sqrt{\frac{2}{\pi}}\EE\left[\int_t^{t+h} \frac{\kappa_s} {1+\pi_{s}^{\bF}\kappa_s}\frac{dM_s}{\sqrt{T-s}}  \vert\cF_t\right]\ ,
\end{align}
where in the last equality we used Section 5.6.6 of \cite{Jeanblanc2009mathematical}.
Since $dM_s\perp ds$, see for example page 522 in \cite{Obloj2006},
\eqref{approach.intensity} is not satisfied.
Then, the agent who is playing with the filtration $\bG$ can take advantage of the increments of $M$ in order to maximize her expected utility and we find no optimal strategy in $\cA(\bG)$. 
See \cite{FontanaJeanblanc14} for a detailed explanation of arbitrage and Honest times.

However, we can proceed as in the after default 
with the $(\cS)^*$ approach
and try to look an optimal strategy
by giving an explicit expression for the right-hand side of \eqref{ex.eq.compensator.Z}.
We define the process $K_s:=M_s-W_s$ 
%sometimes named \emph{rally} of $W$ in the literature. %, see~\cite{drawdown2004}.
and we consider the expression appearing in Example~4.1.7.5 of~\cite{Jeanblanc2009mathematical},
$$  g(s,K_s):= Z_s = \frac{2}{\sqrt{2\pi}}\int_{K_s/\sqrt{T-s}}^{+\infty} \exp(-y^2/2) dy\ . $$
We compute $dZ_s$ by applying the \Ito Lemma with the following property of the partial derivatives:
\begin{align*}
%    \frac{\partial g}{\partial s}(s,K_s) &= - \frac{K_s}{\sqrt{2\pi(T-s)^3}}\exp\left(-\frac{K_s^2}{2(T-s)} \right)\\
%    \frac{\partial g}{\partial K}(s,K_s) &= - \frac{1}{\sqrt{2\pi(T-s)}}\exp\left(-\frac{K_s^2}{2(T-s)} \right)\\
    \frac{\partial^2 g}{\partial K^2}(s,K_s) &= \frac{2K_s}{\sqrt{2\pi(T-s)^3}}\exp\left(-\frac{K_s^2}{2(T-s)} \right) = 2\frac{\partial g}{\partial s}(s,K_s)\ ,
\end{align*}
so we conclude that
$$ dZ_s = \frac{\partial g}{\partial K}(s,K_s)dK_s\ . $$
In (30.2.55) of page 462 in the monograph \cite{PeskirShiryaev2006}, the term $dK_s$ is computed as follows
$$ dK_s = sgn(W_s)\bOne_{\{W_s\neq 0\}}dW_s + dL_s^0 \ ,$$
being $L^0=\prT{L^0}$ the local time at zero of the Brownian motion.
According to Theorem~4.13 of~\cite{BENDER03}, $L^0_t$ takes the expression
\begin{equation}\label{eq.SDE.localtime}
    dL^0_t = \left(\frac{1}{\sqrt{2\pi t}} - \int_0^t \frac{W_s}{\sqrt{2\pi(t-s)^3}}\exp\left(-\frac{W_s^2}{2(t-s)}\right) dW_s\right)dt
\end{equation}
in $(\cS)^*$.
%Applying again the \Ito Lemma to the stochastic integral of \eqref{eq.SDE.localtime}, we rewrite it as
%    $$ \frac{dL^0_t}{dt} =\frac{1}{\sqrt{2\pi t}} +\frac{1}{\sqrt{2\pi(T-t)}}\exp\left(-\frac{W_t^2}{2(T-t)}\right)\ . $$
    As $d\langle K,K\rangle_t = dt$ $\PP$-a.s., we conclude that
    \begin{equation*}
        dZ_s =\frac{\partial g}{\partial K}(s,K_s)   \frac{dL^0_s}{ds}ds +  \frac{\partial g}{\partial K}(s,K_s)  sgn(W_s)\bOne_{\{W_s\neq 0\}}dW_s
    \end{equation*}
%Note that $dK_s = -dW_s$ holds true under the condition $K_s > 0$ and by Comment 4.1.7.3 of \cite{Jeanblanc2009mathematical} we assure that $\PP(K_s>0) = 1$, for any $s>0$.
%$$dZ_s = \left(\frac{\partial g}{\partial s}(s,K_s) + \frac{1}{2}\frac{\partial^2 g}{\partial K^2}(s,K_s)\right)ds -  \frac{\partial g}{\partial K}(s,K_s) sgn(W_s)\bOne_{\{W_s\neq 0\}} dW_s  \ ,$$
and, using the martingale property, 
we can reformulate $\eqref{ex.eq.compensator.Z}$ as
\begin{equation*}
    \EE\left[\int_t^{t+h} \frac{\kappa_s} {1+\pi_{s}^{\bF}\kappa_s} dZ_s \vert\cF_t\right] = 
    \EE\left[\int_t^{t+h} \frac{\kappa_s} {1+\pi_{s}^{\bF}\kappa_s} \frac{\partial g}{\partial K}(s,K_s) \frac{dL^0_s }{ds} ds
    \vert\cF_t\right]\ .
 %   \int_t^{t+h}  \EE\left[\frac{\kappa_s} {1+\pi_{s}^{\bF}\kappa_s} \left(\frac{\partial g}{\partial s}(s,K_s) + \frac{1}{2}\frac{\partial^2 g}{\partial K^2}(s,K_s) (\phi(s,K_s) -1)^2  \right) \vert\cF_t\right] ds
\end{equation*}
%which is absolutely continuous with respect to the Lebesgue measure. % and the integrand function is null at the origin $s= t$.
Finally the before default strategy is
\begin{align}\label{max.log.before}
    \pi^\bF_s =&\frac{\mu^\bF_s-\rho^\bF_s}{\left(\sigma^\bF_s\right)^2}-  \frac{1}{\left(\sigma^\bF_s\right)^2}\int_{\bR_0}\frac{\pi^\bF_s\left(\theta^\bF_s(z)\right)^2}{1+\pi^\bF_s\theta^\bF_s(z)} \nu(dz)+\frac{1}{\sigma^\bF_s }\frac{ \exp\left(-\frac{K_s^2}{2(T-s)}\right)}{Z_s\sqrt{T-s}}\notag\\
    &-\frac{1}{\left(\sigma^\bF_s\right)^2}\frac{\kappa_s} {1+\pi_{s}^{\bF}\kappa_s} \frac{\partial g}{\partial K}(s,K_s)  \frac{dL^0_s }{ds}\ ,\quad 0 < s \leq T\ .
\end{align}
We achieved an implicit formula for the local maximum portfolio under the progressive enlargement given by the time in which the Brownian motion reaches its maximum and
we remind that this progressive enlargement does not satisfy neither the {density} nor the {intensity} hypothesis.
However, note that $\frac{dL^0_s}{ds}$ is not well-defined in $L^2(dt\times\PP)$ and we have gotten a strategy which is not in $\cA(\bG)$, 
because it not satisfies the condition~\eqref{hyp1}.
\end{Example}
In the following example, we show a random time satisfying the intensity hypothesis \eqref{approach.intensity} but not the density one \eqref{approach.density}.
In this case, we get the condition characterizing a local maximum strategy.
\begin{Example}\label{ex.aksamit.half-final}
We consider the following example with
\begin{equation}\label{ex.tau.half-final}
    \tau := \sup_{t\in [0,T]}\Big\{ W_t = \frac{W_T}{2} \Big\}\ ,
\end{equation}
deeply studied in \cite{aksamit:tel-01016672,Jeanblanc2009mathematical}.
The process $Z$ is computed in Subsection 5.6.5 of~\cite{Jeanblanc2009mathematical} as
\begin{equation}\label{def.Z.ex.Aksamit}
    Z_t = 1 - h\left(\frac{\abs{W_t}}{\sqrt{T-t}}\right)\ ,\quad h(x):=\sqrt{\frac{2}{\pi}}\int_0^x y^2 e^{-y^2/2}dy\ .
\end{equation}
By the decomposition given by Proposition~1.46 of \cite{Aksamit2017}, we conclude that 
$A^\tau_t = h\left(\dfrac{\abs{W_t}}{\sqrt{T-t}}\right)$
and, by~\cite{Jeanblanc2009mathematical}, the \Ito Lemma is computed with the following result,
$$ dA^\tau_t = \frac{\abs{W_t}\exp\left(-\frac{W_t^2}{2(T-t)}\right)}{(T-t)^{3/2}}  dt\ . $$
Then we conclude that \eqref{approach.intensity} holds true. 
Moreover, after Theorem~4.13 of \cite{aksamit:tel-01016672}, it is mentioned that the hypothesis $(\cH')$ fails, so we have found an example satisfying \eqref{approach.intensity} but not \eqref{approach.density}.
In order to determine the optimality condition \eqref{gen.opt.condition.log.bef} before default, we need to compute the Malliavin derivative.
\begin{Lemma}
The Malliavin trace of $\bOne_{\{\tau > s\}}$ is as follows,
    \begin{align}\label{ex.malliavin.derivative.aksamit}
    D_{s^+}\bOne_{\{\tau > s\}}
    =& -\delta_0(2m_{s,T}-W_T) \bOne_{\{W_T \geq 0 \}} - \delta_{0}(W_T)\bOne_{\{2m_{s,T} \geq W_T\}}\notag\\
    &+\delta_0(2M_{s,T}-W_T) \bOne_{\{W_T \leq 0 \}} + \delta_{0}(W_T) \bOne_{\{2M_{s,T} \leq W_T\}}\ .
\end{align}
Moreover, its conditional expectation is computed as 
\begin{align}\label{ex.malliavin.derivative.cond.exp.Aksamit}
     \EE\left[ D_{s^+}\bOne_{\{\tau>s\}} |\cF_s \right] =& 
    \sqrt{\frac{2}{\pi}}\left( -sgn(W_s)
    \frac{W^2_s}{\sqrt{(T-s)^3}}\exp\left(-\frac{W_s^2}{2(T-s)}\right)\right.\notag\\
    &+\left.
    \frac{\abs{W_s}}{\sqrt{(T-s)^3}}\exp\left(-\frac{W_s^2}{2(T-s)}\right) \right)
     %\bOne_{\{W_s\geq 0\}}\left(\int_0^{W_s}f^m_{s,T}(m,2m)dm \,+f^m_{s,T}(0,0)\right)\notag\\
    %&+ \bOne_{\{W_s\leq 0\}}\left(\int_{W_s}^0f^M_{s,T}(m,2m)dm \,+f^M_{s,T}(0,0)\right) 
\end{align}
\end{Lemma}
\begin{proof}
Note that, by using \cite{Jeanblanc2009mathematical}, we can rewrite $\bOne_{\{\tau > s\}}$ as
\begin{align*}
    \bOne_{\{\tau > s\}} = 1- \bOne_{\{\tau \leq s\}} = 1-\bOne_{\{2m_{s,T} \geq W_T\}}\bOne_{\{W_T \geq 0 \}} - \bOne_{\{2M_{s,T} \leq W_T \}}\bOne_{\{W_T \leq 0 \}}\ .
\end{align*}
Then we use the chain rule and we get,
\begin{align*}
   D_{s^+}\bOne_{\{\tau > s\}} =& - D_{s^+}\left(\bOne_{\{2m_{s,T} \geq W_T\}}\bOne_{\{W_T \geq 0 \}}\right) - D_{s^+}\left(\bOne_{\{2M_{s,T} \leq W_T \}}\bOne_{\{W_T \leq 0 \}}\right)\\
    =& -D_{s^+}\left(\bOne_{\{2m_{s,T} \geq W_T\}}\right) \bOne_{\{W_T \geq 0 \}} - D_{s^+}\left(\bOne_{\{W_T \geq 0 \}}\right) \bOne_{\{2m_{s,T} \geq W_T\}}\\
    &-D_{s^+}\left(\bOne_{\{2M_{s,T} \leq W_T\}}\right) \bOne_{\{W_T \leq 0 \}} - D_{s^+}\left(\bOne_{\{W_T \leq 0 \}}\right) \bOne_{\{2M_{s,T} \leq W_T\}}\ ,
\end{align*}
where the main computation is
\begin{align*}
    D_{s^+}\bOne_{\{2M_{s,T} \leq W_T \}} &= D_{s^+}\bOne_{\{2M_{s,T}-W_T\leq 0 \}} =  -\delta_0(2M_{s,T}-W_T) D_{s^+}\left(2M_{s,T}-W_T\right)\\
    &= -\delta_0(2M_{s,T}-W_T)\left(2\bOne_{\{M_{s,T}> W_s\}}-1\right)
\end{align*}
and $D_{s^+}\bOne_{\{W_T \geq 0 \}} = \delta_{0}(W_T)$ for any $s< T$.
We get
\begin{align*}
    D_{s^+}\bOne_{\{\tau > s\}}
    =& -\delta_0(2m_{s,T}-W_T)\left(2\bOne_{\{m_{s,T}\leq W_s\}}-1\right) \bOne_{\{W_T \geq 0 \}} - \delta_{0}(W_T)\bOne_{\{2m_{s,T} \geq W_T\}}\\
    &+\delta_0(2M_{s,T}-W_T)\left(2\bOne_{\{M_{s,T} \geq W_s\}}-1\right) \bOne_{\{W_T \leq 0 \}} + \delta_{0}(W_T) \bOne_{\{2M_{s,T} \leq W_T\}}\\
    =& -\delta_0(2m_{s,T}-W_T) \bOne_{\{W_T \geq 0 \}} - \delta_{0}(W_T)\bOne_{\{2m_{s,T} \geq W_T\}}\\
    &+\delta_0(2M_{s,T}-W_T)\bOne_{\{W_T \leq 0 \}} + \delta_{0}(W_T) \bOne_{\{2M_{s,T} \leq W_T\}}\ .
\end{align*}
We finally compute the conditional expectation,
We denote by $f^M_{s,T}(m,w)$ %with $(m,w)\in\bR_+\times\bR$ 
the density of $(M_{s,T},W_T|W_s)$ 
and by $f^m_{s,T}(m,w)$ %with $(m,w)\in\bR_-\times\bR$ 
the density of $(m_{s,T},W_T|W_s)$.
\begin{align*}
    \EE[D_{s^+}\bOne_{\{\tau > s\}}\vert\cF_s] %=& %-\int_{-\infty}^{W_s}\int_m^{+\infty} \delta_0(2m-w)2\bOne_{\{m < W_s\}} \bOne_{\{w \geq 0 \}} f^m_{s,T}(m,w) dwdm\\
     %&+\int_{-\infty}^{W_s}\int_m^{+\infty} \delta_0(2m-w) \bOne_{\{w \geq 0 \}} f^m_{s,T}(m,w) dwdm\\
    %&-\int_{-\infty}^{W_s}\int_m^{+\infty} \delta_{0}(w)\bOne_{\{2m \geq w\}}f^m_{s,T}(m,w) dwdm\\
    %&+\int_{W_s}^{\infty}\int_{-\infty}^{m}\delta_0(2m-w)2\bOne_{\{m > W_s\}} \bOne_{\{w \leq 0 \}}f^M_{s,T}(m,w)dwdm\\
    %&-\int_{W_s}^{\infty}\int_{-\infty}^{m}\delta_0(2m-w)\bOne_{\{w \leq 0 \}}f^M_{s,T}(m,w)dwdm\\
    %&+\int_{W_s}^{\infty}\int_{-\infty}^{m} \delta_{0}(w) \bOne_{\{2m \leq w\}}f^M_{s,T}(m,w)dwdm\\
    =& -\int_{-\infty}^{W_s}\int_{m\vee 0}^{+\infty} \delta_0(2m-w)f^m_{s,T}(m,w) dwdm\\
    &-\int_{-\infty}^{W_s}\int_m^{2m} \delta_{0}(w) f^m_{s,T}(m,w) dwdm\\
    &+\int_{W_s}^{\infty}\int_{-\infty}^{m\wedge 0}\delta_0(2m-w)f^M_{s,T}(m,w)dwdm\\
    &+\int_{W_s}^{\infty}\int_{2m}^{m} \delta_{0}(w) f^M_{s,T}(m,w)dwdm\\
    =& \bOne_{\{W_s\geq 0\}}\left(\int_0^{W_s}f^m_{s,T}(m,2m)dm \,+f^m_{s,T}(0,0)\right)\\
    &+ \bOne_{\{W_s\leq 0\}}\left(\int_{W_s}^0f^M_{s,T}(m,2m)dm \,+f^M_{s,T}(0,0)\right)\ .
\end{align*}
    Finally, we aim to give a more explicit expression by computing the densities $f^M_{s,T}$ and $f^m_{s,T}$. We get,
    \begin{align*}
        f^M_{s,T}(x,y) =& \PP(M_{s,T}\in dx, W_T\in dy\vert\cF_s)\\ 
        &= \PP(M_{s,T}-W_s\in dx-W_s, W_T-W_s\in dy-W_s\vert\cF_s)\\
        &= \PP(M_{T-s}\in dx-W_s, W_{T-s}\in dy-W_s)\\
        &= \sqrt{\frac{2}{\pi}}\frac{2x-y-W_s}{\sqrt{(T-s)^3}}\exp\left(-\frac{(2x-y-W_s)^2}{2(T-s)}\right)\ ,\quad x\geq W_s\ ,\  x\geq y\ ,
    \end{align*}
    and we conclude
    \begin{align*}
        \bOne_{\{W_s\leq 0\}}\EE[D_{s^+}\bOne_{\{\tau > s\}}\vert\cF_s] = 
    \bOne_{\{W_s\leq 0\}}\sqrt{\frac{2}{\pi}}&\left( 
    \frac{W^2_s}{\sqrt{(T-s)^3}}\exp\left(-\frac{W_s^2}{2(T-s)}\right)\right.\\
    &-\left.
    \frac{W_s}{\sqrt{(T-s)^3}}\exp\left(-\frac{W_s^2}{2(T-s)}\right) \right)\ ,
    \end{align*}
    and by reasoning analogously with $f^m_{s,T}$ the result follows.
\end{proof}
Then, the optimal strategy in $\cA(\bF)$ exists if and only if the the following equation has a solution,
\begin{align}
    \pi^\bF_s\left(\sigma^\bF_s\right)^2 =& \mu^\bF_s-\rho^\bF_s -  \int_{\bR_0}\frac{\pi^\bF_s(\theta^\bF_s(z))^2}{1+\pi^\bF_s\theta^\bF_s(z)}\nu(dz) +\frac{\EE\left[D_{s^+}\bOne_{\{\tau > s\}}\,|\cF_s\,\right]}{Z_{s}}\sigma^\bF_s\notag\\
    & +  \frac{\kappa_s}{1+\pi^{\bF}_s\kappa_s}\frac{1}{Z_{s}} \frac{\abs{W_s}\exp\left(-\frac{W_s^2}{2(T-s)}\right)}{(T-s)^{3/2}} \ ,
\end{align}
where $Z$ satisfies \eqref{def.Z.ex.Aksamit} and $\EE[D_s\bOne_{\{\tau>s\}}\vert\cF_s]$ satisfies \eqref{ex.malliavin.derivative.cond.exp.Aksamit}.
\begin{Remark}
    In particular, if
    the market coefficient of the Poisson process is $\theta=0$, 
    we have the following equation for the local maximum $\pi^\bF$,
\begin{equation}\label{eq.opt.condition.log.bef.aksamit}
    \pi^\bF_s= \frac{\mu^\bF_s-\rho^\bF_s}{\left(\sigma^\bF_s\right)^2 }+\frac{\EE\left[D_{s^+}\bOne_{\{\tau > s\}}\,|\cF_s\,\right]}{Z_{s}\sigma^\bF_s} +  \frac{\kappa_s}{1+\pi^{\bF}_s\kappa_s}\frac{1}{Z_{s}\left(\sigma^\bF_s\right)^2} \frac{\abs{W_s}\exp\left(-\frac{W_s^2}{2(T-s)}\right)}{(T-s)^{3/2}} \ .
\end{equation}
To short the notation, we define
$$ a_s:= \frac{\mu^\bF_s-\rho^\bF_s}{\left(\sigma^\bF_s\right)^2 } +\frac{\EE\left[D_{s^+} \bOne_{\{\tau > s\}}\,|\cF_s\,\right]} {Z_{s}\sigma^\bF_s} \ ,\quad b_s:= \frac{1}{Z_{s}\left(\sigma^\bF_s\right)^2} \frac{\abs{W_t}\exp\left(-\frac{W_t^2}{2(T-t)}\right)}{(T-t)^{3/2}}\ , $$
and \eqref{eq.opt.condition.log.bef.aksamit} is reduced to
$$ \pi^\bF_s= a_s +  \frac{\kappa_s}{1+\pi^{\bF}_s\kappa_s}b_s \Longrightarrow 
\kappa_s\left(\pi^\bF_s\right)^2 + (1-a_s\kappa_s)\pi^\bF_s  -(a_s+b_s\kappa_s) = 0\ , $$
and we solved it as
\begin{align*}
    \pi^\bF_s &= \frac{a_s\kappa_s-1+\sqrt{(1-a_s\kappa_s)^2+4\kappa_s(a_s+b_s\kappa_s)}}{2\kappa_s}\\
    &= \frac{a_s\kappa_s-1+\sqrt{(1+a_s\kappa_s)^2+4b_s\kappa^2_s}}{2\kappa_s}\ ,
\end{align*}
where it can be checked that the positive solution is the only admissible that satisfies \eqref{hyp3}.
\end{Remark}
\end{Example}

\section{Sufficient Approach}\label{sec:suff.aproach}
In Section \ref{sec:opt.problem}, we showed
that the existence of a local maximum for the logarithmic utility optimization problem is strictly related with the intensity hypothesis \eqref{approach.intensity}.
In this section, we assume the existence of such local maximum.

In order to conclude that the optimal problem is well-posed, we introduce some conditions in the following Assumption \ref{assumption.unif.integral}.
The first one is obvious while the second one is needed to properly define the new measure $\QQ$ in Theorem \ref{thm.suff.condition}.
The third and the fourth ones are needed to compute local maximum strategies by allowing the differentiation under the integral sign.
We define the following process,
\begin{align}\label{def.psi}
    \Psi_t(y,\beta,\pi) :=& \int_0^t\beta_s\left(\mu_s-\rho_s-(\pi_s+y\beta_s)\sigma^2_s - \int_{\bR_0}\frac{(\pi_s+y\beta_s)\theta^2_s(z)}{1+(\pi_s+y\beta_s)\theta_s(z)} \nu(dz)\right)ds\notag\\
    &+ \int_0^t\beta_s\sigma_sd^-W_s+\int_0^t\int_{\bR_0}\frac{\beta_s\theta_s(z)}{1+(\pi_s+y\beta_s)\theta_s(z)}\tilde N(d^-z,ds) \notag\\
    &+ \int_0^t\frac{\beta_s\kappa_s\,dH_s }{1+(\pi_s+y\beta_s)\kappa_s}
\end{align}
assuming that there exists some $\delta > 0$,
which may depend on $\pi\in\cA(\bE)$,
such that $y\in(-\delta,\delta)$, $\pi+y\beta\in\cA(\bE)$ for $\bE\in\{\bF,\bG\}$ and $\beta\in\cA(\bE)$ is bounded.
\begin{Assumption}\label{assumption.unif.integral}
\begin{enumerate}
    For any $\pi\in\cA(\bE)$,
    \item $\EE[U(X^{\pi}_T)] <+\infty.$
    \item $0<U'(X^{\pi}_T)X^{\pi}_T <+\infty\  \PP-$a.s.
    \item The following family of processes is uniformly integrable, $$\{U'(X^{\pi+y\beta}_T)X^{\pi+y\beta}_T \abs{\Psi_T(y,\beta,\pi)}\}_{y\in(-\delta,\delta)}\ .$$ 
    \item The following family of processes is uniformly integrable,
    \begin{align*}
    \Big\{ &U^{\prime \prime}\left(X^{\pi+y \beta}_T\right) (X_T^{\pi+y \beta})^{2} \Psi^{2}_T(y, \beta, \pi) \\
    &+\,U^{\prime}\left(X^{\pi+y \beta}_T\right) X^{\pi+y \beta}_T\left(\Psi_T(y, \beta, \pi)+\frac{d}{dy}\Psi_T(y, \beta, \pi)\right)\Big\}_{y \in(-\delta, \delta)}\ .
    \end{align*}
\end{enumerate}
\end{Assumption}
\begin{Remark}
    If we consider the particular case with 
    $0\leq\pi_t\leq 1,\,\forall t\in[0,T]$, 
    under logarithmic utility and bounded market coefficients, 
    it can be shown that 
    Assumption~\ref{assumption.unif.integral} holds true.
\end{Remark}
The following lemma is quite technical but we need it in order to simplify the following computations.
\begin{Lemma}\label{lemma.derivatives}
Let $\pi,\beta\in\cA(\bE)$ be some admissible strategies with $\beta$ bounded, and define the function
$$ g(y) := \EE[U(X^{\pi+y\beta}_T)]\ . $$
Then under Assumption~\ref{assumption.unif.integral} the first and second derivatives of the function g can be computed as follows
\begin{align*}
    g'(y) &= \frac{d}{dy}\EE\left[ U( X^{\pi+y\beta}_T) \right] = \EE[ U'( X^{\pi+y\beta}_T) X^{\pi+y\beta}_T\,\Psi_T(y,\beta,\pi)] \\
    g''(y)&= \EE\left[ X^{\pi+y \beta}_T \Psi^{2}_T(y, \beta, \pi)\left(U^{\prime \prime}\left(X^{\pi+y \beta}_T\right) X^{\pi+y \beta}_T+U^{\prime}\left(X^{\pi+y \beta}_T\right)\right)\right.\\
    &\quad \left.+U^{\prime}\left(X^{\pi+y \beta}_T\right) X^{\pi+y \beta}_T \frac{d}{dy}\Psi_T(y, \beta, \pi)\right]\ ,
\end{align*}
where $\Psi$ is defined in \eqref{def.psi}.
\end{Lemma}
\begin{proof}
The result directly follows from the conditions given in Assumption~\ref{assumption.unif.integral}.
\end{proof}
\begin{Theorem}\label{thm.suff.condition}
%Suppose that $\mu.\sigma,\theta$ and $\lambda$ are $\bG$-adapted processes and random fields, assumption $A_{u.i.}$ holds and $\cF_t\bigvee\cH_t^{\Lambda}\subset\cG_t$ for all $t\in[0,T]$.
If $\pi$ is a local maximum for the problem \eqref{optimization.problem}, with $\bE\in\{\bF,\bG\}$, and Assumption~\ref{assumption.unif.integral} holds, 
then the process $M^{\pi}=\prT{M^{\pi}}$, defined as,
\begin{align}\label{M.martingale}
    M^{\pi}_t:=& \Psi_t(0,1,\pi) = \int_0^{t} \left( \mu_s-\rho_s-\pi_s\sigma^2_s - \int_{\bR_0}\frac{\pi_s\theta^2_s(z)}{1+\pi_s\theta_s(z)}\nu(dz) \right) ds\notag  \\
    &+\int_0^{t}\sigma_sd^-W_s +\int_0^{t}\int_{\bR_0}\frac{\theta_s(z)}{1+\pi_s\theta_s(z)}\tilde N(d^-s,dz) + \int_0^{t} \frac{\kappa_s}{1+\pi_s\kappa_s} dH_s
\end{align}
    has the martingale property under $(\QQ,\bE)$ with
    \begin{equation}\label{Q.measure}
        d\QQ := F^{\pi}_T d\PP,\ \  F^{\pi}_T := \frac{U'(X^{\pi}_T) X^{\pi}_T}{\EE[U'(X^{\pi}_T)X^{\pi}_T]}\ .
    \end{equation}
\end{Theorem}
\begin{proof}
    The measure $\QQ$ in \eqref{Q.measure} is well-defined thanks to Assumption~\ref{assumption.unif.integral}(2).
    If $\pi$ is a local maximum, for all bounded $\beta\in\cA(\bE)$ we have, using the notation of Lemma~\ref{lemma.derivatives}, $g'(0) = 0$, i.e.,
    $$ 0 = \EE[ U'(X^{\pi}_T) X^{\pi}_T\,\Psi_T(0,\beta,\pi)]\ . $$
    Let $\beta_s = \xi_t\bOne_{\{t < s \leq t+h\}}$, where $\xi_t$ is a bounded $\cE_t$-measurable random variable, then, 
        \begin{align}\label{eq.local.opt.aux}
        0 &= \EE\left[\xi_t F^{\pi}_T \Big\{ \int_t^{t+h} \left( \mu_s-\rho_s-\pi_s\sigma^2_s\,- \int_{\bR_0} \frac{\pi_s\theta^2_s(z)} {1+\pi_s\theta_s(z)} \nu(dz)\right) ds \right.\notag\\
        &\left. +\int_t^{t+h} \sigma_sd^-W_s+\int_t^{t+h}\int_{\bR_0}\frac{\theta_s(z)}{1+\pi_s\theta_s(z)}\tilde N(d^-z,ds) + \int_t^{t+h} \frac{\kappa_s} {1+\pi_s\kappa_s}dH_s\Big\}\right]\ .
        \end{align}
    As \eqref{eq.local.opt.aux} holds true for every $\cE_t$-measurable bounded random variable, we conclude
    $$ 0 = \EE[F^{\pi}_T(M^{\pi}_{t+h} - M^{\pi}_t)|\cE_t] = \EE_{\QQ}[M^{\pi}_{t+h} - M^{\pi}_t|\cE_t]\ . $$
    Moreover, by Assumption \ref{assumption.unif.integral}(3) we conclude that $$\EE_{\QQ}[\abs{M_T^{\pi}}] = \EE[F_T^\pi\abs{\Psi_t(0,1,\pi)}] < +\infty\ ,$$
    and the result for every $t\in[0,T]$ follows by applying the Jensen inequality,
    \begin{equation*}
        \EE_{\QQ}[\abs{M_t^\pi}] = \EE_{\QQ}[\abs{\EE_{\QQ}[M_T^{\pi}\vert \cE_t]}]\leq \EE_{\QQ}[\EE_{\QQ}[\abs{M_T^{\pi}}\vert \cE_t]] = \EE_{\QQ}[\abs{M_T^{\pi}}]< +\infty\ .
    \end{equation*}
    \end{proof}
    In particular, if we consider the optimization problem with $\bE = \bG$, then the process $M^{\pi}$ is a martingale, because every term included is $\bG$-adapted.
    \begin{Remark}\label{P.decomposition.martingale}
    Let $F^{\pi}_t := \EE\left[\frac{d\QQ}{d\PP}|\cG_t\right]$ and $Z^{\pi}_t:=\EE\left[\frac{d\PP}{d\QQ}|\cG_t\right]$, 
    by the Girsanov Theorem we know that if $\pi$ is a local maximum, then the process $F^{\pi}M^{\pi}$ is a $(\PP,\bG)$-martingale.
    Moreover, according to Theorem 37 of \cite{Protter2005},
    the process
    $$M^{\pi}_t-\int_0^t\frac{1}{Z^{\pi}_{s-}}d\langle M^{\pi}_s,Z^{\pi}_s\rangle\ ,\quad 0\leq t\leq T\ ,$$
    is also a $(\PP,\bG)$-martingale, 
    where $\langle M^{\pi},Z^{\pi}\rangle$ is the unique $\bG$-predictable and increasing process such that $M^{\pi}Z^{\pi}-\langle M^{\pi},Z^{\pi}\rangle$ is a $(\PP,\bG)$-martingale.
    \end{Remark}
    \begin{Remark}
        Theorem \ref{thm.suff.condition} extends the results of Theorem 3.3 of \cite{diNunno2014} as it considers the interval $[[\tau\wedge T,T]]$.
        However in \cite{diNunno2014} also the converse is proved, that is, 
        under the concavity assumption on the function $g$,  if $M^{\pi}$ has the martingale property under $(\QQ,\bE)$, then $\pi$ is a local maximum for the problem~\eqref{optimization.problem}.
    \end{Remark}
   \begin{Proposition}\label{prop.integrals.semimartingales}
   If $\pi$ is a local maximum for the problem \eqref{optimization.problem} with the information flow $\bG$ and Assumption~\ref{assumption.unif.integral} holds, 
   then the following processes
   $$ \int_0^{t}\sigma_sd^-W_s \quad \text{and} \quad \int_0^{t}\int_{\bR_0}\frac{\theta_s(z)}{1+\pi_s\theta_s(z)}\tilde N(d^-s,dz)\ ,\ 0\leq t \leq T\ $$
   are $(\QQ,\bG)$-semimartingales.
   \end{Proposition} 
    \begin{proof}
    In the definition of the process $M^\pi$ given by \eqref{M.martingale}, we can rewrite the default term by applying Lemma~\ref{H.decomposition} 
    as the sum of a $(\QQ,\bG)$-local martingale and a bounded variation predictable process,
    \begin{equation*}
        \int_0^t \frac{\kappa_s}{1+\pi_s\kappa_s}dH_s =  \int_0^t \frac{\kappa_s}{1+\pi_s\kappa_s}d\widetilde J_s +\int_0^t \frac{\kappa_s}{1+\pi_s\kappa_s}\frac{d\langle J_s,F_s^\pi\rangle}{F^\pi_{s-}} %\\
        %+\int_0^t \frac{\kappa_s}{1+\pi_s\kappa_s}\frac{d\langle J_s,F_s^\pi\rangle}{F^\pi_s}
        + \int_0^t\frac{\kappa_s}{1+\pi_s\kappa_s}\frac{\bOne_{\{\tau > s\}}}{Z_{s-}}dA^\tau_s\ ,
    \end{equation*}
    where $\widetilde{J}_s = J_s - \dfrac{d\langle J_s,F_s^\pi\rangle}{F^\pi_{s-}}.$
%   Using the process defined in Lemma \ref{H.decomposition},        we consider the following  $(\QQ,\bG)$-martingale.
%    \begin{align}\label{process.C.localmartingale}
%    C_t&:=M^{\pi}_t-\int_0^{t}\frac{\kappa_s}{1+\kappa_s\pi_s}dJ_s\notag\\
%    &= \int_0^{t}\left( \mu_s-\rho_s-\pi_s\sigma^2_s -\int_{\bR_0}\frac{\pi_s\theta^2_s(z)}{1+\pi_s\theta_s(z)}\nu(dz) \right)ds + \int_0^{t}\sigma_sd^-W_s\notag\\
%    &+\int_0^{t}\int_{\bR_0}\frac{\theta_s(z)}{1+\pi_s\theta_s(z)}\tilde N(d^-s,dz)+  \int_0^{t}\frac{\kappa_s}{1+\pi_s\kappa_s}\frac{dA^{\tau}_s}{\PP(\tau > s|\cF_s)}\ .
%    \end{align}
We conclude that the sum process
\begin{align*}
    \int_0^{t}\sigma_sd^-W_s &+ \int_0^{t}\int_{\bR_0} \frac{\theta_s(z)}{1+\pi_s\theta_s(z)}\tilde N(d^-s,dz) = M_t^\pi - \int_0^t \frac{\kappa_s}{1+\pi_s\kappa_s}d\widetilde J_s \\
&-\int_0^{t} \left( \mu_s-\rho_s-\pi_s\sigma^2_s - \int_{\bR_0}\frac{\pi_s\theta^2_s(z)}{1+\pi_s\theta_s(z)}\nu(dz) \right) ds\\
&-\int_0^t \frac{\kappa_s}{1+\pi_s\kappa_s}\frac{d\langle J_s,F_s^\pi\rangle}{F^\pi_{s-}} - \int_0^t\frac{\kappa_s}{1+\pi_s\kappa_s}\frac{\bOne_{\{\tau > s\}}}{Z_{s-}}dA^\tau_s
\end{align*}
is a special $(\QQ,\bG)$-semimartingale because is the sum of a  $(\QQ,\bG)$-local martingale and a predictable bounded variation process. 
%Using Theorem 26.14 of \cite{kallenberg2002} on 
By using the the semimartingale decomposition in the continuous and pure discontinuous part, we conclude that there exist $\alpha^{(i)}$, $\gamma^{(i)}$ with $i\in\{1,2,3\}$ such that
\begin{subequations}\label{aux.comp}
    \begin{align}
     \widehat M_t=&\int_0^t \sigma_sd^-W_s + \int_0^t  \sigma_s\alpha^{(1)}_s ds + \int_0^t  \sigma_s\alpha_s^{(2)} \frac{d\langle J_s,F_s^\pi\rangle}{F_{s-}^\pi} + \int_0^{t\wedge\tau} \sigma_s \alpha_s^{(3)} \frac{dA^\tau_s}{Z_{s-}} \label{aux.comp.W}\\
     \widehat J_t=&\int_0^{t}\int_{\bR_0} \frac{\theta_s(z)}{1+\pi_s\theta_s(z)}\tilde N(d^-s,dz) + \int_0^t  \gamma_s^{(1)}ds + \int_0^t  \gamma_s^{(2)} \frac{d\langle J_s,F_s^\pi\rangle}{F_{s-}^\pi} + \int_0^{t\wedge\tau} \gamma_s^{(3)} \frac{dA^\tau_s}{Z_{s-}} \label{aux.comp.N}
\end{align}
\end{subequations}
are $(\QQ,\bG)$-local martingales and $\alpha^{(i)},\gamma^{(i)}$ are the unique $\bG$-adapted processes satisfying 
\begin{align}
    &\int_0^t  \left(\sigma_s\alpha_s^{(1)} +  \gamma_s^{(1)} \right)ds + \int_0^t \left( \sigma_s\alpha_s^{(2)} +  \gamma_s^{(2)} \right)\frac{d\langle J_s,F_s^\pi\rangle}{F_{s-}^\pi} + \int_0^{t\wedge \tau} \left(\sigma_s \alpha_s^{(3)} 
    +  \gamma_s^{(3)} \right) \frac{dA^\tau_s}{Z_{s-}}\notag \\
    =& \int_0^{t} \left( \mu_s-\rho_s-\pi_s\sigma^2_s - \int_{\bR_0}\frac{\pi_s\theta^2_s(z)}{1+\pi_s\theta_s(z)}\nu(dz) \right) ds +\int_0^t \frac{\kappa_s}{1+\pi_s\kappa_s}\left(\frac{d\langle J_s,F_s^\pi\rangle}{F^\pi_{s-}} + \frac{\bOne_{\{\tau > s\}}}{Z_{s-}}dA^\tau_s\right) \ .\label{optimality.condition.first-one}
    %&+\int_0^t \frac{\kappa_s}{1+\pi_s\kappa_s}\frac{d\langle J_s,F_s^\pi\rangle}{F^\pi_{s-}} + \int_0^t\frac{\kappa_s}{1+\pi_s\kappa_s}\frac{\bOne_{\{\tau > s\}}}{Z_{s-}}dA^\tau_s\ .\label{optimality.condition.first-one}
\end{align}
\end{proof}
Proposition~\ref{prop.integrals.semimartingales} can be extended to $W$ and $\tilde N$, 
as shown in the following corollary. 
\begin{Corollary}\label{Q.local.martingale}
   If $\pi$ is a local maximum for the problem \eqref{optimization.problem} with the information flow $\bG$ and Assumption~\ref{assumption.unif.integral} holds, 
   then the following processes 
    $$ W_t\quad \text{and}\quad \int_0^{t}\int_{\bR_0}z\ \tilde N(ds,dz)\ ,\ 0\leq t \leq T $$
    are $(\QQ,\bG)$-semimartingales.
\end{Corollary}
\begin{proof}
The reasoning is analogous to Theorem 15 of \cite{DiNunnoOksendalProske2006}.
We consider the $(\QQ,\bG)$-local martingale $\widehat M$ defined in \eqref{aux.comp.W}, so the result of a It\^o integral is still a local martingale.
In particular,
$$\int_0^t\frac{1}{\sigma_s} d\widehat M_s =W_t  + \int_0^t \alpha_s^{(1)}ds + \int_0^t \alpha_s^{(2)} \frac{d\langle J_s,F_s^\pi\rangle}{F_{s-}^\pi} + \int_0^{t\wedge\tau} \alpha_s^{(3)} \frac{dA^\tau_s}{Z_{s-}} \ .$$
The same argument works with the compensated Poisson random measure.
\end{proof}
\begin{Remark}\label{rk.comp.cont}
    Using Proposition 4.4 of \cite{diTella2021}, 
    we know that the compensated Poisson measure is quasi-left continuous on $\{\tau>s\}$, so its compensator is continuous before default. 
    From this fact we conclude that 
    \begin{align*}
        0 =&  \bOne_{\{\tau > s\}}  \alpha_s^{(2)} \frac{\Delta\langle J_s,F_s^\pi\rangle}{F_{s-}^\pi} 
        + \bOne_{\{\tau > s\}} \alpha_s^{(3)} \frac{\Delta A^\tau_s}{Z_{s-}} \\
        0 =& \bOne_{\{\tau > s\}}   \gamma_s^{(2)} \frac{\Delta\langle J_s,F_s^\pi\rangle}{F_{s-}^\pi} 
        + \bOne_{\{\tau > s\}} \gamma_s^{(3)} \frac{\Delta A^\tau_s}{Z_{s-}} \ .
    \end{align*}
\end{Remark}
Using Theorem 1.8 of \cite{JacodShiryaev03} and the fact that $\cF_t\subset\cG_t$, we know that there exists a unique predictable compensator $\nu^{\bG}(dt,dz)$ for the Poisson random measure $N(dt,dz)$. 
However, Corollary~\ref{Q.local.martingale} is stronger because we claim that this compensator is of finite variation. In the following, to short the notation, we will use $\nu^{\bG}(dt,dz)$ in order to refer it.
In the next corollary we rewrite the above processes under the measure $\PP$.
\begin{Corollary}\label{W.N.decomposition}
    If $\pi$ is a local maximum for the problem \eqref{optimization.problem} with the information flow $\bG$ and Assumption~\ref{assumption.unif.integral} holds, 
    then the following processes
    \begin{align*}
    \widehat W_t:= W_t-A^W_t \quad\text{and}\quad
    \int_0^t\int_{\bR_0}z\tilde N(ds,dz) - \int_0^t\frac{1}{Z^\pi_s}\langle \int_0^s\int_{\bR_0}z\tilde N(du,dz), Z^{\pi}_{s-}\rangle\ ,\quad 0\leq t \leq T\ ,
    \end{align*}
    are a $(\PP,\bG)$-local martingale and semimartingale, respectively,
    where $A^W$ is a $\bG$-continuous predictable process. 
\end{Corollary}
\begin{proof}
It follows by Corollary~\ref{Q.local.martingale} and the Girsanov theorem for local-martingales, see Theorem~35 on page~132 and its corollaries in \cite{Protter2005} for details.
\end{proof}
Once we have proved that, 
under Assumption~\ref{assumption.unif.integral} and the existence of a local maximum of the optimization problem,
the enlargement of filtration $\bG\supset\bF$ preserves the semimartingale property for $W$ and $\tilde N$ with respect to the measure $\PP$, 
then we can apply Lemmas~\ref{forward.brownian.semimart} and~\ref{forward.poisson.semimart} in order to conclude that the involved forward integrals reduce to the classical It\^o integrals for semimartingale processes.
We summarize the result in the next theorem in which the risky asset $S$ is driven by a $(\PP,\bG)$-local martingale and a \hbox{$(\PP,\bG)$-semimartingale.}
\begin{Theorem}\label{theo.semimarting.W.N}
    Suppose $\pi$ is a local maximum for the utility problem \eqref{optimization.problem} with the information flow $\bG$, then, the asset's dynamics satisfy the following SDEs under $(\PP,\bG)$.
\begin{subequations}
\begin{align}
    \frac{dD_t}{D_t} =&\rho_tdt\label{def.assets1.semimart}\\
    \frac{dS_t}{S_{t-}} =&\mu_tdt +\sigma_tdA^W_t + \int_{\bR_0}\theta_t(z) \left(\nu^{\bG}-\nu\right)(dt,dz)+\sigma_td\widehat W_t + \int_{\bR_0}\theta_t(z)\widehat N(dt,dz) +\kappa_tdH_t \label{def.assets2.semimart}
\end{align}
\end{subequations}
where the integrals are well-defined in the classical It\^o sense because $\widehat W$ and $\widehat N:=N-\nu^{\bG}$ are 
%$(\PP,\bG)$-local martingale and 
$(\PP,\bG)$-semimartigales. %, respectively.
Moreover, we rewrite the dynamics of $X^{\pi}$ as follows
\begin{align}\label{G.dynamic.X}
    \frac{dX^{\pi}_t}{X^{\pi}_{t-}} =& (1-\pi_t)\rho_tdt + \pi_t\left( \mu_tdt  + \sigma_tdA^W_t + \int_{\bR_0}\theta_t(z) \left(\nu^{\bG}-\nu\right)(dt,dz)  \right)\notag\\
    &+\pi_t\left( \sigma_td\widehat W_t +  \int_{\bR_0}\theta_t(z)\widehat N(dt,dz) +\kappa_tdH_t\right)\ .
%    =& \left[\rho_t + \pi_t(\mu(t)-\rho_t)\right]dt +  \pi_t\sigma_td^-W_t + \int_{\bR_0} \pi_t\theta_t(z)\tilde N(d^-t,dz)\\
%    &+  \pi_t\kappa_tdH_t \ .
\end{align}
\end{Theorem}
\begin{Remark}
    Note that for $U(x) = \ln(x)$ the measure $\QQ$ coincides with~$\PP$, 
    then a similar version of Theorem~\ref{theo.semimarting.W.N} can be achieved only by assuming the existence of a local maximum for the logarithmic utility problem as it was motivated in Subsections \ref{subsec:after.default.log} and \ref{subsec:before.default.log}.
\end{Remark}
In particular, the following statement solves the optimal portfolio problem by using the compensators instead of the Malliavin derivatives. 
We decompose the compensators as the absolutely continuous, the singular and the purely discontinuous parts as follows,
\begin{align}
    dA^\tau_s &= \lambda_sds +  dA^{\tau,\perp}_s +\Delta A^\tau_s\label{A.tau.decomp}\\
    d\langle J_s,F^{\pi}_s\rangle &= d\langle J_s,F^{\pi}_s\rangle^{ac}+d\langle J_s,F^{\pi}_s\rangle^\perp +\Delta\langle J_s,F^{\pi}_s\rangle \label{bracket.decomp}\ .
\end{align}
%As the jumps of $N$ are totally inaccessible, we conclude that $\nu^\bG$ is continuous. 
%while the continuity of $A^\tau$ and $A$ has been argued in Theorem~\ref{theo.opt.log.bef} and Corollary~\ref{W.N.decomposition}.
%In the last decompositions
%we indicate the possibility of jumps.
\begin{Theorem}
    Suppose $\pi$ is a local maximum for the optimization problem~\eqref{optimization.problem}, then it satisfies the following equations,
    \begin{align}
   0=& \left( \sigma_s\alpha_s^{(2)} +  \gamma_s^{(2)}  - \frac{\kappa_s}{1+\pi_s\kappa_s}\right)\frac{d\langle J_s,F_s^\pi\rangle^{ac}}{F^\pi_{s-}} + \left(\sigma_s \alpha_s^{(3)} +  \gamma_s^{(3)} - \frac{\kappa_s}{1+\pi_s\kappa_s}\right) \frac{\bOne_{\{\tau > s\}}}{Z_{s-}}\lambda_s\notag \\
    &-  \left(\mu_s-\rho_s-\pi_s\sigma^2_s - \int_{\bR_0}\frac{\pi_s\theta^2_s(z)}{1+\pi_s\theta_s(z)}\nu(dz)\right) +\sigma_s\alpha_s^{(1)} + \gamma_s^{(1)}  \label{gen.opt.pi.ac}\\
    0=& \left( \sigma_s\alpha_s^{(2)} +  \gamma_s^{(2)} - \frac{\kappa_s}{1+\pi_s\kappa_s} \right)\frac{d\langle J_s,F_s^\pi\rangle^\perp}{F_{s-}^\pi}  
    + \left(\sigma_s\alpha_s^{(3)} +  \gamma_s^{(3)} - \frac{\kappa_s}{1+\pi_s\kappa_s}\right) \frac{\bOne_{\{\tau > s\}}}{Z_{s-}} dA^{\tau,\perp}_s .\label{gen.opt.pi.sing}\\
    0=&  \frac{\Delta\langle J_s,F_s^\pi\rangle}{F_{s-}^\pi} + \frac{\bOne_{\{\tau > s\}}}{Z_{s-}}\Delta A^{\tau}_s\ .\label{gen.opt.pi.jump}
%    0=&  \textcolor{red}{\left( \sigma_s\alpha_s^{(2)} +  \gamma_s^{(2)} - \frac{\kappa_s}{1+\pi_s\kappa_s} \right)\frac{\Delta\langle J_s,F_s^\pi\rangle}{F_{s-}^\pi}}\notag \\
%    &+ \textcolor{red}{\left(\sigma_s \alpha_s^{(3)} +  \gamma_s^{(3)} - \frac{\kappa_s}{1+\pi_s\kappa_s}\right) \frac{\bOne_{\{\tau > s\}}}{Z_{s-}} \Delta A^{\tau}_s }.\label{gen.opt.pi.jump}
%    0 =& \mu_s-\rho_s-\pi_s\sigma^2_s\,- \int_{\bR_0} \frac{\pi_s\theta^2_s(z)} {1+\pi_s\theta_s(z)} \nu(dz) - \frac{\gamma_s}{Z^\pi_{s}}+ \sigma_s\widehat\alpha_s \notag\\ 
%    & + \int_{\bR_0}\frac{\theta_s(z)}{1+\pi_s\theta_s(z)}\nu_s^\bG(dz) + \frac{\bOne_{\{\tau > s\}}}{Z_{s-}}\frac{\kappa_s} {1+\pi_s\kappa_s}\lambda_s   \ .\label{gen.opt.pi.ac}\\
%    0 =& -\frac{1}{Z^{\pi}_s}d\langle M^{\pi}_s,Z^{\pi}_s\rangle^\perp + \sigma_sdA^\perp_s  +\int_{\bR_0}\frac{\theta_s(z)}{1+\pi_s\theta_s(z)}\nu^{\bG,\perp}(dt,dz) \notag\\ 
%    &  +\frac{\bOne_{\{\tau > s\}}}{Z_{s-}}\frac{\kappa_s} {1+\pi_s\kappa_s}dA^{\tau,\perp}_s \ .\label{gen.opt.pi.sing}\\
%    0 =& \Delta\langle M^{\pi}_s,Z^{\pi}_s\rangle \ . \label{gen.opt.pi.jump}
\end{align}
    Moreover, if \eqref{gen.opt.pi.ac}, \eqref{gen.opt.pi.sing} and \eqref{gen.opt.pi.jump}
    hold true,
    where $\alpha^{(i)},\gamma^{(i)}$ satisfy \eqref{aux.comp},
    and  
    $$xU''(x) + U'(x) \leq 0\ ,\quad\forall x >0 $$ 
    then $\pi$ is a local maximum for the optimization problem \eqref{optimization.problem}.
\end{Theorem}
\begin{proof}
It follows by~\eqref{optimality.condition.first-one} and by splitting in absolutely continuous and singular parts via~\eqref{A.tau.decomp} and~\eqref{bracket.decomp}.
The condition~\eqref{gen.opt.pi.jump} is simplified using Remark~\ref{rk.comp.cont}.

%By Theorem \ref{thm.suff.condition} and its corollaries, if $\pi$ is a local maximum, the following process
%\begin{align*}
%    Y_t :=& M_t^\pi-\int_0^t\frac{1}{Z^{\pi}_s}d\langle M^{\pi}_s,Z^{\pi}_s\rangle 
%    -\int_0^t\sigma_sd\widehat W_s 
%    - \int_{\bR_0}\frac{\pi_s\theta^2_s(z)}{1+\pi_s\theta_s(z)}\widehat N(dt,dz)\\
%    &- \int_0^t \frac{\kappa_s} {1+\pi_s\kappa_s}d\left(H_s - \int_0^{s\wedge\tau}\frac{dA^{\tau}_u}{Z_{u-}}\right)\ ,\quad 0\leq t\leq T\ ,
%\end{align*}
%is a $(\bG,\PP)$-local martingale. 
%Using the definition of $M^\pi$, we can rewrite it as 
%\begin{align*}
%    Y_t =& \int_0^t \left( \mu_s-\rho_s-\pi_s\sigma^2_s\,- \int_{\bR_0} \frac{\pi_s\theta^2_s(z)} {1+\pi_s\theta_s(z)} \nu(dz)\right) ds -\int_0^t\frac{1}{Z^{\pi}_s}d\langle M^{\pi}_s,Z^{\pi}_s\rangle \\ &+\int_0^t \sigma_sdA_s
%    +\int_0^{t}\int_{\bR_0}\frac{\theta_s(z)}{1+\pi_s\theta_s(z)}\nu^\bG(dt,dz)+\int_0^t \frac{\bOne_{\{\tau > s\}}}{Z_{s-}}\frac{\kappa_s} {1+\pi_s\kappa_s}dA^{\tau}_s \ ,
%\end{align*}
%and we find that $Y$ is a local martingale and a finite variation process, then we conclude that $Y_t = 0$.
%In particular, we equal to zero each part of the decomposition, i.e.,
%the absolutely continuous, the singular and the purely discontinuous parts. 
%We get the conditions the conditions \eqref{gen.opt.pi.ac} and \eqref{gen.opt.pi.sing} and the continuity of $\langle M^{\pi}_s,Z^{\pi}_s\rangle$.
Moreover, by assuming the conditions \eqref{gen.opt.pi.ac}, \eqref{gen.opt.pi.sing} and \eqref{gen.opt.pi.jump}, 
we recover the first order condition for optimality and we need to evaluate the second derivative. 
Note that $ \frac{d}{dy}\Psi_T(y, \beta, \pi)$ is computed as the right-hand side of the following equation
\begin{align*}
  0 >&  -\int_{0}^{T} \beta^{2}_s \sigma^{2}_s d s-\int_{0}^{T} \int_{\mathbb{R}_{0}} \frac{\beta^{2}_s \theta^2_s( z)}{(1+(\pi_s+y \beta_s) \theta_s(z))^2} \nu(dz)ds \\
&-\int_{0}^{T} \int_{\mathbb{R}_{0}} \frac{\beta^{2}_s \theta^2_s( z)}{(1+(\pi_s+y \beta_s) \theta_s(z))^2} \tilde N\left(d s, d z\right)-\int_{0}^{T} \frac{\beta^{2}_s \kappa^2_s}{(1+(\pi_s+y \beta_s) \kappa_s)^2} d H_s\ .
\end{align*}
Therefore, by Lemma~\ref{lemma.derivatives}
\begin{align*}
    g^{\prime \prime}(0)=&\EE\left[U^{\prime \prime}\left(X^{\pi}_T\right) (X_T^{\pi})^2\Psi^2_T(0,\beta,\pi)\right]\\
    &+\EE\left[U^{\prime}\left(X^{\pi}_T\right) X^{\pi}_T\left(\Psi^2_T(0,\beta,\pi) + \frac{d}{dy}\Psi_T(0,\beta,\pi)\right)\right]< 0\ ,
\end{align*}
for all bounded $\beta \in \mathcal{A}(\bG)$ if
$x U^{\prime \prime}(x)+U^{\prime}(x)\leq 0,$ and $x>0$\ .
\end{proof}
Note that the utility function $U$ appears in \eqref{gen.opt.pi.ac},~\eqref{gen.opt.pi.sing} and~\eqref{gen.opt.pi.jump} 
by the processes~$\alpha^{(i)},\gamma^{(i)}$, 
because they are induced by the change of measure from $\QQ$, which is anticipating by the definition given in~\eqref{Q.measure}.
In the next result we show how \eqref{gen.opt.pi.ac},~\eqref{gen.opt.pi.sing} 
and~\eqref{gen.opt.pi.jump} are simplified in case of $U(x)=\ln x$ recovering the intensity hypothesis as in Theorem~\ref{theo.opt.log.bef}.
\begin{Corollary}
        There exists a local maximum portfolio $\pi$ for the optimization problem~\eqref{optimization.problem}, with $U(x)=\ln x$, 
        if and only if
        it satisfies the condition
    \begin{align}%\label{gen.opt.pi.log.ac}
   0 =&  \left(\sigma_s \alpha_s^{(3)} +  \gamma_s^{(3)} - \frac{\kappa_s}{1+\pi_s\kappa_s}\right) \frac{\bOne_{\{\tau > s\}}}{Z_{s-}}\lambda_s\notag \\
    &-\left(  \mu_s-\rho_s-\pi_s\sigma^2_s - \int_{\bR_0}\frac{\pi_s\theta^2_s(z)}{1+\pi_s\theta_s(z)}\nu(dz) \right)+\sigma_s\alpha_s^{(1)} + \gamma_s^{(1)}\ .  \label{gen.opt.pi.ac.log}
\end{align}
    and the intensity hypothesis \eqref{approach.intensity} holds true on $\tau$, with intensity $\lambda$.
\end{Corollary}
\begin{proof}
Under logarithmic utility, the process 
$M^{\pi,\bF}_t := \EE[M^\pi_t\vert\cF_t]$, with $0\leq t \leq T$, is a $(\PP,\bF)$-martingale. 
%Under logarithmic utility, the process $M^\pi$ is a $(\PP,\bG)$-martingale. 
We define
\begin{align*}
    Y_t =& \int_0^{t} \EE\left[\mu_s-\rho_s-\pi_s\sigma^2_s - \int_{\bR_0}\frac{\pi_s\theta^2_s(z)}{1+\pi_s\theta_s(z)}\nu(dz) \vert \cF_t\right] ds
    +\int_0^{t} \EE\left[D_{s^+}\sigma_s\vert\cF_t\right] ds \notag  \\
    &
    +\int_0^{t}\int_{\bR_0} \EE\left[ D_{s^+,z}\frac{\theta_s(z)}{1+\pi_s\theta_s(z)}\vert\cF_t\right] \nu(dz)ds
    + \int_0^{t} \EE\left[\frac{\kappa_s}{1+\pi_s\kappa_s}\vert\cF_t\right] dA^\tau_s 
    \ ,\  0\leq t \leq T\ ,
\end{align*}
which is also a $(\PP,\bF)$-martingale, because
$$ \EE[Y_{t+h}-Y_t\vert\cF_t] = \EE[M^{\pi,\bF}_{t+h}-M^{\pi,\bF}_t\vert\cF_t] = 0$$
where we have used Theorems \ref{theo.forw.skor.mall.BM} and \ref{theo.forw.sko.mall.pois}.
$Y$ is a finite variation process, so we conclude that it is a null process.
From the decomposition in its absolutely continuous, singular and pure jump parts we conclude that
$$0 = \int_0^t  \frac{\kappa_s}{1+\pi_s\kappa_s} dA^{\tau,\perp}_s +  \sum_{0\leq s \leq t}  \frac{\kappa_s}{1+\pi_s\kappa_s} \Delta A^{\tau}_s \ ,\quad 0\leq t \leq T\ ,$$
and then $A^{\tau,\perp}=\Delta A=0$ because $\kappa_s\neq 0$.
We finally simplify the conditions 
\eqref{gen.opt.pi.ac},~\eqref{gen.opt.pi.sing} 
and~\eqref{gen.opt.pi.jump}
into \eqref{gen.opt.pi.ac.log}.
\begin{comment}
\textcolor{red}{
 The conditions \eqref{gen.opt.pi.ac}, \eqref{gen.opt.pi.sing} and \eqref{gen.opt.pi.jump} are simplified because $\QQ = \PP$ and $d\langle J_s,F^\pi_s\rangle = 0$. 
 Then \eqref{gen.opt.pi.sing} and \eqref{gen.opt.pi.jump} are written as
 \begin{align*}
     0 &=\left(\sigma_s \alpha^3_s +  \gamma^3_s - \frac{\kappa_s}{1+\pi_s\kappa_s}\right) \frac{\bOne_{\{\tau > s\}}}{Z_{s-}} dA^{\tau,\perp}_s\\
     0 &=\left(\sigma_s \alpha^3_s +  \gamma^3_s - \frac{\kappa_s}{1+\pi_s\kappa_s}\right) \frac{\bOne_{\{\tau > s\}}}{Z_{s-}} \Delta A^{\tau}_s
 \end{align*}
    If we assume that
   $$ \sigma_s \alpha^3_s +  \gamma^3_s - \frac{\kappa_s}{1+\pi_s\kappa_s} = 0\ ,\text{ on } \{\omega:\Delta A^{\tau}_s>0\}\cup\{\omega:dA^{\tau,\perp}_s\neq 0\}$$
   then the solution of $\pi$ is explicit as 
   $$ \pi_s = \frac{1}{\sigma_s\alpha_s^3+\gamma_s^3}-\frac{1}{\kappa_s} $$
   which is a contradiction with the equation \eqref{gen.opt.pi.ac.log}.
   Then we conclude that $dA^{\tau,\perp}_s = \Delta A^{\tau}_s = 0$ and the random time $\tau$ satisfies the intensity approach.}
\end{comment}
\end{proof}
\begin{Remark}
We can achieve a more explicit result under the assumption of Hunt processes.
Indeed, if the risky asset $S$ is a Hunt process, the agent plays with the price filtration~$\bS$ and~$\tau$ is a totally inaccessible $\bS$-stopping time.
By applying \cite{Protter10} we conclude that the default process $H$ has an absolutely continuous compensator $\lambda$ with respect to the Lebesgue measure.
In particular, the necessary and sufficient condition for a local maximum strategy is 
$$     0= \mu_s-\rho_s-\pi_s\sigma^2_s -\int_{\bR_0}\frac{\pi_s\theta^2_s(z)}{1+\pi_s\theta_s(z)}\nu(dz) +\frac{\kappa_s}{1+\pi_s\kappa_s} \frac{\bOne_{ \{\tau > s\}} }{Z_{s-}} \lambda_s\ . $$
\end{Remark}
%\subsection{Examples}
In the following example we consider the same random time as in Example~\ref{ex.aksamit.half-final}.
In this case, we apply the sufficient approach developed in Section~\ref{sec:suff.aproach} and we show how this is connected with the computations in terms of Malliavin derivative of Section~\ref{sec:opt.problem}.
\begin{Example}[Example \ref{ex.aksamit.half-final} revisited]\label{ex.aksamit.rev}
We consider the random time 
$\tau = \sup_t\{ W_t = W_T/2 \}$.
In Theorem 4.13 of \cite{aksamit:tel-01016672}, 
the compensator of $W$ in $\bG$ is computed as the right-hand side of the following equation,
\begin{align*}
    \alpha_s^{(1)} + \alpha_s^{(3)}\frac{dA^\tau_s}{Z_s} =& -\bOne_{\{\tau>s\}}\sqrt{\frac{2}{\pi}}\frac{sgn(W_s)}{Z_s}\frac{W_s^2}{\sqrt{(T-s)^3}}\exp\left(-\frac{W_s^2}{2(T-s)}\right)\\
    &+\bOne_{\{\tau\leq s\}}\left(\frac{W_T-W_s}{T-s} - \frac{1}{\varphi_s}\frac{W_T}{(T-s)}\right)\ ,
\end{align*}
where $\varphi_s=1-\exp\left(\frac{2W_sW_T-W_T^2}{2(T-s)}\right)$.
By assuming that there is no Poisson jumps in the market, i.e., $\theta=0$, 
the optimality condition under logarithmic utility is given by
\begin{equation}\label{eq.opt.aksamit}
    0 =  \left(\sigma_s \alpha_s^{(3)} - \frac{\kappa_s}{1+\pi_s\kappa_s}\right) \frac{\bOne_{\{\tau > s\}}}{Z_{s-}}dA^\tau_s
-(  \mu_s-\rho_s)+\pi_s\sigma^2_s+\sigma_s\alpha_s^{(1)}\ .
\end{equation}
Note that, the condition before default is the same that we got in \eqref{eq.opt.condition.log.bef.aksamit}, so both approaches match.
The optimality condition after default is reduced to 
$$    0 = 
-(  \mu_s-\rho_s)+\pi_s\sigma^2_s+\sigma_s\left(\frac{W_T-W_s}{T-s} - \frac{1}{\varphi_s}\frac{W_T}{(T-s)}\right)\ .$$
\end{Example}
\begin{Remark}\label{Remark.Honest}
    Under the assumption of the existence of a local maximum, we derive the existence of the compensators for $W$ and $\tilde N$ and an explicit condition for this local maximum.
    In \cite{Imkeller02} it is mentioned that whenever $\tau$ is an Honest time, the process 
    $$ \widehat W_t := W_t - \int_0^{\tau\wedge t} \frac{\frac{d}{ds}\langle \widehat A_s,W_s \rangle}{Z_{s-}} ds + \int^T_{\tau\wedge t} \frac{\frac{d}{ds}\langle \widehat A_s,W_s \rangle}{1-Z_{s-}} ds \ ,$$
    is a $(\PP,\bG)$-Brownian motion, 
    where $\widehat A$ is the unique $\bG$-predictable 
    and increasing process such that $Z_{t}+\widehat A_t$ is a $(\PP,\bG)$-martingale. 
    A similar result is obtained in \cite{ElKarouiJeanblancJiao2009} for those random times which satisfy the density hypothesis.
\end{Remark}
\section{Conclusion}\label{sec:conclusion}
In this paper, we extend the solution of the classical optimal portfolio problem to the case when default happens and it is not a stopping time in the natural filtration of the structural components of the processes modeling the risky asset. 
In case the default happens, the price of the risky asset will jump,  but the agent is still able to invest in it.
In this work we succeed in identifying the intensity hypothesis~\eqref{approach.intensity} as the essential one for the existence of a local maximum strategy under logarithmic utility 
and enlarged filtration $\bG$.
In particular, we drop the density assumption so-called, generally known as the Jacod's hypothesis.
The logarithmic utility problem is solved by splitting it into two subproblems, namely before and after default occurrence.
We also deal with a general utility, and for that we assume the existence of a local maximum and
we deduce the semimartingale decomposition of the Brownian motion and the compensated Poisson random measure in the information flow $\bG$ and the conditions that such local maximum must satisfy.
We include many examples to show how to apply the methodology in specific cases. The most relevant one is the analysis of the time at which the Brownian motion reaches its maximum, see Example~\ref{ex.argmax}. At the best of our knowledge this is the first time this random time has been used for default occurrence.
In our arguments we apply the interplay between the forward integral, the Skorohod integral, and the Malliavin derivative, by developing the interrelationships among these operators in the white noise framework. This extension is needed to overcome the limitations of the domain of the operators in the classical setup and this result has an independent mathematical value. 
We have used this to compute the expressions of the semimartingale decomposition of the noises in the enlarged filtration.
The last Example~\ref{ex.aksamit.rev} shows the agreement of the different techniques we have developed in this work. 

%% file: short-appendix-BM.tex
\section[Anticipating calculus in the white noise framework]{Anticipating calculus in the white noise framework: Brownian motion}\label{App.Brownian}
\begin{Lemma}\label{forward.brownian.semimart}
Let $\bG$ be a given filtration such that $\bF^{W}\subset\bG$ holds true. 
Suppose that~$W$ is a $\bG$-semimartingale. 
Let $\varphi$ be a $\bG$-predictable process and the integral
$ \int_{0}^{T} \varphi_t d W_t $
exists as a classical It\^o integral.
Then $\varphi$ is forward integrable with respect to $W$ and
$$\int_{0}^{T} \varphi_t d^{-} W_t=\int_{0}^{T} \varphi_t d W_t\ .$$
\end{Lemma}
\begin{proof}
We refer to the proof of Lemma 8.9 in \cite{DiNunnoOksendalProske2009}.
\end{proof}
\begin{proof}
[$\mathbf{Proof\, of\, Lemma\ \ref{lem.wick.ord.mall}}$]
\hypertarget{proof.lem.wick.ord.mall}
We consider the representation of $F$ and $W_t-W_s$ with respect to the orthogonal basis $\{H_\alpha\}$
$$ F = \sum_{\alpha\in\cJ}a_{\alpha}H_{\alpha}\ ,\quad W_t-W_s = \sum_{k=1}^\infty \left( \int_s^t e_k(y)dy \right) H_{\varepsilon(k)}\ . $$
%where the second one can be consulted in Example 5.3 of \cite{DiNunnoOksendalProske2009}.
By definition of Wick product we get,
\begin{align*}
    F \diamond (W_t-W_s) =& \sum_{\alpha,k} a_\alpha \left( \int_s^t e_k(y) dy \right) H_{\alpha+\varepsilon(k)} \\
    =& \sum_{\alpha,k}a_\alpha \left( \int_s^t e_k(y) dy \right) h_{\alpha_k + 1}(\theta_k) \prod_{j\neq k} h_{\alpha_j}(\theta_j) \\
    =& \sum_{\alpha,k}a_\alpha \left( \int_s^t e_k(y) dy \right) \left( \theta_k h_{\alpha_k}(\theta_k) - \alpha_k h_{\alpha_k-1}(\theta_k) \right) \prod_{j\neq k} h_{\alpha_j}(\theta_j) \\
    =& \sum_{\alpha,k}a_\alpha \left( \int_s^t e_k(y) dy \right) \theta_k  H_{\alpha} - \sum_{\alpha,k}a_\alpha \alpha_k \left( \int_s^t e_k(y) dy \right) h_{\alpha_k-1}(\theta_k)\prod_{j\neq k} h_{\alpha_j}(\theta_j)\\
    =& \sum_{\alpha,k}a_\alpha \left( \int_s^t e_k(y) dy \right) H_{\varepsilon(k)}  H_{\alpha} - \sum_{\alpha,k}a_\alpha \alpha_k \left( \int_s^t e_k(y) dy \right) H_{\alpha-\varepsilon(k)}\\
    =& F\cdot (W_t- W_s) -\int_s^t D_y F dy\ ,
\end{align*}
where we have used Definition~\ref{def.malliavin.S*}, 
equation \eqref{def.chaos.H}, the fact
\hbox{$H_{\varepsilon(k)} = h_1(\theta_h) = \theta_k$} and the following recurrence relationship of the Hermite polynomials,
$$ h_{n+1}(x) = xh_n(x) - nh_{n-1}(x)\ , $$
with $n=\alpha_k$ and $x=\theta_k$.
\end{proof}
\begin{proof}[$\mathbf{Proof\,of\,Theorem\,\ref{theo.forw.skor.mall.BM}}$]
\hypertarget{proof.theo.forw.skor.mall.BM}
We follow the lines of Theorem 8.18 in \cite{DiNunnoOksendalProske2009}, 
by considering the convergence in the space $(\cS)^*$ via Lemma \ref{lem.wick.ord.mall}, and Theorem~3.7 in \cite{BiaginiOksendalForward} and allowing $D_{t^+}Y_t\in(\cS)^*$ for any $t$.
\begin{align*}
    \int_0^T Y_t d^-W_t &= \lim_{\epsilon\to 0}\int_0^T  Y_t\frac{W_{t+\epsilon}-W_t}{\epsilon} dt\\ 
    &= \lim_{\epsilon\to 0} \frac{1}{\epsilon} \int_0^T \left( Y_t \diamond (W_{t+\epsilon}-W_t) + \int_t^{t+\epsilon} D_s Y_t ds \right) dt \\
    &= \lim_{\epsilon\to 0} \frac{1}{\epsilon} \int_0^T \left( \int_t^{t+\epsilon} Y_t \diamond \bW_s ds + \int_t^{t+\epsilon} D_s Y_t ds \right) dt \\
    &= \lim_{\epsilon\to 0} \frac{1}{\epsilon} \left( \int_0^T \left( \int_{s-\epsilon}^{s}  Y_t dt\right)   \diamond \bW_s ds + \int_0^T\int_{s-\epsilon}^s D_{s} Y_t dt  ds \right)\\
    &= \int_0^T Y_s \diamond \bW_s ds + \int_0^T D_{u^+} Y_u du\ ,
\end{align*}
where in the third equality we applied Example~2.5.13 of \cite{oksendal2010book} in order to conclude
$$ Y_t \diamond (W_{t+\epsilon} - W_t) = \int_t^{t+\epsilon} Y_t \diamond\bW_s ds\ . $$
For a reference of the Fubini Theorem in $(\cS)^*$, see for example Theorem A.9 in~\cite{JUNG20143846}. 
In the fourth equality, 
as $Y\in L^2(dt\times\PP,\bG)$
we applied the following convergence in $L^2(\PP)\subset (\cS)^*$,
$$Y_s^\epsilon:= \frac{1}{\epsilon} \int_{s-\epsilon}^s Y_t dt \to Y_s \,\text{ if }\,\epsilon\to 0 \ ,$$
 see the proof of Theorem 8.18 in \cite{DiNunnoOksendalProske2009}.
In the fifth equality, we first applied Theorem 5.24 of the mentioned reference in order to get
$$ \int_0^T Y_s^\epsilon\diamond\bW_s ds \to \int_0^T Y_s\diamond\bW_s ds \,\text{ if }\,\epsilon\to 0 $$
in $(\cS)^*$.
Then, we use Lemma 6.7 in \cite{DiNunnoOksendalProske2009} to conclude that
$$ \int_0^T 
\left\|\displaystyle D_t \left( \frac{1}{\epsilon}\int_{u-\epsilon}^u Y_s ds \right) \right\|_{-\hat q}^2 dt
%\norm{D_t \left( \frac{1}{\epsilon}\int_{u-\epsilon}^u Y_s ds \right)}_{-\hat q}^2 dt 
\leq 
\left\|\displaystyle  \frac{1}{\epsilon}\int_{u-\epsilon}^u Y_s ds \right\|_{-q}^2
%\norm{\left( \frac{1}{\epsilon}\int_{u-\epsilon}^u Y_s ds \right)}_{-q}^2
\leq \norm{Y_u}^2_{L^2(\PP)}\ ,$$
$\hat q \geq 2q + \frac{1}{\ln 2},\ q\in\bZ^+$ and
$u\in[\epsilon,T]$
and the last convergence is stated in 
$(\cS)_{-\hat q}\subset(\cS)^*$.
\end{proof}
%%%%%%%%%%%%%%%%%%%%%%%%%%%%%%%%%%%%%%%%%%%%%%%%%%%%%%%%%%%%

%% file: short-appendix-Pois.tex
\section[Anticipating calculus in the white noise framework]{Anticipating calculus in the white noise framework: Poisson random measure}\label{App.Poisson}
Let $f\in L^2((dt\times\nu)^n)$, we introduce the following functional
$$ J^N_n(f) := \int_0^T\int_{\bR_0} \ldots \int_0^{t_2}\int_{\bR_0} f(t_1,z_1,\ldots,t_n,z_n) \tilde N(dt_1,dz_1)\ldots \tilde N(dt_n,dz_n)\ .$$
In the case $f$ is symmetric, we define $ I^N_n(f) := n!J^N_n(f).$ 
We consider the symmetrization $sym(f)$ or $\widehat f$ of the deterministic function $f$ as
$$sym(f):= \widehat f(t_1,z_1,\ldots,t_n,z_n) := \frac{1}{n!}\sum_{\sigma\in S_n} f(t_{\sigma(1)},z_{\sigma(1)},\ldots,t_{\sigma(n)},z_{\sigma(n)})\ , $$
and we extend the definition $ I^N_n(f):= n!J^N_n(\widehat{f} )$ for any $f\in L^2((dt\times\nu)^n)$.
We consider $\{l_m\}_m$ the orthogonalization of the family $\{1,z,z^2,...\}$ and let's define
\begin{align}
    p_{j+1}(z) &:=zl_{j}(z) \left(\int_{\bR_0}l_{j}^2(u) u^2\nu(du) \right)^{-1/2} \ ,\quad (z,j)\in \bR_0\times\bZ^+ \label{polinom.p}\\
    \delta_{z(i,j)}(t,z) &:= e_i(t) p_j(z) \ ,\quad (t,z)\in [0,T]\times\bR_0\ ,\quad (i,j)\in\bZ^+\times\bZ^+\ ,  
\end{align}
where $z(i,j):= j + (i+j-2)(i+j-1)/2$ is a bijective map and $\{e_i\}_i$ are the Hermite functions. 
For any $\alpha\in\cJ$ satisfying $\alpha = (\alpha_1,\ldots,\alpha_j)$ with
\hbox{$\abs{\alpha}:=\alpha_1+\ldots+\alpha_j = m$}, we define
$$ \delta^{\widehat\otimes\alpha}((t_1,z_1),\ldots,(t_m,z_m)) := \delta_1^{\widehat\otimes\alpha_1}\widehat\otimes\ldots\widehat\otimes\delta_j^{\widehat\otimes\alpha_j}((t_1,z_1),\ldots,(t_m,z_m))\ , $$
where $\delta^{\widehat\otimes\alpha}$ denotes the symmetrized tensor product. 
Finally we define
\begin{equation}\label{eq.chaos.pois}
    K_{\alpha} := I^N_{\abs{\alpha}} \left( \delta^{\widehat\otimes\alpha}((t_1,z_1),\ldots,(t_m,z_m)) \right)\ .
\end{equation}
The family $\{K_{\alpha}\}_{\alpha\in\cJ}$ is an orthogonal basis for the space $L^2(\PP)$ in the sense of the following theorem.
\begin{Theorem}
Let X be a $\cF_T^N$-measurable random variable in $L^2(\PP)$,
then there exists a unique sequence $\{a_\alpha\}\subset\bR$ such that
$$ X = \sum_{\alpha\in\cJ} a_{\alpha}K_{\alpha}\ . $$
and the norm can be computed as 
$\norm{X}^2_{L^2(\PP)} = \sum\alpha!a_{\alpha}^2$.
\end{Theorem}
%We introduce the notation  $$ \left(2\bN\right)^{\alpha} = \prod_{j=1}^{m}(2j)^{\alpha_j}\ ,\quad \alpha = (\alpha_1,...,\alpha_m)\in\cJ\ . $$
\begin{Definition}
\begin{itemize}
    \item Let $f = \sum_{\alpha\in\cJ}a_{\alpha}K_{\alpha}\in L^2(\PP)$ be a random variable, 
    we say that $f$ belongs to the Hida test function Hilbert space $(\cS)_{k}$, for $k\in\bR$, if 
    $$ \norm{f}^2_{k} := \sum_{\alpha\in\cJ} \alpha !\, a_{\alpha}^2 \prod_{j=1}^{m}(2j)^{k\alpha_j} < \infty\ . $$
    We define the Hida test function space
    $ (\cS) = \bigcap_{k\in\bR}(\cS)_{k} $
    equipped with the projective topology.
    \item Let $F = \sum_{\alpha\in\cJ}a_{\alpha}K_{\alpha}$ be a formal sum, 
    we say that $F$ belongs to the Hida distribution Hilbert space $(\cS)_{-q}$, for $q\in\bR$, if 
    $$ \norm{F}^2_{-q} := \sum_{\alpha\in\cJ} \frac{\alpha !\, a_{\alpha}^2}{\prod_{j=1}^{m}(2j)^{q\alpha_j}} < \infty\ . $$
    We define the Hida distribution space
    $ (\cS)^* = \bigcup_{q\in\bR}(\cS)_{-q} $
    equipped with the inductive topology, i.e, convergence is studied with $\norm{\cdot}_q$ for some $q\in\bR$.
\end{itemize}
\end{Definition}
The generalized expectation and the generalized conditional expectation are defined analogously to the Brownian case.
\begin{Definition}\label{def.white.noise.pois}
The white noise process $\tilde{\bN}(t,z)$ of $\tilde N(dt, dz)$ is defined by the expansion
\begin{equation}
    \tilde{\bN}(t,z):= \sum_{i,j\in\bZ^+} e_i(t)p_j(z) K_{\varepsilon(i,j)}\ ,
\end{equation}
where $\varepsilon(i,j):=\varepsilon(z(i,j))$.
\end{Definition}
As before, it can be proved that $\tilde{\bN}(t,z)$ is well-defined as an object in $(\cS)^*$ and it satisfies
$$ \tilde{\bN}(t,z) = \frac{\tilde N(dt, dz)}{dt\times\nu(dz)}\ . $$
\begin{Definition}\label{def.wick.pois}
Let $F = \sum_{\alpha\in\cJ}a_\alpha K_\alpha,$  $G = \sum_{\beta\in\cJ}b_\beta K_\beta\in (\cS)^*$, then we define the Wick product as
$$ F\diamond G := \sum_{\alpha,\beta\in\cJ}a_\alpha b_\beta K_{\alpha+\beta}\ . $$
\end{Definition}
\begin{Lemma}\label{forward.poisson.semimart}
Let $\bG$ be a given filtration such that $\bF^{N}\subset\bG$. 
Suppose that 
$\int_{0}^{t} \int_{\mathbb{R}_{0}} z \tilde{N}(d s, d z)$, with $0\leq t\leq T$, 
is a semimartingale with respect to $\bG$.
Let $\theta=\theta_t(z)$, with $(t,z)\in [0,T]\times\mathbb{R}_{0}$, be a $\bG$-predictable process and the integral 
$\int_{0}^{T} \int_{\mathbb{R}_{0}} \theta_t(z) \tilde N(d t, d z)$ 
exists as a classical Itô integral.
Then $\theta$ is forward integrable with respect to $\tilde{N}$ and
$$
\int_{0}^{T} \int_{\mathbb{R}_{0}} \theta_t(z) \tilde{N}\left(d^{-} t, d z\right)=\int_{0}^{T} \int_{\mathbb{R}_{0}} \theta_t(z) \tilde{N}(d t, d z)\ .
$$
\end{Lemma}
Before giving the proof of Lemma \ref{lem.wick.ord.mall.pois}, we introduce some notions of the \emph{contraction kernel operation} that it will appear in the next statement. 
See the beginning of the Section 3 of \cite{Tudor2007} for more details.
Let \hbox{$f\in L^2((dt\times\nu)^n)$} and \hbox{$g\in L^2((dt\times\nu))$}, then we define
\begin{align*}
    \left(f\ast_1^0g\right)\left(t_1,z_1,\ldots,t_m,z_m\right) &:=  f\left(t_1,z_1,\ldots,t_m,z_m\right)g(t_1,z_1)\\
    \left(f\ast_1^1g\right)\left(t_1,z_1,\ldots,t_m,z_m\right) &:=  \int_0^T\int_{\bR_0} f\left(t_1,z_1,\ldots,t_m,z_m\right)g(t_1,z_1) \nu(dz_1)dt_1\ .
\end{align*}
Sometimes, if the variables $\left(t_1,z_1,\ldots,t_m,z_m\right)$ are clear we will omit them. 
In the next lemma, we will compute the symmetrization of some kernel operations that we will need.
\begin{Lemma}\label{A.lemm.kernel}
Let $\alpha=(\alpha_1,\ldots,\alpha_n)\in\cJ$ with $\abs{\alpha} = m$ and we consider the functions
$f = \delta^{\widehat\otimes\alpha}$ and 
$g = \bOne_{\{(u,z)\in\Lambda\}}$, with $\Lambda:=[s,t]\times U_m$.
Then, 
\begin{align*}
    sym\left(f\ast_1^0g\right) &= \sum_{i,j} \frac{\alpha_{z(i,j)}}{\abs{\alpha}}\left( \delta^{\widehat\otimes(\alpha-\varepsilon(i,j))}\widehat\otimes\, e_ip_j \bOne_{\Lambda}\right) \\
    sym\left(f\ast_1^1g\right) &= \sum_{i,j} \frac{\alpha_{z(i,j)}}{\abs{\alpha}}\int_\Lambda e_i(t_1)p_j(z_1) \nu(z_1)dt_1\, I_{\abs{\alpha}-1} \left( \delta^{\widehat\otimes(\alpha-\varepsilon(i,j))}\right)
\end{align*}
\end{Lemma}
\begin{proof}
\begin{align*}
    sym\left(f\ast_1^0g\right) &= sym\left( \delta^{\widehat\otimes\alpha} \left(t_1,z_1,\ldots,t_m,z_m\right) \bOne_{\{(t_1,z_1)\in\Lambda\}}\right)\\
    &= \sum_{k=1}^m \frac{1}{\abs{\alpha}} \delta^{\widehat\otimes\alpha} \left(t_1,z_1,\ldots,t_m,z_m\right) \bOne_{\{(t_k,z_k)\in\Lambda\}}\\
    &= \sum_{i,j}\frac{\alpha_{z(i,j)}}{\abs{\alpha}} \left(\delta^{\widehat\otimes\alpha-\varepsilon(i,j)}\left(t_2,z_2,\ldots,t_m,z_m\right)\widehat\otimes\left( e_i(t_1)p_j(z_1)\bOne_{\{(t_1,z_1)\in\Lambda\}}\right)\right)\ ,
\end{align*}
where the first step follows by the definition of the symmetrization and the fact that $\delta^{\widehat\otimes\alpha}$ is indeed a symmetric function.
The second step follows by the definition of $\alpha$ and the fact that each $\delta_{z(i,j)}$ appears $\alpha_{z(i,j)}$ times.
Although we fix $(t_1,z_1)$ in the last expression, by the symmetric property, 
we can write any pair $(t_k,z_k)$.
\begin{align*}
    f\ast_1^1g &=  \int_0^T\int_{\bR_0} \delta^{\widehat\otimes\alpha}\left(t_1,z_1,\ldots,t_m,z_m\right)\bOne_{ \{(t_1,z_1)\in\Lambda \}} \,\nu(dz_1)dt_1\\
     &=  \int_{\Lambda} \delta^{\widehat\otimes\alpha}\left(t_1,z_1,\ldots,t_m,z_m\right) \,\nu(dz_1)dt_1\\
     &= \sum_{i,j}\frac{\alpha_{z(i,j)}}{\abs{\alpha}} \int_{\Lambda}\left( \delta^{\widehat\otimes\alpha-\varepsilon(i,j)}\left(t_2,z_2,\ldots,t_m,z_m\right)
     \otimes e_i(t_1)p_j(z_1) \right)\nu(dz_1)dt_1\\
     &= \sum_{i,j}\frac{\alpha_{z(i,j)}}{\abs{\alpha}} \int_{\Lambda}e_i(t_1) p_j(z_1) \nu(dz_1) dt_1\ I_{\abs{\alpha}-1}\left(\delta^{\widehat\otimes\alpha-\varepsilon(i,j)}\right)\ ,
\end{align*}
where we have developed the symmetrization in the variables $(t_1,z_1)$ in order to compute the integral in the set $\Lambda$.
Note that $f\ast_1^1g = sym(f\ast_1^1g)$ in this particular case.
\end{proof}
\begin{Lemma}\label{lem.wick.ord.mall.pois}
Let $F\in L^2(\PP)$ be a $\cF_T^N$-measurable random variable such that $D_{t,z} F$ is~($\nu\times dt$)-integrable in $(\cS)^*$
and let $U_m\subset\bR_0$ be a compact set, then
\begin{align}
     F \cdot (\tilde N(t,U_m)-\tilde N(s,U_m)) =& \int_s^t \int_{U_m} \left(F+D_{u,z} F\right)\diamond \tilde\bN(u,z)\,\nu(dz)du \notag\\
     &+\int_s^t\int_{U_m} D_{u,z}F\,\nu(dz) du\ .
\end{align}
\end{Lemma}
\begin{proof} %[Proof of Lemma \ref{lem.wick.ord.mall.pois}]
We consider the representation of $F$ with respect to the orthogonal basis $\{K_\alpha\}$ and $\tilde N(t,U_m)-\tilde N(s,U_m)$ as an one-dimensional integral:
$$ F = \sum_{\alpha\in\cJ}a_{\alpha}K_{\alpha} =  \sum_{\alpha\in\cJ}a_{\alpha}I^N_{\abs{\alpha}}(\delta^{\widehat\otimes\alpha}) \ ,\quad 
\tilde N(t,U_m)-\tilde N(s,U_m) =I^N_1\left(\bOne_{\{(u,z)\in\Lambda\}}\right)\ , $$
%where the second one can be consulted in Example 5.3 of \cite{DiNunnoOksendalProske2009}.
where we short the notation by defining $\Lambda:=[s,t]\times U_m$.
We consider also the conditional expectation of the chaos $K_\alpha$ as
$ K_{\alpha}^{(t)} := \EE[K_{\alpha}|\cF_t^N] .$
We compute the following usual product as they are $L^2(\PP)$-random variables:
\begin{align*}
F\cdot (\tilde N(t,U_m)-\tilde N(s,U_m)) &= \sum_{\alpha} a_\alpha K_{\alpha} I^N_1(\bOne_{\{(u,z)\in\Lambda\}})\ .
\end{align*}
We apply the following product formula appearing in Theorem 3.1 of~\cite{Tudor2007} for the Poisson case:
\begin{align*}
  K_{\alpha} I^N_1(\bOne_{\{(u,z)\in\Lambda\}}) =&
I^N_{\abs{\alpha}+1}\left(\delta^{\widehat\otimes\alpha}\widehat\otimes\bOne_{\{(u,z)\in\Lambda\}}\right) +
\abs{\alpha}I^N_{\abs{\alpha}}\left(sym\left(\delta^{\widehat\otimes\alpha}\ast_1^0\bOne_{\{(u,z)\in\Lambda\}}\right)\right)\\ 
&+ \abs{\alpha} I^N_{\abs{\alpha}-1}\left(sym\left(\delta^{\widehat\otimes\alpha}\ast_1^1 \bOne_{\{(u,z)\in\Lambda\}}\right)\right) \ .
\end{align*}
By the definition of Wick product it follows
$$ \sum_{\alpha}a_{\alpha} I^N_{\abs{\alpha}+1}\left(\delta^{\widehat\otimes\alpha}\widehat\otimes\bOne_{\{(u,z)\in\Lambda\}}\right) = F\diamond (\tilde N(t,U_m) - \tilde N(s,U_m)) = \int_\Lambda F\diamond \tilde \bN(u,z) \,\nu(dz)du\ , $$
where we have used the Wick product according to Remark~3.12 of~\cite{DiNunnoOksendalProske2004}.
For the second and the third terms we apply Lemma~\ref{A.lemm.kernel}.
In particular,
\begin{align*}
    &\sum_{\alpha}a_{\alpha}\abs{\alpha}I^N_{\abs{\alpha}}\left(\delta^{\widehat\otimes\alpha} \ast_1^0 \bOne_{\{(u,z)\in\Lambda\}}\right) = \sum_{\alpha,i,j} a_{\alpha} \alpha_{z(i,j)} I^N_{\abs{\alpha}}\left( \delta^{\widehat\otimes(\alpha-\varepsilon(i,j))} \widehat\otimes  \left(e_ip_j\bOne_{\{\cdot\in\Lambda\}}\right)\right) \\
    =& \sum_{\alpha,i,j}a_{\alpha}\alpha_{z(i,j)}\int_{\Lambda} e_i(u)p_j(z) \left(K_{\alpha-\varepsilon(i,j)}- K_{\alpha-\varepsilon(i,j)}^{(u)} +  K_{\alpha-\varepsilon(i,j)}^{(u)}\right) \tilde N(du,dz)\\
    =& \sum_{\alpha,i,j}a_{\alpha}\alpha_{z(i,j)}\int_{\Lambda} e_i(u)p_j(z) \left(K_{\alpha-\varepsilon(i,j)}- K_{\alpha-\varepsilon(i,j)}^{(u)} \right) \diamond \tilde\bN(u,z)\,\nu(dz)du \\
    &+ \sum_{\alpha,i,j}a_{\alpha}\alpha_{z(i,j)}\int_{\Lambda} e_i(u)p_j(z) K_{\alpha-\varepsilon(i,j)}^{(u)}  \diamond \tilde\bN(u,z)\,\nu(dz)du \\
    =&\int_{\Lambda}  \left(D_{u,z}F-\EE[D_{u,z}F|\cF^N_u] \right) \diamond \tilde\bN(u,z)\,\nu(dz)du + \int_{\Lambda} \EE[D_{u,z}F|\cF^N_u]  \diamond \tilde\bN(u,z)\,\nu(dz)du\\
    =& \int_{\Lambda} D_{u,z}F \diamond \tilde\bN(u,z)\,\nu(dz)du
    %&= \sum_{\alpha,i,j}a_{\alpha}\alpha_{z(i,j)}\int_{\Lambda} e_i(u)p_j(z) I_{\abs{\alpha}-1}\left( \delta^{\widehat\otimes(\alpha-\varepsilon(i,j))}(\cdot,u,z)\right) \tilde N(du,dz)\\
    %&= \int_{\Lambda}\EE[D_{u,z}F|\cF^N_u] \diamond \tilde \bN(u,z)\,\nu(dz)du
    \end{align*}
where we have split the term $K_{\alpha-\varepsilon(i,j)}$ in order to get the Skorohod integral.
For the third one, we get
\begin{align*}
    \sum_{\alpha}a_{\alpha} \abs{\alpha} I_{\abs{\alpha}-1}\left(\delta^{\widehat\otimes\alpha} \ast_1^1 \bOne_{\{(u,z)\in\Lambda\}}\right) &= \sum_{\alpha,i,j} a_{\alpha} \alpha_{z(i,j)} \int_\Lambda e_i(u)p_j(z)\nu(dz)du\, I_{\abs{\alpha}-1}\left(  \delta^{\widehat\otimes(\alpha-\varepsilon(i,j))} \right)\\
   &=\int_\Lambda D_{u,z} F\, \nu(dz)du\ .
\end{align*}
\end{proof}